\documentclass[11pt]{amsart}

\usepackage[T1]{fontenc}
\usepackage[utf8]{inputenc}
\usepackage{amsmath,amssymb,amsthm,mathtools}
\usepackage[a4paper,margin=25mm]{geometry}
\usepackage{xcolor}
\usepackage{booktabs}
\usepackage[expansion=false]{microtype}
\usepackage{enumitem}
\usepackage{array}
\usepackage{graphicx}
\usepackage{makecell,multirow}
\usepackage{mathptmx}
\usepackage[hidelinks=false]{hyperref}

\hypersetup{
  colorlinks=true,
  linkcolor=blue!45!black,
  citecolor=blue!45!black,
  urlcolor=blue!55!black,
  pdftitle={Effective atom-dimer bands of three-body lattice systems at strong coupling},
  pdfauthor={M. V. Dolgopolov}
}

\theoremstyle{plain}
\newtheorem{theorem}{Theorem}[section]
\newtheorem{proposition}[theorem]{Proposition}
\newtheorem{lemma}[theorem]{Lemma}
\newtheorem{corollary}[theorem]{Corollary}

\theoremstyle{definition}

\newtheorem{remark}[theorem]{Remark}

\newtheorem{mainprop}{Proposition}

\numberwithin{equation}{section}

\DeclareMathOperator{\Ran}{Ran}
\DeclareMathOperator{\dist}{dist}
\DeclareMathOperator{\spec}{spec}
\DeclareMathOperator{\supp}{supp}

\newcommand{\T}{\mathbb{T}}
\newcommand{\Z}{\mathbb{Z}}
\newcommand{\R}{\mathbb{R}}
\newcommand{\C}{\mathbb{C}}
\newcommand{\bK}{\mathbf K}
\newcommand{\bp}{\mathbf p}
\newcommand{\bq}{\mathbf q}
\newcommand{\bm}{\mathbf m}
\newcommand{\bn}{\mathbf n}
\newcommand{\bx}{\mathbf x}
\newcommand{\by}{\mathbf y}
\newcommand{\be}{\mathbf e}

\newcommand{\bpi}{\boldsymbol\pi}
\newcommand{\eps}{\varepsilon}
\newcommand{\Ltwo}{L^2_{\rm sym}((\T^2)^2)}
\newcommand{\Ell}{\mathrm{K}}
\newcommand{\specess}{\spec_{\rm ess}}
\newcommand{\specdisc}{\spec_{\rm disc}}

\newcommand{\figinclude}[2]{%
  \IfFileExists{#2}{\includegraphics[width=#1]{#2}}%
  {\fbox{\parbox[c][0.28\textheight][c]{0.9\textwidth}{\centering\ttfamily\detokenize{#2}\\(figure file of version~1)}}}}

\newcolumntype{L}[1]{>{\raggedright\arraybackslash}p{#1}}

\begin{document}

\title[Effective atom--dimer bands of three-body lattice systems]
{Effective atom--dimer bands of three-body lattice systems at strong coupling:
exact critical mass ratios and a statistics--quasimomentum duality}

\author{Mikhail V. Dolgopolov}
\address{Samara State Technical University, Samara, Russian Federation}
\email{mvdolg@yandex.ru}

\date{Version 2, \today}



\begin{abstract}
We study three particles on the square lattice $\Z^2$ with contact
interactions of strength $\mu\to\infty$, in every sector of total
quasimomentum $\bK\in\T^2$, for three statistics: three identical
bosons, the $2+1$ system of two identical fermions of mass $m_f$ and a
third particle of mass $m_3$ (ratio $\gamma=m_f/m_3$), and the $2+1$
system of two identical bosons and a third particle. The
\emph{atom--dimer} sector, associated with the infinitely degenerate
eigenvalue $1$ of the pair-contact operator, is analysed on the whole
Brillouin zone. A~Feshbach--Schur reduction gives explicit effective
Hamiltonians on $\ell^2(\Z^2\setminus\{0\})$,
\[
h_B=4+\eps_D-S_{\bK},\qquad
h_F=2+2\gamma+\eps_D+\tfrac{\gamma}{2}S_{\bK},\qquad
h_{2+1B}=2+2\gamma+\eps_D-\tfrac{\gamma}{2}S_{\bK},
\]
with $\eps_D$ the Dirichlet Laplacian and $S_\bK$ a rank-four exchange
operator. We prove a \emph{spectral transfer theorem} with explicit
constants: for $\mu$ above an explicit bound the bottom $\theta_\mu(\bK)$
of the essential spectrum of the three-body fibre operator lies within
$\eta=O(\mu^{-1})$ of $-\mu$ plus the threshold of $h$, and the
eigenvalue counting function of the three-body operator below
$\theta_\mu(\bK)$ is squeezed between those of $-\mu+h$ shifted by $\pm\eta$.
In~particular the number of atom--dimer states bound uniformly in $\mu$
equals the number of eigenvalues of $h$ below its threshold, for every
$\bK$.
The latter number is given, for every $\bK$ and every $\gamma$ including
the critical values, by a $4\times4$ Hermitian matrix built from the
two-point resistances $A=2-\frac4\pi$, $B=\frac2\pi$ of the square
lattice (derived here from first principles), and hence by a quartic
polynomial $P_{\bK}$. This yields: the complete atom--dimer zones in
closed form; the critical hopping (mass) ratio $\gamma_c(\bK)$ on the
whole Brillouin zone, with $1\le\gamma_c(\bK)\le\pi/(\pi-2)$,
$\gamma_c(\bpi)=1$, $\gamma_c(\pi,0)=\sqrt{\pi/(4-\pi)}$,
$\gamma_c(\frac\pi2,\frac\pi2)=\sqrt{\pi/(\pi-2)}$ and the closed form
$\gamma_c(\mathbf0)=\pi/(\pi-2)$ of the value $2.75194$ found
numerically earlier; at most two atom--dimer states for every $\bK$ and
$\gamma$, with a second threshold
$\pi/(\pi-2)\le\gamma_c^{(2)}(\bK)\le\pi/(4-\pi)$; the onset laws, namely
an essential singularity $\ln\delta=-\pi c(\bK)/(2(\xi(\bK)-1))+O(1)$ at
the zone boundary wherever the threshold mode has an $s$-wave
component, $\delta\simeq16\,e^{-\pi\gamma/(\gamma-1)}$ at $\bK=\bpi$, and
the law $\delta(\ln\frac1\delta+\kappa_0)\simeq2\pi(\gamma-\gamma_c)/(\gamma\gamma_c)$
with the closed-form constant $\kappa_0=1+4\ln2-\pi$ for the $p$-wave at
$\bK=\mathbf0$; and two exact dualities,
$h_{F}(\bK+\bpi;2)=h_B(\bK)+2$ and
$h_{2+1B}(\bK;\gamma)=h_{F}(\bK+\bpi;\gamma)$ for every $\gamma$:
on the lattice a change of statistics is equivalent, at the level of the
effective atom--dimer Hamiltonians, to the shift $\bK\mapsto\bK+\bpi$.
The reduction of the three-body eigenvalue problem to the lattice
Skornyakov--Ter-Martirosyan equation, used for the independent numerical
verification, is proved with an explicit isomorphism of eigenspaces.
\end{abstract}

\maketitle

\setcounter{tocdepth}{1}
\tableofcontents

\bigskip
\noindent\textbf{Note on version 2.} Compared with version~1 the
following has been added or made rigorous: the location of the essential
spectrum and of the atom--dimer threshold of the three-body operator and
the two-sided transfer of the eigenvalue counting function
(Theorem~\ref{thm:feshbach}, Corollary~\ref{cor:Nsigma}); the counting
criterion at the critical parameters, with the one-sided approach to the
critical eigenvalue (Theorem~\ref{thm:count}, Proposition~\ref{prop:onesided});
the reduction to the lattice Skornyakov--Ter-Martirosyan equation
(Lemma~\ref{lem:STM}); a first-principles derivation of the lattice
constants $A$, $B$ and of $\kappa_0=1+4\ln2-\pi$
(Lemma~\ref{lem:constants}, \S\ref{sec:S25}); the corrected remainder
constant for the $2+1$ bosonic mixture; and an explicit separation of
the contextual comparisons from the proofs.

\section{Introduction}

The spectral theory of few-body lattice Schr\"odinger operators
combines fundamental questions of mathematical physics with direct
relevance to quantum simulation experiments in optical
lattices~\cite{Bloch2012}. The one-particle case is well
developed, whereas the multi-particle case is considerably harder:
the Hilbert space is (anti)symmetrized, the interaction is pairwise, and
the fibre operator admits no closed-form solution even in two
dimensions.

Previous work on this model has focused on the following three
directions. First, the strong-coupling expansion of the three-boson
\emph{trimer} ground state at fixed quasimomentum has been obtained by
Kato perturbation theory \cite{Kato1995} around the isolated eigenvalue $3$ of the
pair-contact operator, giving the uniform leading term $z_1^{\bK,s}(\mu)=-3\mu+6+O(\mu^{-1})$
\cite{Trimer2026}. An elliptic reduction of the fibre Fredholm determinant and an order-matching criterion for the trimer sector are given in \cite{arXiv2608}.
Second, the transcendental lattice constant
$C\approx3.96458$ characterising the $\bK=\mathbf 0$ second bound state
has been determined from a limit equation involving the two-dimensional
lattice Green function \cite{Trimer2026}. Third, the strong-coupling $2+1$ fermionic problem at $\bK=\mathbf 0$
has been analysed in \cite{Companion}, where a doubly degenerate level
$z=-\lambda+e_0(\gamma)+O(\lambda^{-1})$ was found to appear at the
numerically determined critical mass ratio
$\gamma_c\approx2.75194\pm0.00001$ (the fermionic coupling constant is
denoted by $\lambda$ in \cite{Companion}, in~contrast to the bosonic
coupling $\mu$ used there and here). In the present work this level is
identified as the $p$-wave atom--dimer state, and
$\gamma_c(\mathbf0)=\pi/(\pi-2)$ in closed form
(Theorem~\ref{thm:zoneF}).

The present work does not refine the trimer expansion. Its object is
the non-isolated atom--dimer sector associated with the infinitely
degenerate eigenvalue \(1\) of the pair-contact operator.

All of the above treatments address either the \emph{trimer} (a genuine
three-particle bound state) or the $\bK=\mathbf 0$ threshold. Neither
the full Brillouin-zone structure of the atom--dimer states nor the
critical mass ratios at general quasimomentum have been derived so far.

\medskip

The present article closes this gap. Our central object is the
\emph{atom--dimer effective Hamiltonian}, obtained by a
Feshbach--Schur reduction of the fibre operator $H_\mu(\bK)$ onto the
infinitely degenerate eigenvalue $1$ of the pair-contact operator
$V=V_1+V_2+V_3$. In relative coordinates, $V$ is the operator of
multiplication by $\nu(\bx,\by)=[\by=0]+[\bx=0]+[\bx=\by]$, the number
of pairs of particles at the same lattice site; its spectrum on the
bosonic sector is $\{0,1,3\}$ with the eigenvalue $3$ simple
(trimer) and $1$ of infinite multiplicity (atom--dimer). The reduction
of the perturbation $E_{\bK}=\eps(\bp)+\eps(\bq)+\eps(\bK-\bp-\bq)$
onto the atom--dimer subspace gives the explicit effective operators
$h_B$ and $h_F$ of Section~\ref{sec:effective}.

\medskip

\subsection{Statement of the problem}\label{ssec:problem}

Let $\eps(\bp)=2-\cos p_1-\cos p_2$ on $\T^2=(-\pi,\pi]^2$, fix
$\bK\in\T^2$, a~coupling $\mu>0$ and a mass ratio $\gamma>0$. We
consider the fibre operators
\begin{equation}\label{eq:models}
H^{\sigma}_\mu(\bK)=E^{\sigma}_\bK-\mu V^{\sigma},\qquad
\sigma\in\{3B,\ 2{+}1F,\ 2{+}1B\},
\end{equation}
where $E^{3B}_\bK(\bp,\bq)=\eps(\bp)+\eps(\bq)+\eps(\bK-\bp-\bq)$ on
the bosonic subspace $\Ltwo$ of \S\ref{sec:setup} with
$V^{3B}=V_1+V_2+V_3$, and
$E^{2+1\,F/B}_\bK=\eps(\bp)+\eps(\bq)+\gamma\,\eps(\bK-\bp-\bq)$ on the
subspace antisymmetric (F) or symmetric (B) under
$\bp\leftrightarrow\bq$ with $V^{2+1}=V_1+V_2$ (only the two
interspecies pairs interact); the $V_i$ are defined in
\S\ref{sec:setup}. For large $\mu$ the spectrum of $H^\sigma_\mu(\bK)$
splits into clusters near $-3\mu$ (trimer, $\sigma=3B$), near $-2\mu$
($\sigma=2{+}1B$ only: triple occupancy, two interspecies pairs; this
configuration is Pauli-forbidden for $\sigma=2{+}1F$), near
$-\mu$ (atom--dimer) and near $0$. Let
\[
\theta^\sigma_\mu(\bK):=\inf\specess H^\sigma_\mu(\bK)
\]
be the bottom of the essential spectrum. Theorem~\ref{thm:feshbach}(b)
below shows that, for $\mu$ above an explicit bound,
$\theta^\sigma_\mu(\bK)$ lies in the atom--dimer cluster,
\begin{equation}\label{eq:theta-conv}
\bigl|\theta^\sigma_\mu(\bK)+\mu-t_\sigma\bigr|\le\eta_\sigma(\mu)=O(\mu^{-1})
\quad\text{uniformly in }\bK,\qquad t_{3B}=4,\quad t_{2+1F}=t_{2+1B}=2+2\gamma,
\end{equation}
so that $\theta^\sigma_\mu(\bK)$ is the \emph{atom--dimer threshold}, and
everything below it is discrete spectrum. We define the number of
atom--dimer states bound uniformly in $\mu$ by
\begin{equation}\label{eq:Nsigma}
N^\sigma(\bK):=\lim_{\epsilon\downarrow0}\ \liminf_{\mu\to\infty}\
\#\bigl\{\text{eigenvalues of }H^\sigma_\mu(\bK)\text{ in }
[-\mu,\ \theta^\sigma_\mu(\bK)-\epsilon]\bigr\}
\end{equation}
(eigenvalues counted with multiplicity; the limit in $\epsilon$ exists
because the inner quantity is non-increasing in $\epsilon$), and call
$\bK$ an \emph{atom--dimer point} for $\sigma$ if $N^\sigma(\bK)\ge1$. By
Corollary~\ref{cor:Nsigma}, for every $\epsilon>0$ small enough the inner
count is \emph{constant} for all large $\mu$, and $N^\sigma(\bK)$ equals
the number of eigenvalues of the limit operator $h^\sigma(\bK)$ below its
threshold $t_\sigma$. The problems solved in this article are:
\begin{enumerate}[label=(P\arabic*)]
\item \emph{Reduction:} find the limit operator $h^\sigma(\bK)$ on
$\ell^2(\Z^2\setminus\{0\})$ governing the atom--dimer cluster and
bound the error between $\spec H^\sigma_\mu(\bK)$ and
$-\mu+\spec h^\sigma(\bK)$ uniformly in $\bK$.
\item \emph{Zones:} determine the sets
$Z^\sigma=\{\bK: N^\sigma(\bK)\ge1\}$ and the numbers $N^\sigma(\bK)$
in closed form.
\item \emph{Critical mass ratio:} for $\sigma=2{+}1F,\,2{+}1B$
determine $\gamma_c^\sigma(\bK)=\inf\{\gamma:N^\sigma(\bK)\ge1\}$ as a~function on $\T^2$.
\item \emph{Onset:} determine the asymptotics of the binding energy
$\delta^\sigma(\bK)$ as $\bK$ or $\gamma$ approaches the boundary of
the zone.
\item \emph{Statistics:} relate $h^{3B}$, $h^{2+1F}$, $h^{2+1B}$ to
one another.
\end{enumerate}
We write $h_B=h^{3B}$, $h_F=h^{2+1F}$, $h_{2+1B}=h^{2+1B}$ and
$\delta_B$, $\delta_F$, $\delta_{2+1B}$ for the corresponding binding
energies measured from the threshold of the limit operator.
Problem (P1) is solved by Lemma~\ref{lem:heff}, Proposition~\ref{prop:mixed},
Theorem~\ref{thm:feshbach} and Corollary~\ref{cor:Nsigma}; (P2) by Theorems~\ref{thm:count},
\ref{thm:zoneB} and Corollary~\ref{cor:twostate}; (P3) by
Theorem~\ref{thm:zoneF}; (P4) by Theorems~\ref{thm:crit} and
\ref{thm:boundary}; (P5) by Theorem~\ref{thm:dual} and
Proposition~\ref{prop:mixed}.

\medskip

\textbf{Main results.} In the notation of \eqref{eq:models}:
\begin{enumerate}
\item[$\bullet$] \textbf{(Effective Hamiltonians, Lemma~\ref{lem:heff})}
Reduction of $\mu^{-1}H_\mu(\bK)$ onto the atom--dimer sector gives
\[
h_B=4+\varepsilon_D-S_\bK,\qquad
h_F=2+2\gamma+\varepsilon_D+\tfrac{\gamma}{2}S_\bK
\]
on $\ell^2(\Z^2\setminus\{0\})$, where $\varepsilon_D$ is the Dirichlet
Laplacian and $S_\bK$ is the first-shell exchange operator with phase
$e^{i\bK\cdot\be}$.
\item[$\bullet$] \textbf{(Threshold matrix, Lemma~\ref{lem:M})}
The $4\times4$ matrix $M(\delta)=[(\eps_D+\delta)^{-1}(\be,\be')]$
has the regular limit
$M_0=J-D$ with
$A=2-\tfrac4\pi$, $B=\tfrac2\pi$, $A+2B=2$, and eigenvalues
$2$ ($s$), $A,A$ ($p$), $8/\pi-2$ ($d$). Here \(J=\mathbf1\mathbf1^T\) and \(D\) is the matrix of two-point
resistances between the four nearest-neighbour sites; $A$ and $B$ are
derived from first principles in Lemma~\ref{lem:constants}. Since
\[
2>0,\qquad
A=2-\frac4\pi>0,\qquad
\frac8\pi-2>0,
\]
the matrix \(M_0\) is strictly positive definite.
\item[$\bullet$] \textbf{(Bosonic zone, Theorem~\ref{thm:zoneB})}
Atom--dimer states exist iff $F_B(u,v)<0$ with
$F_B$ an~explicit quadratic in $u=\cos K_1$, $v=\cos K_2$. In
particular $\bK=\mathbf 0$ and $(\pi,0)$ are in the zone, while
$\bK=\bpi$ is not.
\item[$\bullet$] \textbf{(Fermionic zone, Theorem~\ref{thm:zoneF})}
Critical mass ratios
\[
\gamma_c(\mathbf 0)=\frac{\pi}{\pi-2},\quad
\gamma_c(\bpi)=1,\quad
\gamma_c(\pi,0)=\sqrt{\frac{\pi}{4-\pi}},\quad
\gamma_c(\tfrac\pi2,\tfrac\pi2)=\sqrt{\frac{\pi}{\pi-2}}.
\]
For $\gamma\le1$ no
atom--dimer state exists for any $\bK$; for
$\gamma>\pi/(\pi-2)$
the limiting atom--dimer operator has a bound state for every $\bK$,
and by Corollary~\ref{cor:Nsigma} the three-body operator has one for
every $\bK$ and every $\mu\ge\mu_0(\gamma)$, uniformly in $\bK$.
\item[$\bullet$] \textbf{(Critical behaviour, Theorems~\ref{thm:crit}, \ref{thm:boundary})}
At $\bK=\bpi$, $\delta_F\simeq16\,e^{-\pi\gamma/(\gamma-1)}$
(infinite-order onset), and more generally an essential singularity
along the zone boundary whenever the threshold mode has an $s$-wave
component; at $\bK=\mathbf 0$ (pure $p$-wave),
$\delta_F(\ln\frac1{\delta_F}+\kappa_0)\simeq2\pi(\gamma-\gamma_c)/(\gamma\gamma_c)$
with $\kappa_0=1+4\ln2-\pi$, hence
$\delta_F\simeq\displaystyle\frac{2\pi}{\gamma_c^2}(\gamma-\gamma_c)/\ln(1/(\gamma-\gamma_c))$.
\item[$\bullet$] \textbf{(Spectral transfer, Theorem~\ref{thm:feshbach}, Corollary~\ref{cor:Nsigma})}
For $\mu\ge46$ (bosons) the spectrum of $H_\mu(\bK)$ in $[-\mu,-\mu+10]$
lies within $\eta_B=36/(\mu-10)$ of $-\mu+\spec h_B(\bK)$; the
atom--dimer threshold satisfies $|\theta_\mu(\bK)+\mu-4|\le\eta_B$; and
the number $n_H$ of eigenvalues of $H_\mu(\bK)$ in $[-\mu,-\mu+s)$
satisfies $n_h(s-\eta_B)\le n_H\le n_h(s+\eta_B)$ for all $s<4-\eta_B$,
where $n_h$ is the counting function of $h_B(\bK)$. Hence
$N^\sigma(\bK)$ of \eqref{eq:Nsigma} equals the number of eigenvalues of
$h^\sigma(\bK)$ below its threshold. The same holds for the $2+1$
systems with explicit constants.
\item[$\bullet$] \textbf{(Dualities, Theorem~\ref{thm:dual},
Proposition~\ref{prop:mixed})}
\[
\underbrace{h_F(\bK+\bpi;2)=h_B(\bK)+2}_{\text{Theorem~\ref{thm:dual}}},\qquad
\underbrace{h_{2+1B}(\bK;\gamma)=h_F(\bK+\bpi;\gamma)\ \ (\text{all }\gamma)}_{\text{Proposition~\ref{prop:mixed}}},
\]
hence $\gamma_c^{2+1B}(\bK)=\gamma_c(\bK+\bpi)$.
\item[$\bullet$] \textbf{(Number of states, Corollary~\ref{cor:twostate})}
$N_F(\bK;\gamma)\le2$ for all $\bK,\gamma$; the second state appears
at $\gamma_c^{(2)}(\bK)\in[\pi/(\pi-2),\pi/(4-\pi)]$.
\item[$\bullet$] \textbf{(External couplings, Propositions~\ref{prop:nn}--\ref{prop:floqAD})}
Nearest-neighbour attraction, three-body force and Floquet modulation
shift the zone boundaries and, in the case of Floquet driving, allow
$\bK\to\bK+\bpi$ inversion of the effective band.
\end{enumerate}

The methods are purely operator-theoretic: a Feshbach--Schur
reduction onto the infinitely degenerate eigenvalue $1$ of $V$ (which
replaces analytic perturbation theory for this non-isolated-multiplicity
situation), combined with a monotonicity argument that transfers the
essential spectrum and the eigenvalue counting function; an exact
finite-rank Birman--Schwinger analysis with the Dirichlet lattice Green
function; and independent numerical verification, based on the lattice
Skornyakov--Ter-Martirosyan equation of Lemma~\ref{lem:STM}. All proofs
are self-contained; references to the classical literature indicate the
origin of the methods and are not used as substitutes for proofs. The order-matching criterion of \cite{arXiv2608}, developed there
for the trimer sector,
is not
used in the present atom--dimer analysis; the~elliptic reduction of
\cite{arXiv2608} is likewise not needed here, since the atom--dimer
sector is governed by the Dirichlet lattice Green function rather
than by the full lattice Green function.

The results obtained here have a bearing beyond the immediate model.
The exact Brillouin-zone dependence of the atom--dimer thresholds and
the closed-form critical mass ratios provide benchmark quantities for
quantum simulation of mass-imbalanced three-body systems in optical
lattices~\cite{Bloch2012}, where the hopping ratio
$\gamma=t_3/t_f$ is tunable and quasimomentum-resolved spectroscopy is
available. The~statistics--quasimomentum duality identifies a lattice-specific
equivalence that can be exploited in Floquet engineering, allowing one
to tune through the critical ratio without changing atomic species.
The~onset laws admit a reading in terms of a marginal coupling and
dimensional transmutation (Remark~\ref{rem:RG}); this interpretation,
and the comparisons with continuum mechanisms in \autoref{tab:novelty},
are contextual and are not used in any proof.

\autoref{tab:tasks} places the problem solved here among the
neighbouring problems of lattice few-body theory, and
\autoref{tab:novelty} summarises what distinguishes it from the
continuum theory and from the lattice trimer analyses
\cite{Trimer2026,Companion,arXiv2608} on which it builds.

\begin{table}[t]\centering\footnotesize
\caption{The problems of the strong-coupling three-body lattice theory
on $\Z^2$, their status and their difficulty. ``Order'' lists the
power of $\mu$
at which the quantity is controlled, in the strong-coupling
expansion; higher orders not used in the present analysis are indicated
as such. The last two
rows are solved in the present article; the~first two are the content
of \cite{Trimer2026,arXiv2608} and are used here only as
context and for self-checks; the value $\gamma_c(\mathbf0)\approx2.75194$ in the last row was
found numerically in \cite{Companion}.}
\label{tab:tasks}
\setlength{\tabcolsep}{4pt}
\begin{tabular}{L{2.9cm}L{1.9cm}L{2cm}L{3cm}L{4.6cm}}
\toprule
Problem & Sector & Order
& Difficulty & Status \\
\midrule
Trimer level $z_1^{\bK,s}(\mu)$ &
$\nu=3$, bosons only &
$O(\mu)$ leading, $O(1)$ constant; higher orders not used here &
moderate: isolated, non-degenerate &
leading terms $-3\mu+6$ known~\cite{Trimer2026}; used here only as background \\
\addlinespace
Second bosonic level at $\bK=\mathbf0$, constant $C$ &
$\nu=1$, $s$-wave &
$\mu^{0}$ &
one transcendental equation &
known \cite{Trimer2026}; independent validation in \cite{arXiv2608}; re-derived here as $C=4-\delta_\infty$ (self-check) \\
\addlinespace
\textbf{Atom--dimer zones for all $\bK$} &
$\nu=1$, four-dimensional threshold &
$\mu^{0}$ &
high: infinite degeneracy, non-isolated spectrum &
\textbf{this work}: Lemma~\ref{lem:heff}, Theorems~\ref{thm:count}, \ref{thm:zoneB}, \ref{thm:zoneF} \\
\addlinespace
\textbf{Critical mass ratios and onset laws} &
$2+1$ fermions and mixtures &
$\mu^{0}$ &
high: $2$D threshold logarithm survives in the $s$-wave only &
\textbf{this work}: Theorems~\ref{thm:zoneF}, \ref{thm:crit},\,\ref{thm:boundary} \\
\bottomrule
\end{tabular}
\end{table}

\begin{table}[t]\centering\footnotesize
\caption{Relevance, novelty and the principal differences of the present
problem from the situations in which three-body critical mass ratios are
classically discussed. ``KM'' denotes the universal Kartavtsev--Malykh
trimers of the
$3$D $2+1$ fermionic problem with $L^P=1^-$ \cite{KM2007}:
one state
for $8.17260<M/m<12.91743$, two for $12.91743<M/m<13.607$; for
$M/m>13.607$ the Efimov effect gives infinitely many states
\cite{Efimov1973,Petrov2003}. ``PP'' refers to the planar
$(1+2)$-body trimers of \cite{PP2010}. The entries of the two continuum
columns are quoted from the cited works for comparison only: no
statement about the continuum problems, in particular about the Efimov
effect, is made or used in this article.}
\label{tab:novelty}
\setlength{\tabcolsep}{4pt}
\begin{tabular}{L{1.7cm}L{4.2cm}L{3cm}L{5.9cm}}
\toprule
Aspect & $3$D continuum, zero range & $2$D continuum, zero range & $2$D lattice $\Z^2$ (this work) \\
\midrule
Total momentum &
irrelevant (Galilean invariance) &
irrelevant (Galilean invariance) &
\textbf{essential}: zones exist only for $\bK$ in a subregion of $\T^2$ \\
\hline
\addlinespace
Efimov effect &
present for $|a|=\infty$; for $2{+}1$ fermions at $M/m>13.607$ \cite{Efimov1973,Petrov2003} &
absent &
absent; replaced by a finite, explicitly countable set ($N\le2$, Corollary~\ref{cor:twostate}) \\
\hline
\addlinespace
Critical mass ratio &
KM: $8.17260$ (ground), $12.91743$ (excited) &
PP: thresholds in $M/m$, odd $L$ (F) / even $L$ (B), no $\bK$ dependence &
$\gamma_c(\bK)$ \textbf{is a function on the Brillouin zone}, in~closed form, $1\le\gamma_c\le\pi/(\pi-2)$ \\
\hline
\addlinespace
Threshold singularity &
power law &
logarithmic &
logarithm survives in the $s$-channel only; essential singularity at the zone boundary (Theorem~\ref{thm:boundary}) \\
\hline
\addlinespace
Role of statistics &
selects partial waves &
selects parity of $L$ \cite{PP2010} &
\textbf{the effective atom--dimer operator is related
to a~quasimomentum shift $\bK\mapsto\bK+\bpi$} (Theorem~\ref{thm:dual}, Proposition~\ref{prop:mixed}) \\
\hline
\addlinespace
Mechanism of the zone &
short-range attraction &
short-range attraction &
Dirichlet condition (triple occupancy removed: deep trimer for $3B$, $2{+}1B$;
Pauli exclusion for $2{+}1F$) $+$ rank-four exchange~$S_\bK$ \\
\bottomrule
\end{tabular}
\end{table}

\begin{remark}[Scope and separation from the trimer analysis]
The trimer sector associated with the isolated eigenvalue \(3\) of
\(V\) is not the subject of this paper. Its strong-coupling expansion
is used only to identify the spectral scale and to separate the
\(\nu=3\) cluster from the atom--dimer cluster. The new analysis begins
at the infinitely degenerate eigenvalue \(1\), where ordinary isolated
eigenvalue perturbation theory is inapplicable and the effective
Hamiltonians \(h_B,h_F,h_{2+1B}\) emerge from a Feshbach--Schur
reduction. All Brillouin-zone zones, critical ratios, threshold
onset laws, and statistics--quasimomentum dualities proved below
belong to this atom--dimer sector.
\end{remark}

\section{Setup and notation}\label{sec:setup}

Let $\T^2=(-\pi,\pi]^2$ and
$\eps(\bp)=2-\cos p_1-\cos p_2$. The two-dimensional fibre Hamiltonian
in the symmetric bosonic sector is
\[
H_\mu(\bK)=E_\bK-\mu(V_1+V_2+V_3)\quad\text{on }\Ltwo,
\]
with
\[
(V_1 f)(\bp,\bq)=\!\int_{\T^2}\!\!f(\bp,\mathbf s)\,d\mathbf s,\quad
(V_2 f)(\bp,\bq)=\!\int_{\T^2}\!\!f(\mathbf s,\bq)\,d\mathbf s,
\]
\[
(V_3 f)(\bp,\bq)=\!\int_{\T^2}\!\!f(\mathbf s,\bp+\bq-\mathbf s)\,d\mathbf s,
\]
where each integral is understood as an average with respect to the
normalized Haar measure on $\T^2$:
$\displaystyle\frac{1}{4\pi^2}\int_{\T^2}\dots\,d\mathbf s$.
Each $V_i$ is an orthogonal projection, none depends on $\bK$, and
$0\le V\le 3I$. The momenta of the three particles are $\bp$, $\bq$ and
$\bK-\bp-\bq$, and
\[
\Ltwo:=\bigl\{f\in L^2((\T^2)^2):\ f(\bp,\bq)=f(\bq,\bp)=f(\bK-\bp-\bq,\bq)\bigr\}
\]
is the space of functions invariant under all permutations of the
three momenta (the bosonic fibre; the~two transpositions generate
$S_3$). $E_\bK$ and $V$ commute with these permutations, so
$H_\mu(\bK)$ is a bounded self-adjoint operator on $\Ltwo$; for
$f\in\Ltwo$ one has $(V_3f)(\bp,\bq)=(V_1f)(\bK-\bp-\bq,\bq)$.

For the $2+1$ fermionic configuration, two identical fermions of mass
$m_f$ and one particle of mass $m_3$, the mass ratio is
$\gamma=m_f/m_3$, \footnote{On the lattice $\gamma$ is, strictly, the hopping ratio $t_3/t_f$; it equals
$m_f/m_3$ in the effective-mass sense and in optical lattices depends also on
species-selective lattice depths.} the two identical fermions interact with the third
particle but not with each other (contact interaction is Pauli-forbidden
in the antisymmetric spatial sector), and the fibre dispersion becomes
\[
E_\gamma(\bp,\bq)=\eps(\bp)+\eps(\bq)+\gamma\,\eps(\bK-\bp-\bq).
\]
The Hilbert space is the antisymmetric sector
$L^2_a((\T^2)^2)$ with respect to $\bp\leftrightarrow\bq$, and the
interaction reduces to $V=V_1+V_2$ ($V_1$ and $V_2$ are the contacts of
the pairs $\{2,3\}$ and $\{1,3\}$; the fermion--fermion contact $V_3$ is
absent). The $2+1$ \emph{bosonic} mixture
(two identical bosons without mutual interaction and a third particle)
has the same $E_\gamma$ and $V=V_1+V_2$ on the symmetric subspace
$L^2_s((\T^2)^2)$ with respect to $\bp\leftrightarrow\bq$.
For uniformity we write $H^\sigma_\mu(\bK)=E^\sigma_\bK-\mu V^\sigma$ on
$\mathcal H^\sigma$, with $\mathcal H^{3B}=\Ltwo$,
$\mathcal H^{2+1F}=L^2_a((\T^2)^2)$, $\mathcal H^{2+1B}=L^2_s((\T^2)^2)$.

All operators above are bounded and self-adjoint;
$H^\sigma_\mu(\bK)$ depends norm-continuously on $\bK$, and
$0\le E^{3B}_\bK\le12$, $0\le E^{2+1}_\bK\le8+4\gamma$.

\begin{remark}[Why the Birman--Schwinger operator is not used directly]
For $z<\min E_\bK$ the Birman--Schwinger operator
$B_\mu(\bK,z)=\mu\,V^{1/2}(E_\bK-z)^{-1}V^{1/2}$ is bounded and
self-adjoint but \emph{not} compact: it~dominates $(\mu/(12-z))V$, and
$V$ has infinite-dimensional eigenspaces. Hence the usual compactness
argument does not apply to the atom--dimer sector, which sits at the
infinitely degenerate eigenvalue $1$ of $V$. In this article the
three-body spectrum is controlled instead by the Feshbach--Schur
reduction and monotonicity argument of Theorem~\ref{thm:feshbach},
while finite-rank Birman--Schwinger counting is applied only to the
limit operators $h^\sigma(\bK)$ (Theorem~\ref{thm:count}). The
reduced Skornyakov--Ter-Martirosyan equation, used for numerical
verification, is stated and proved in Lemma~\ref{lem:STM}.
\end{remark}

\section{The pair-contact operator in Fourier modes}\label{sec:fourier}

In relative coordinates $\bx=\mathbf r_1-\mathbf r_3$,
$\by=\mathbf r_2-\mathbf r_3$ conjugate to $(\bp,\bq)$, the operator
$V$ is diagonal in the Fourier basis indexed by
$(\bm,\bn)\in\Z^2\times\Z^2$.

\begin{lemma}[Spectrum of $V$]\label{lem:Vdiag}
$V$ is multiplication by $\nu(\bx,\by)=[\by=0]+[\bx=0]+[\bx=\by]$;
equivalently, in the Fourier modes
$e_{(\bm,\bn)}(\bp,\bq)=e^{-i(\bm\cdot\bp+\bn\cdot\bq)}$, which correspond to the
position $(\bx,\by)=(\bm,\bn)$ in the convention of \S\ref{sec:S1},
\[
V e_{(\bm,\bn)}=\bigl([\bn=\mathbf 0]+[\bm=\mathbf 0]+[\bm=\bn]\bigr)
e_{(\bm,\bn)}.
\]
Consequently, on $\Ltwo$, $\spec(V)=\{0,1,3\}$, the eigenvalue $3$ is simple
(constant mode), the eigenvalue $1$ has infinite multiplicity (support
on $\{(\bm,0),(0,\bn),(\bm,\bm)\}$ with $\bm,\bn\ne0$), and
$\dist(3,\spec(V)\setminus\{3\})=2$. For the $2+1$ systems
$V^{2+1}=V_1+V_2$ is multiplication by $\nu_{2+1}=[\by=0]+[\bx=0]$; on
$L^2_a$ its spectrum is $\{0,1\}$ (the only point with $\nu_{2+1}=2$,
$(\bx,\by)=(0,0)$, carries no antisymmetric function), and on $L^2_s$ it is
$\{0,1,2\}$ with the eigenvalue $2$ simple. We write $\Pi_\nu$ for the
spectral projection of $V^\sigma$ onto the eigenvalue $\nu$.
\end{lemma}

\begin{proof}
Direct computation: $V_1 e_{(\bm,\bn)}=[\bn=0]e_{(\bm,\bn)}$,
$V_2 e_{(\bm,\bn)}=[\bm=0]e_{(\bm,\bn)}$,
$V_3 e_{(\bm,\bn)}=[\bm=\bn]e_{(\bm,\bn)}$. The value of $\nu$ is $3$ only
at $(0,0)$ and is never $2$, since two of the three conditions imply
the third. The~permutations of \S\ref{sec:S1} preserve $\nu$, so the
diagonalisation descends to $\Ltwo$, $L^2_a$ and $L^2_s$.
\end{proof}

\begin{remark}
This statement is elementary and is included for self-containedness.
The non-trivial content of the present article is the atom--dimer
reduction of Section~\ref{sec:effective}.
\end{remark}

\section{Effective atom--dimer Hamiltonians}\label{sec:effective}

Fix $\bK\in\T^2$ and rescale the fibre operator as
$T(\kappa):=\mu^{-1}H_\mu(\bK)=-V+\kappa E_\bK$, $\kappa=1/\mu$. The
atom--dimer sector corresponds to the eigenvalue $-1$ of the unperturbed
operator $-V$; this eigenvalue is infinitely degenerate and not
isolated in the spectrum of the perturbed operator, so~we use a
Feshbach--Schur reduction (\S\ref{ssec:transfer}).

\subsection{Structure of the atom--dimer subspace}

\begin{lemma}[$\Ran\Pi_1$ in relative coordinates]\label{lem:Pi1}
Let $\Pi_1$ denote the spectral projection of $V$ onto the eigenvalue
$1$. In relative coordinates,
\[
\Ran\Pi_1=\bigl\{\,f:\supp f\subset\{\nu(\bx,\by)=1\}\,\bigr\}.
\]
Restricted to the bosonic sector
$\Ltwo$, we have
\[
\Ran\Pi_1\cap\Ltwo=\bigl\{\varphi(\bp)+\varphi(\bq)+\varphi(\bK-\bp-\bq)
:\ \langle\varphi\rangle=0\bigr\},
\]
where the constant-free condition arises from the exclusion of the
$(0,0)$ Fourier mode. The Fourier convention and the phase of $\tau_{13}$ are those of
\S\ref{sec:S1}. Restricted to the fermionic sector $L^2_a$, one has
$\Ran\Pi_1\cap L^2_a=\{\varphi(\bp)-\varphi(\bq):\ \langle\varphi\rangle=0\}$,
and for the $2+1$ bosonic mixture
$\Ran\Pi_1\cap L^2_s=\{\varphi(\bp)+\varphi(\bq):\ \langle\varphi\rangle=0\}$.
\end{lemma}

\begin{proof}
Immediate from Lemma~\ref{lem:Vdiag}: the condition $\nu=1$ selects
exactly the configurations with one coincident pair. Under the
permutation group $S_3$ acting on the bosonic sector, the atom--dimer
condition is $\mathbf r_1=\mathbf r_2$ or $\mathbf r_2=\mathbf r_3$
or $\mathbf r_1=\mathbf r_3$; the corresponding wavefunction takes the
symmetric form $\varphi(\bp)+\varphi(\bq)+\varphi(\bK-\bp-\bq)$.
Antisymmetry in $\bp\leftrightarrow\bq$ gives
$\varphi(\bp)-\varphi(\bq)$.
\end{proof}

\subsection{Reduction of the perturbation}

Since $\dim\Ran\Pi_1=\infty$ and the restriction of $E_\bK$ to
$\Ran\Pi_1$ has continuous spectrum, the analytic perturbation theory
of isolated eigenvalues of finite multiplicity \cite{Kato1995} does
not apply verbatim. We instead (i) compute the first-order effective
operator $h_1:=\Pi_1 E_\bK\Pi_1$ on $\Ran\Pi_1$ \emph{exactly}
(Lemma~\ref{lem:heff}), and (ii) control the remainder by a
Feshbach--Schur (Schur-complement) argument with explicit constants
(Theorem~\ref{thm:feshbach}).

\begin{lemma}[Effective operators $h_B$, $h_F$]\label{lem:heff}
On the bosonic sector $\Ran\Pi_1\cap\Ltwo$, the first-order effective
Hamiltonian $h_1=\Pi_1 E_\bK\Pi_1$ is unitarily equivalent to the
operator
\[
h_B(\bK):=4+\eps_D-S_\bK
\]
on $\ell^2(\Z^2\setminus\{0\})$, where $\eps_D$ is the Dirichlet
Laplacian
\[
(\eps_D\varphi)(\bx)=2\varphi(\bx)-\tfrac12\sum_{\be:\,|\be|=1}\varphi(\bx+\be),
\qquad \varphi(0)=0,
\]
and $S_\bK$ is the first-shell exchange operator
\[
(S_\bK\varphi)(\bx)=[\bx\in\{\pm\be_1,\pm\be_2\}]\,e^{i\bK\cdot\bx}\varphi(-\bx),
\]
a self-adjoint partial isometry of rank four with $S_\bK^2=P_{\rm sh}$
(the projection onto the first shell) and $S_{\bK+\bpi}=-S_\bK$.
Spectral quantities depend on $\bK$ only through $(\cos K_1,\cos K_2)$.

On the fermionic sector $\Ran\Pi_1\cap L^2_a$, the operator $h_1$ is
unitarily equivalent to
\[
h_F(\bK;\gamma):=2+2\gamma+\eps_D+\tfrac{\gamma}{2}S_\bK
\]
on $\ell^2(\Z^2\setminus\{0\})$, again with the Dirichlet condition
$\varphi(0)=0$.
\end{lemma}

\begin{proof}
We use the position representation of the fibre (Supplementary
Material, \S\ref{sec:S1}): $\psi(\bx,\by)$ with $\bx=\mathbf r_1-\mathbf r_3$,
$\by=\mathbf r_2-\mathbf r_3$, on which $E_\bK$ acts by
\[
(E_\bK\psi)(\bx,\by)=6\psi(\bx,\by)-\tfrac12\sum_{|\be|=1}\bigl[
\psi(\bx+\be,\by)+\psi(\bx,\by+\be)+e^{-i\bK\cdot\be}\psi(\bx+\be,\by+\be)\bigr]
\]
(the three terms are the hops of particles $1$, $2$, $3$). By
Lemma~\ref{lem:Pi1}, $\Ran\Pi_1$ consists of the $\psi$ supported on
the three branches $\{\by=0\ne\bx\}$, $\{\bx=0\ne\by\}$,
$\{\bx=\by\ne0\}$. In the bosonic sector the $S_3$-symmetry
($\psi(\bx,\by)=\psi(\by,\bx)=e^{i\bK\cdot\bx}\psi(-\bx,\by-\bx)$)
determines $\psi$ on all three branches from
$\varphi(\bx):=\psi(\bx,0)$, $\bx\ne0$:
$\psi(0,\by)=\varphi(\by)$, $\psi(\bx,\bx)=e^{i\bK\cdot\bx}\varphi(-\bx)$,
and $\|\psi\|^2=3\|\varphi\|^2$; no constraint is imposed on
$\varphi\in\ell^2(\Z^2\setminus\{0\})$. Evaluate $(\Pi_1E_\bK\psi)(\bx,0)$:
the diagonal gives $6\varphi(\bx)$; hops of particle~1 give
$-\frac12\sum_{\bx+\be\ne0}\varphi(\bx+\be)$ (the target $\bx+\be=0$ is
the trimer site, $\nu=3$, and is removed by $\Pi_1$); a hop of
particle~2 lands in $\Ran\Pi_1$ only if $\by+\be=\bx$, i.e.\
$\bx=\be$ is on the first shell, with value
$\psi(\be,\be)=e^{i\bK\cdot\be}\varphi(-\be)$; a hop of particle~3
lands in $\Ran\Pi_1$ only if $\bx+\be=0$, with value
$e^{i\bK\cdot\bx}\psi(0,-\bx)=e^{i\bK\cdot\bx}\varphi(-\bx)$. Both
exchange processes carry amplitude $-\tfrac12$ and the same value, so
\[
(h_1\varphi)(\bx)=6\varphi(\bx)-\tfrac12\!\!\sum_{\bx+\be\ne0}\!\!\varphi(\bx+\be)
-[\bx\in\{\pm\be_1,\pm\be_2\}]\,e^{i\bK\cdot\bx}\varphi(-\bx)
=\bigl((4+\eps_D-S_\bK)\varphi\bigr)(\bx).
\]
For $2+1$ fermions the diagonal is $4+2\gamma$, particle~3 hops with
amplitude $\gamma/2$, the pair $\{1,2\}$ does not interact (and
$\psi(\bx,\bx)$ is not in $\Ran\Pi_1$), so only the particle-3
exchange survives; antisymmetry $\psi(0,\by)=-\varphi(\by)$ flips its
sign and gives $h_F=2+2\gamma+\eps_D+\tfrac\gamma2S_\bK$, with
$\|\psi\|^2=2\|\varphi\|^2$. The~normalised maps
$\varphi\mapsto\psi/\sqrt3$, resp.\ $\psi/\sqrt2$, are the required
unitaries. Full bookkeeping, including the verification that the
$S_3$-relations are mutually consistent, is in the proof of
Theorem~\ref{thm:S-heff} below and in \S\ref{sec:S1}.
\end{proof}

\begin{remark}[A formal consistency check, not used in the proofs]\label{rem:indcheck}
The operator $h_B$ is also suggested, formally, by letting $\mu\to\infty$
in the lattice Skornyakov--Ter-Martirosyan equation of
Lemma~\ref{lem:STM}: at~$\bK=\mathbf0$ the resulting $s$-channel
condition is $2+\delta-1/b_0(\delta)=1$, which is the $s$-eigenvalue
condition of Lemma~\ref{lem:M} below, and its root reproduces the
constant $C=4-\delta_\infty=3.96458\ldots$ of \cite{Trimer2026}. This
limit is not justified here and plays no role in the arguments; the
effective operators are derived in Lemma~\ref{lem:heff}, and their
relation to the three-body operator is Theorem~\ref{thm:feshbach}.
\end{remark}

\subsection{Spectral transfer to the three-body operator}\label{ssec:transfer}

Throughout, $h=h^\sigma(\bK)$ denotes $h_B(\bK)$, $h_F(\bK;\gamma)$ or
$h_{2+1B}(\bK;\gamma)$ (Lemma~\ref{lem:heff}, Proposition~\ref{prop:mixed}),
and $H=H^\sigma_\mu(\bK)$. We use the constants
\begin{equation}\label{eq:constants}
\begin{array}{c|cccc}
\sigma & t_\sigma & c_\sigma & d_\sigma(\mu) & w_\sigma\\ \hline
3B & 4 & 6 & \mu-10 & 10\\
2{+}1F & 2+2\gamma & 4+2\gamma & \mu-7-3\gamma & 7+3\gamma\\
2{+}1B & 2+2\gamma & 4+2\gamma & \mu-8-4\gamma & 7+3\gamma
\end{array}
\qquad
\eta_\sigma(\mu):=\frac{c_\sigma^2}{d_\sigma(\mu)},\qquad
W_\sigma:=[-\mu,\,-\mu+w_\sigma],
\end{equation}
and the hypothesis
\[
\textbf{(H$_\sigma$)}\qquad d_\sigma(\mu)>0\quad\text{and}\quad\eta_\sigma(\mu)\le1 ,
\]
i.e.\ $\mu\ge46$ for $3B$, $\mu\ge7+3\gamma+(4+2\gamma)^2$ for $2{+}1F$
and $\mu\ge8+4\gamma+(4+2\gamma)^2$ for $2{+}1B$. Since $c_\sigma\ge4$,
(H$_\sigma$) implies $d_\sigma\ge16$ and $\eta_\sigma/d_\sigma\le\frac1{16}$.
Here $t_\sigma$ is the threshold of $h$: by Remark~\ref{rem:S-spec},
$\specess(h)=t_\sigma+[0,4]$, and
\begin{equation}\label{eq:hrange}
\spec h_B\subset[3,9],\qquad
\spec h_F,\ \spec h_{2+1B}\subset[2+\tfrac32\gamma,\ 6+\tfrac52\gamma]\subset(-\infty,w_\sigma-1].
\end{equation}
For $s\in\R$ we write
$n_h(s):=\dim\Ran\mathbf 1_{(-\infty,s)}(h)$ for $s\le t_\sigma$ (the
number of eigenvalues of $h$ below $s$, with multiplicity; it is at
most $4$, Theorem~\ref{thm:count}), and
$n_H(s):=\dim\Ran\mathbf 1_{[-\mu,-\mu+s)}(H)$.

\smallskip
Two elementary operator-theoretic facts are used; they are proved in
\S\ref{sec:S3}. For a bounded self-adjoint operator $T$ we use
$\specess(T):=\spec(T)\setminus\specdisc(T)$, and recall that for real
$\lambda$, $\lambda\notin\specess(T)$ if and only if $T-\lambda$ is
Fredholm \cite[Sect.~XIII.4]{RS4}.

\begin{lemma}[Schur factorisation]\label{lem:schur}
Let $H$ be bounded and self-adjoint on a Hilbert space $\mathcal H$, $P$
an orthogonal projection, $Q=1-P$, and $z\in\R$ such that
$Q(H-z)Q$ is boundedly invertible on $\Ran Q$; write $R_Q(z)$ for its
inverse. Put
\[
F(z):=PHP-PHQ\,R_Q(z)\,QHP\qquad\text{on }\Ran P .
\]
Then, on $\Ran P\oplus\Ran Q$,
\begin{equation}\label{eq:schur}
H-z=\begin{pmatrix}1&PHQ\,R_Q(z)\\0&1\end{pmatrix}
\begin{pmatrix}F(z)-z&0\\0&Q(H-z)Q\end{pmatrix}
\begin{pmatrix}1&0\\R_Q(z)\,QHP&1\end{pmatrix},
\end{equation}
with bounded, boundedly invertible outer factors. Consequently:
(i)~$\varphi\mapsto\varphi-R_Q(z)QHP\varphi$ is a linear isomorphism of
$\ker(F(z)-z)$ onto $\ker(H-z)$; (ii)~$H-z$ is Fredholm if and only if
$F(z)-z$ is Fredholm on $\Ran P$; hence $z\in\spec H\iff z\in\spec F(z)$
and $z\in\specess H\iff z\in\specess F(z)$. Moreover $F(z)$ is
self-adjoint, real-analytic in $z$, and
$F'(z)=-PHQ\,R_Q(z)^2\,QHP\le0$.
\end{lemma}

\begin{lemma}[Stability under bounded perturbations]\label{lem:essstab}
Let $T,B$ be bounded and self-adjoint and $\|B\|\le r$. Then
\begin{enumerate}[label=(\roman*)]
\item $\spec(T+B)\subset\spec(T)+[-r,r]$ and\\ $\specess(T+B)\subset\specess(T)+[-r,r]$;
in particular $|\inf\specess(T+B)-\inf\specess(T)|\le r$;
\item the min--max values $\lambda_j(\cdot)$ (counted from below) satisfy
$|\lambda_j(T+B)-\lambda_j(T)|\le r$; moreover $\lambda_j$ and
$\inf\specess$ are monotone with respect to the operator order.
\end{enumerate}
\end{lemma}

\begin{theorem}[Spectral transfer with explicit constants]\label{thm:feshbach}
Let $\sigma\in\{3B,2{+}1F,2{+}1B\}$, assume \textup{(H$_\sigma$)}, and
write $t=t_\sigma$, $\eta=\eta_\sigma(\mu)$, $W=W_\sigma$. For every
$\bK\in\T^2$:
\begin{enumerate}[label=(\alph*)]
\item \emph{(Localisation)}
$\spec H\cap W\subset-\mu+\spec h+[-\eta,\eta]$.
\item \emph{(Essential spectrum and atom--dimer threshold)}
$\theta:=\inf\specess H$ satisfies $|\theta+\mu-t|\le\eta$, and
$\theta\in\specess H$. Below $-\mu$ the spectrum of $H$ consists of
exactly one simple eigenvalue for $\sigma=3B$, lying in
$-3\mu+6+[-\frac{18}{\mu-6},\frac{18}{\mu-6}]$ (the trimer), exactly one
simple eigenvalue for $\sigma=2{+}1B$, lying in
$-2\mu+4+2\gamma+[-\eta,\eta]$, and is empty for $\sigma=2{+}1F$.
\item \emph{(Counting)} For every $s<t-\eta$ the spectrum of $H$ in
$[-\mu,-\mu+s)$ is discrete and
\[
n_h(s-\eta)\ \le\ n_H(s)\ \le\ n_h(s+\eta).
\]
\item \emph{(Isolated levels)} If $e$ is an eigenvalue of $h$ of
multiplicity $m$ with $t-e>2\eta$ and $\dist(e,\spec h\setminus\{e\})>2\eta$,
then $H$ has exactly $m$ eigenvalues, counted with multiplicity, in
$I_e:=-\mu+[e-\eta,e+\eta]$; they lie below $\theta$ and are therefore
discrete eigenvalues of finite multiplicity.
\end{enumerate}
\end{theorem}

\begin{proof}[Structure of the proof; details in \S\ref{sec:S3}]
Take $P=\Pi_1$, $Q=1-P$. Since $PVQ=0$, $PHP=-\mu+h_1$ with
$h_1=\Pi_1E\Pi_1\cong h$ (Lemma~\ref{lem:heff}), and
$PHQ=P(E-c)Q$ for every constant $c$, so $\|PHQ\|\le c_\sigma$ with $c$
the midpoint of the range of $E$. By Lemma~\ref{lem:essstab}(i),
$\spec(QHQ)$ lies in the union of $c_\sigma$-neighbourhoods of the
points $-\mu\nu+c$, $\nu\in\spec(V)\setminus\{1\}$, which is at distance
$\ge d_\sigma$ from $W$. Hence Lemma~\ref{lem:schur} applies on $W$, and
$\|F(z)+\mu-h_1\|\le\eta$, $-\frac{\eta}{d_\sigma}\le F'(z)\le0$. Then:
(a) follows from Lemma~\ref{lem:essstab}(i);
(b) follows from $z\in\specess H\iff z\in\specess F(z)$, from
$\specess F(z)\subset-\mu+t+[-\eta,\eta]+[0,4]$, and from the strict
monotonicity of $z\mapsto z-\inf\specess F(z)$, which has exactly one
zero in $-\mu+t+[-\eta,\eta]$; below $-\mu$ the operator $H$ differs from
$E-\mu\Pi_1\ge-\mu$ by a rank-one term, so $\specess H\subset[-\mu,\infty)$,
and the deep level is located by the same argument with the rank-one
projection $\Pi_3$ (resp.\ $\Pi_2$);
(c) follows by applying the monotone fixed-point argument to the
min--max values of $F(z)$, which are non-increasing and Lipschitz in $z$
with constant $\le\frac1{16}$, so that each of them crosses the diagonal
at most once, and by Lemma~\ref{lem:essstab}(ii); (d) follows from (a)
and (c).
\end{proof}

\begin{corollary}[Three-body content of the limit count]\label{cor:Nsigma}
Assume the setting of Theorem~\ref{thm:feshbach} and fix $\bK$. Let
$t-\delta_1>\dots$ denote the eigenvalues of $h$ below $t$, and put
$\delta_{\min}:=\min_j\delta_j$ ($\delta_{\min}:=+\infty$ if there are
none).
\begin{enumerate}[label=(\roman*)]
\item For every $\epsilon\in(0,\delta_{\min})$ and every $\mu$ with
$2\eta_\sigma(\mu)<\min\{\epsilon,\delta_{\min}-\epsilon\}$, the number of
eigenvalues of $H^\sigma_\mu(\bK)$ in $[-\mu,\theta^\sigma_\mu(\bK)-\epsilon]$
equals $n_h(t)$. In particular $N^\sigma(\bK)$ of \eqref{eq:Nsigma}
equals the number of eigenvalues of $h^\sigma(\bK)$ below its threshold,
and the zone $Z^\sigma=\{\bK:N^\sigma(\bK)\ge1\}$ is the set of $\bK$ at
which $h^\sigma(\bK)$ has a bound state.
\item If $h$ has an eigenvalue $t-\delta$ with $\delta>2\eta_\sigma(\mu)$, then
$H^\sigma_\mu(\bK)$ has at least one discrete eigenvalue in
$-\mu+[t-\delta-\eta,\,t-\delta+\eta]$, below $\theta^\sigma_\mu(\bK)$.
\item If $\mathcal K\subset\T^2$ is compact and $\delta_*:=\inf_{\bK\in\mathcal K}\delta_1(\bK)>0$,
where $\delta_1(\bK)$ is the largest binding energy of $h^\sigma(\bK)$, then
for every $\mu$ with $2\eta_\sigma(\mu)<\delta_*$ the operator
$H^\sigma_\mu(\bK)$ has a discrete eigenvalue below
$\theta^\sigma_\mu(\bK)-(\delta_*-2\eta_\sigma(\mu))$, simultaneously for all
$\bK\in\mathcal K$.
\end{enumerate}
\end{corollary}

\begin{proof}
(i) By Theorem~\ref{thm:feshbach}(b),
$[-\mu,-\mu+t-\eta-\epsilon)\subset[-\mu,\theta-\epsilon]\subset[-\mu,-\mu+t+\eta-\epsilon]$.
Apply Theorem~\ref{thm:feshbach}(c) with $s_-=t-\eta-\epsilon$ and with
any $s_+\in(t+\eta-\epsilon,\,t-\eta)$ (non-empty since
$2\eta<\epsilon$): the~count lies between $n_h(t-2\eta-\epsilon)$ and
$n_h(t+2\eta-\epsilon+0)$, and both equal $n_h(t)$ because no eigenvalue
of $h$ lies in $[t-2\eta-\epsilon,t)$ when $2\eta+\epsilon<\delta_{\min}$.
The inner count in \eqref{eq:Nsigma} is therefore eventually constant in
$\mu$ for each small $\epsilon$, which gives $N^\sigma(\bK)=n_h(t)$.
(ii) Apply Theorem~\ref{thm:feshbach}(c) with
$s=t-\delta+\eta+\epsilon'$, $0<\epsilon'<\delta-2\eta$: $n_H(s)\ge
n_h(t-\delta+\epsilon')\ge1$, and by (a) the eigenvalue lies within $\eta$
of $-\mu+\spec h$. (iii)~follows from (ii) and
Theorem~\ref{thm:feshbach}(b), since $\delta_1(\bK)\ge\delta_*$ on $\mathcal K$.
\end{proof}

\begin{remark}[Scope of the conclusions and regime of applicability]\label{rem:rate}
Three levels of statements are made in this article and are kept
distinct. (1)~Statements about the limit operators $h^\sigma(\bK)$
(Sections~\ref{sec:threshold}--\ref{sec:duality}) are exact for every
$\bK$ and $\gamma$. (2)~Statements about the three-body operators
$H^\sigma_\mu(\bK)$ at finite $\mu$ follow from them through
Theorem~\ref{thm:feshbach} and Corollary~\ref{cor:Nsigma}, and are
quantitative: a bound state of $h^\sigma(\bK)$ with binding energy $\delta$ is
guaranteed in the three-body problem as soon as $\delta>2\eta_\sigma(\mu)$,
i.e.\ for
\[
\mu>\mu_*(\bK):=10+\frac{72}{\delta_B(\bK)}\ \ (3B),\qquad
\mu>\mu_*(\bK):=7+3\gamma+\frac{2(4+2\gamma)^2}{\delta_F(\bK)}\ \ (2{+}1F).
\]
(3)~The zones $Z^\sigma$ are $\mu\to\infty$ statements in the sense of
\eqref{eq:Nsigma}. \autoref{fig:transfer}(c) shows $\mu_*(\bK)$ along
the high-symmetry path; representative values are
\begin{center}\small
\begin{tabular}{lcccc}
\toprule
 & bosons, $\Gamma$ & bosons, $(\frac\pi2,\frac\pi2)$ & bosons, $X$ & fermions $\gamma=3$, $\Gamma$ \\
\midrule
$\delta$ & $3.542\cdot10^{-2}$ & $6.505\cdot10^{-4}$ & $1.45\cdot10^{-15}$ & $5.193\cdot10^{-2}$ ($\times2$) \\
$\mu_*$ & $2.04\cdot10^{3}$ & $1.11\cdot10^{5}$ & $5.0\cdot10^{16}$ & $3.87\cdot10^{3}$ \\
\bottomrule
\end{tabular}
\end{center}
The criterion is sufficient, not sharp: the exact solver of
\S\ref{sec:numerical} shows the bosonic state at $\Gamma$ already for
$\mu\ge25$, with $\theta_\mu-z_2(\mu)=\delta_\infty+0.36/\mu+O(\mu^{-2})$
(Table~\ref{tab:kappa}, \autoref{fig:transfer}(b)). Near a zone boundary,
where the binding is exponentially small (Theorem~\ref{thm:boundary}),
$\mu_*$ diverges: there the existence of a three-body state at a given
physical $\mu$ is \emph{not} asserted, and the zone boundary is a
statement about the limit $\mu\to\infty$ only. In~particular, at
$\bK=(\pi,0)$ the bosonic state of $h_B$ has $\delta_B=1.45\cdot10^{-15}$
and is of no physical consequence at any realistic coupling.
\end{remark}

\begin{remark}[Position of the threshold beyond leading order]\label{rem:thrnum}
Theorem~\ref{thm:feshbach}(b) gives $\theta_\mu(\bK)=-\mu+4+O(\mu^{-1})$
uniformly in $\bK$ (bosons). The exact finite-volume computation
(\autoref{fig:transfer}(a)) shows
$\theta_\mu(\bK)=-\mu+4-(2+\cos K_1+\cos K_2)/\mu+O(\mu^{-2})$; this agrees
with the formal expansion
$-\mu+4-(2+\cos k_1+\cos k_2)/\mu$ of the two-body bound-state energy at
pair quasimomentum $\mathbf k$, minimised over the atom momentum
($\mathbf k=\bK$). This refinement is recorded for orientation and is
not used.
\end{remark}

\begin{figure}[t]\centering
\includegraphics[width=\textwidth]{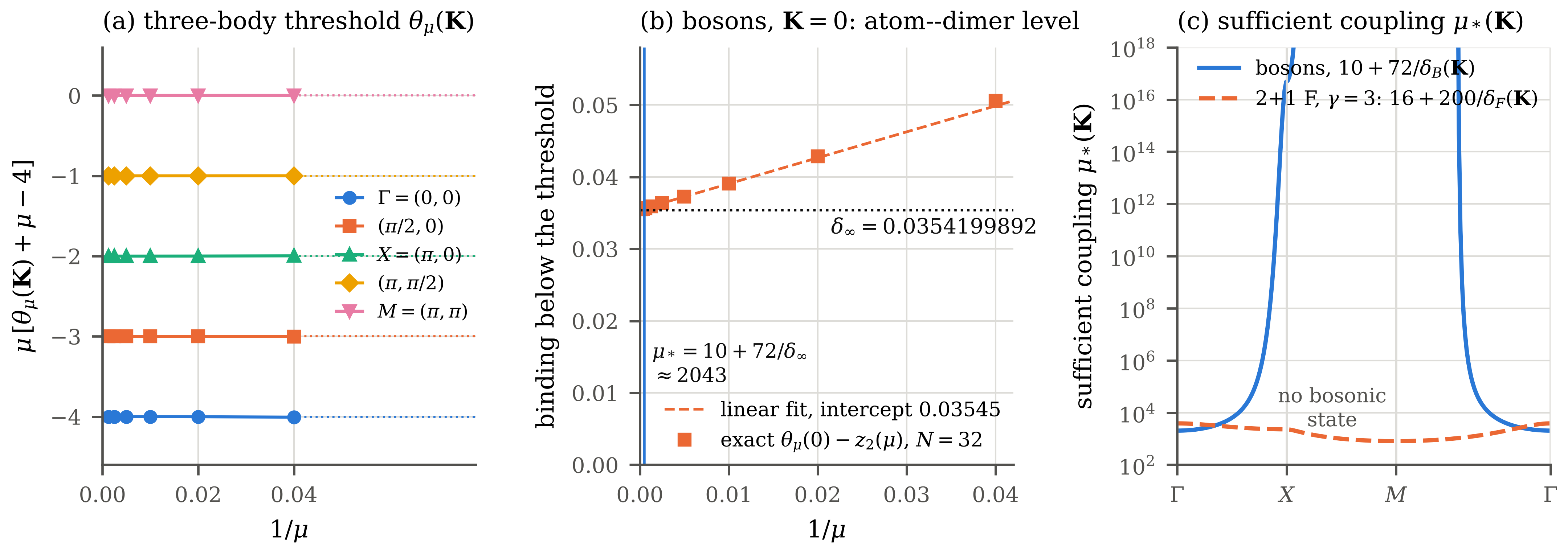}
\caption{Three-body content of the effective theory
(Theorem~\ref{thm:feshbach}, Corollary~\ref{cor:Nsigma}).
\textbf{(a)}~The atom--dimer threshold
$\theta_\mu(\bK)=\inf\specess H_\mu(\bK)$ of the three-boson fibre
operator, from the exact finite-volume solver of Lemma~\ref{lem:STM}
($N=32$), plotted as $\mu[\theta_\mu(\bK)+\mu-4]$ against $1/\mu$ for
five quasimomenta; dotted lines: $-(2+\cos K_1+\cos K_2)$
(Remark~\ref{rem:thrnum}). The rigorous statement is
$|\theta_\mu(\bK)+\mu-4|\le36/(\mu-10)$ (Theorem~\ref{thm:feshbach}(b)).
\textbf{(b)} Bosons at $\bK=\mathbf0$: distance of the atom--dimer
eigenvalue $z_2(\mu)$ to the exact threshold; it tends to
$\delta_\infty=\delta_B(\mathbf0)$, the binding energy of $h_B(\mathbf0)$,
linearly in $1/\mu$ (fit on $\mu\ge200$). Vertical line: the coupling
$\mu_*=10+72/\delta_\infty$ above which Corollary~\ref{cor:Nsigma}(ii)
guarantees the state. \textbf{(c)}~The sufficient coupling $\mu_*(\bK)$ of
Remark~\ref{rem:rate} along $\Gamma$--$X$--$M$--$\Gamma$, from the exact
infinite-lattice binding energies of $h_B$ (solid; no state on a
neighbourhood of $M$) and of $h_F$ at $\gamma=3$ (dashed). Its divergence
at the bosonic zone boundary is the essential singularity of
Theorem~\ref{thm:boundary}.}
\label{fig:transfer}
\end{figure}

\begin{remark}[Physical interpretation]
The effective Hamiltonian $h_B$ describes an atom propagating on the
lattice with Dirichlet boundary condition at the dimer position
(preventing the atom from sitting on top of the dimer, which would
form a trimer), plus an exchange term $S_\bK$ accounting for the
possibility that the dimer's constituent hops to the atom's position
and the atom becomes the new constituent of the dimer. The Dirichlet
condition is a hard core: the configuration with the atom on top of
the dimer has energy $-3\mu$ and is removed by $\Pi_1$.
\end{remark}

\section{Threshold matrix and counting criterion}\label{sec:threshold}

The effective Hamiltonians $h_B$ and $h_F$ are dominated by the
Dirichlet Laplacian $\eps_D$, whose Green function on the first
shell has a finite limit at the band edge (approached logarithmically
slowly in the $s$-channel only, Remark~\ref{rem:slow}). This is the
crucial simplification over the full two-body lattice Green function,
and leads to an explicit algebraic counting criterion.

\subsection{Dirichlet Green function on the first shell}

For $\delta>0$ and $\bx\in\Z^2$ set
\[
G(\bx,\delta)=\frac{1}{4\pi^2}\int_{\T^2}\frac{e^{i\bp\cdot\bx}}{\eps(\bp)+\delta}\,d\bp,
\qquad b_0(\delta):=G(\mathbf0,\delta),\qquad
R(\bx,\delta):=b_0(\delta)-G(\bx,\delta)=\Bigl\langle\frac{1-\cos\bp\cdot\bx}{\eps(\bp)+\delta}\Bigr\rangle .
\]
$b_0(\delta)$ diverges logarithmically as $\delta\downarrow0$, whereas
$R(\bx,\delta)$ increases, by monotone convergence, to the finite limit
$R(\bx)=\langle(1-\cos\bp\cdot\bx)/\eps(\bp)\rangle$, since
$(1-\cos\bp\cdot\bx)/\eps(\bp)$ is bounded. The numbers $R(\bx)$ are the
two-point resistances of the square lattice of unit resistors
\cite{Cserti,Atk}; we do not rely on these references, since the three
values used below are derived in Lemma~\ref{lem:constants}. The
Dirichlet Green function on $\Z^2\setminus\{0\}$ is given by the
rank-one (Krein) formula
\begin{equation}\label{eq:krein}
(\eps_D+\delta)^{-1}(\bx,\by)=G(\bx-\by,\delta)
-\frac{G(\bx,\delta)G(\by,\delta)}{b_0(\delta)}\qquad(\bx,\by\ne0)
\end{equation}
(Lemma~\ref{lem:S-krein}).

\begin{lemma}[Lattice constants from first principles]\label{lem:constants}
Let $a=2+\delta$ and $g_1=G(\be_1,\delta)$, $g_2=G(2\be_1,\delta)$,
$g_{11}=G(\be_1+\be_2,\delta)$, $A(\delta):=R(2\be_1,\delta)$,
$B(\delta):=R(\be_1+\be_2,\delta)$. Then:
\begin{enumerate}[label=(\roman*)]
\item $ab_0-2g_1=1$ and $2ag_1=b_0+g_2+2g_{11}$; hence
$R(\be_1,\delta)=\frac12(1-\delta b_0)$ and
$A(\delta)+2B(\delta)=2+\delta-\delta(4+\delta)\,b_0(\delta)$;
\item $b_0(\delta)=\frac{2}{\pi a}\Ell(\frac2a)
=\frac1{2\pi}\ln\frac{16}{\delta}+O(\delta\ln\frac1\delta)$, with $\Ell$
the complete elliptic integral of the first kind;
\item $A(\delta)=\dfrac2\pi\displaystyle\int_0^2\sqrt{\frac{u(2-u)}{(u+\delta)(u+2+\delta)}}\,du$;
in particular $A:=A(0^+)=2-\frac4\pi$;
\item $R(\be_1)=\frac12$ and $B:=B(0^+)=\frac2\pi$, so that $A+2B=2$;
\item $A(\delta)=A-\dfrac{\delta}{\pi}\Bigl(\ln\dfrac1\delta+\kappa_0\Bigr)+o(\delta)$
as $\delta\downarrow0$, with $\kappa_0=1+4\ln2-\pi=0.6309960686\ldots$
\end{enumerate}
\end{lemma}

\begin{proof}
(i) is the lattice equation $(\eps+\delta)G(\cdot,\delta)=\delta_{\mathbf0}$
at $\bx=\mathbf0$ and $\bx=\be_1$, using the symmetries
$G(\pm\be_j)=g_1$, $G(\pm2\be_j)=g_2$, $G(\pm\be_1\pm\be_2)=g_{11}$; the
two consequences are algebra. (ii) is Lemma~\ref{lem:S-b0}. (iii)~Integrating over $p_2$ with
$\frac1{2\pi}\int_{-\pi}^{\pi}\frac{d\theta}{\alpha-\cos\theta}=(\alpha^2-1)^{-1/2}$,
$\alpha=a-\cos p_1>1$, gives
$A(\delta)=\langle2\sin^2p_1/(\eps+\delta)\rangle
=\frac2\pi\int_0^\pi\sin^2p\,((a-\cos p)^2-1)^{-1/2}dp$, and the
substitution $u=1-\cos p$ gives the stated formula. At $\delta=0$ the
integrand is $\sqrt{(2-u)/(2+u)}$; with $u=2\cos\vartheta$ it becomes
$4\sin^2(\vartheta/2)\,d\vartheta$ on $[0,\pi/2]$, whose integral is
$\pi-2$; monotone convergence gives $A=\frac2\pi(\pi-2)$. (iv) By (ii),
$\delta b_0(\delta)\to0$, so (i) gives $R(\be_1)=\frac12$ and
$A+2B=2$; with (iii), $B=\frac2\pi$. (v) is proved in \S\ref{sec:S25}.
\end{proof}

\begin{lemma}[Threshold matrix]\label{lem:M}
Order the shell as $(\be_1,-\be_1,\be_2,-\be_2)$ and let $P$ denote the
orthogonal projection of $\ell^2(\Z^2\setminus\{0\})$ onto it. The
matrix $M(\delta)=P(\eps_D+\delta)^{-1}P=\bigl[(\eps_D+\delta)^{-1}(\be,\be')\bigr]$ is, exactly,
\begin{equation}\label{eq:Mexact}
M(\delta)_{\be\be'}=G(\be-\be',\delta)-\frac{g_1^2}{b_0}.
\end{equation}
It is invariant under the point group $C_{4v}$ of the shell and has
the exact eigenvalues
\begin{equation}\label{eq:Meig}
m_s(\delta)=2+\delta-\frac1{b_0(\delta)},\qquad
m_p(\delta)=A(\delta)\ (\times2),\qquad
m_d(\delta)=2B(\delta)-A(\delta),
\end{equation}
on $\mathbf1=(1,1,1,1)$, on $(1,-1,0,0),(0,0,1,-1)$, and on
$(1,1,-1,-1)$ respectively. As $\delta\downarrow0$,
\begin{equation}\label{eq:Mlimit}
M(\delta)=M_0-\frac{1}{4b_0(\delta)}\,\mathbf1\mathbf1^{T}+O(\delta\ln\tfrac1\delta),
\qquad
M_0=\mathbf1\mathbf1^{T}-\begin{pmatrix}
0&A&B&B\\ A&0&B&B\\ B&B&0&A\\ B&B&A&0\end{pmatrix},
\end{equation}
with $A=2-\tfrac4\pi$, $B=\tfrac2\pi$ (Lemma~\ref{lem:constants}). Thus
$M_0$ has eigenvalues $2$ ($s$), $A,A$ ($p$), $\tfrac8\pi-2$ ($d$) and is
positive definite. Moreover $M(\delta)$ is positive definite and
strictly decreasing in $\delta$, and
\begin{equation}\label{eq:Mstrict}
M_0-M(\delta)\ \text{is positive definite for every }\delta>0 .
\end{equation}
\emph{The approach to $M_0$ is logarithmically slow, but only in the
$s$-component}: $m_p(\delta)=A-\frac\delta\pi(\ln\frac1\delta+\kappa_0)+o(\delta)$
and $m_d(\delta)=\frac8\pi-2-(1-\frac2\pi)\delta+o(\delta)$.
\end{lemma}

\begin{proof}
\eqref{eq:Mexact} is \eqref{eq:krein}. Row sums of \eqref{eq:Mexact}
equal $b_0+g_2+2g_{11}-4g_1^2/b_0=2ag_1-4g_1^2/b_0=2g_1(ab_0-2g_1)/b_0
=2g_1/b_0=a-1/b_0$ by Lemma~\ref{lem:constants}(i), which is $m_s$; on
the $p$- and $d$-vectors the rank-one term cancels and \eqref{eq:Meig}
follows from the definitions of $A(\delta),B(\delta)$. Writing
$R_1:=R(\be_1,\delta)=(1-\delta b_0)/2$, one has
$M_{\be\be'}=2R_1-R(\be-\be',\delta)-R_1^2/b_0$ and
$R_1^2/b_0=\frac1{4b_0}-\frac\delta2+\frac{\delta^2b_0}4$, whence
\eqref{eq:Mlimit} by Lemma~\ref{lem:constants} and the bound
$|R(\bx,\delta)-R(\bx)|\le C_\bx\delta\ln\frac1\delta$ of
\eqref{S:Mlim}. Positivity and monotonicity: for $v\in\C^4\setminus\{0\}$,
$\langle v,M(\delta)v\rangle=\|(\eps_D+\delta)^{-1/2}Pv\|^2>0$ and
$\frac{d}{d\delta}\langle v,M(\delta)v\rangle=-\|(\eps_D+\delta)^{-1}Pv\|^2<0$,
because $(\eps_D+\delta)^{-1}$ is injective. Hence for $0<\delta'<\delta$
the difference $M(\delta')-M(\delta)$ is positive definite and
increasing as $\delta'\downarrow0$; letting $\delta'\downarrow0$ gives
\eqref{eq:Mstrict}. The expansion of $m_p$ is
Lemma~\ref{lem:constants}(v). For $m_d$, Lemma~\ref{lem:constants}(i)
gives $m_d(\delta)=2+\delta-\delta(4+\delta)b_0(\delta)-2A(\delta)$, and
by Lemma~\ref{lem:constants}(ii),(v) the logarithms cancel:
$m_d(\delta)=2-2A+\delta\bigl[1+\frac2\pi(\kappa_0-4\ln2)\bigr]+o(\delta)
=\frac8\pi-2-(1-\frac2\pi)\delta+o(\delta)$.
\end{proof}

\begin{remark}\label{rem:slow}
The $s$-channel correction $-\displaystyle\frac1{4b_0}\approx-\frac{\pi}{2\ln(16/\delta)}$
in \eqref{eq:Mlimit} is not small at any practical $\delta$
($-0.13$ at $\delta=10^{-4}$, $-0.09$ at $\delta=10^{-6}$). It is
responsible for the essential-singularity onsets of
Section~\ref{sec:critical}, and it makes finite-volume
determinations of weak bindings unreliable (Table~\ref{tab:effective}).
\end{remark}

\subsection{Counting criterion, including the critical parameters}

\begin{theorem}[Counting criterion]\label{thm:count}
Let $C$ be a self-adjoint operator on $\ell^2(\{\pm\be_1,\pm\be_2\})\cong\C^4$,
extended by zero to $\ell^2(\Z^2\setminus\{0\})$, and $h:=\eps_D+C$, so
that $\specess(h)=[0,4]$.
\begin{enumerate}[label=(\roman*)]
\item For every $\delta>0$, the number of eigenvalues of $h$ below
$-\delta$ (with multiplicity) equals the number of eigenvalues of
$M(\delta)^{1/2}CM(\delta)^{1/2}$ less than $-1$; $-\delta$ is an
eigenvalue of $h$ if and only if $\det(I+M(\delta)C)=0$, and its
multiplicity equals $\dim\ker\bigl(I+M(\delta)^{1/2}CM(\delta)^{1/2}\bigr)$.
\item For \emph{every} $C$, the number of eigenvalues of $h$ below $0$ is
\[
N=\#\bigl\{\text{eigenvalues of }M_0^{1/2}CM_0^{1/2}\text{ less than }-1\bigr\}
=n_-\bigl(C+M_0^{-1}\bigr),
\]
where $n_-$ is the number of negative eigenvalues. If $-1$ is an
eigenvalue of $M_0^{1/2}CM_0^{1/2}$, the corresponding eigenvalues of
$M(\delta)^{1/2}CM(\delta)^{1/2}$ are $>-1$ for every $\delta>0$: the
critical eigenvalue $-1$ is approached from above and does not produce
an eigenvalue of $h$ below $0$.
\end{enumerate}
The effective Hamiltonians of the three statistics correspond to
$h_B-4$ with $C=-S_\bK$, $h_F-(2+2\gamma)$ with $C=\frac\gamma2S_\bK$, and
$h_{2+1B}-(2+2\gamma)$ with $C=-\frac\gamma2S_\bK$.
\end{theorem}

\begin{proof}
(i) Write $h+\delta=(\eps_D+\delta)^{1/2}(I+Y)(\eps_D+\delta)^{1/2}$ with
$Y=(\eps_D+\delta)^{-1/2}C(\eps_D+\delta)^{-1/2}=TCT^*$,
$T:=(\eps_D+\delta)^{-1/2}P$. Since $(\eps_D+\delta)^{1/2}$ is bounded
and boundedly invertible, Sylvester's law of inertia for quadratic forms
gives that the spectral subspaces of $h+\delta$ and of $I+Y$ for
$(-\infty,0)$ have the same dimension, and
$\ker(h+\delta)=(\eps_D+\delta)^{-1/2}\ker(I+Y)$. The operator $Y$ has
finite rank, and its non-zero eigenvalues, with multiplicities, coincide
with those of $CT^*T=CM(\delta)$, hence with those of the Hermitian
matrix $M(\delta)^{1/2}CM(\delta)^{1/2}$, to which $CM(\delta)$ is similar.
This gives the count, and $\dim\ker(I+Y)=\dim\ker(I+M(\delta)^{1/2}CM(\delta)^{1/2})$;
finally $\det(I+M^{1/2}CM^{1/2})=\det(I+MC)$.

(ii) Let $N(\delta)$ be the count in (i). Since $M^{1/2}(C+M^{-1})M^{1/2}
=I+M^{1/2}CM^{1/2}$ for $M>0$, Sylvester's law gives
$N(\delta)=n_-(X(\delta))$ with $X(\delta):=C+M(\delta)^{-1}$, and
likewise $N=n_-(X_0)$, $X_0:=C+M_0^{-1}$. By \eqref{eq:Mstrict} and the
strict operator monotonicity of inversion on positive definite
matrices, $X(\delta)-X_0=M(\delta)^{-1}-M_0^{-1}$ is positive definite;
by Weyl's inequality every eigenvalue satisfies
$\lambda_j(X(\delta))>\lambda_j(X_0)$, and $\lambda_j(X(\delta))\to\lambda_j(X_0)$
as $\delta\downarrow0$ by \eqref{eq:Mlimit}. Hence a negative eigenvalue
of $X_0$ gives a negative eigenvalue of $X(\delta)$ for all small
$\delta$, while a zero eigenvalue of $X_0$ gives a positive eigenvalue of
$X(\delta)$ for every $\delta>0$; so $N(\delta)=n_-(X_0)$ for all small
$\delta>0$. On the other hand $N(\delta)$ is non-increasing in $\delta$,
and its limit as $\delta\downarrow0$ is the number of eigenvalues of $h$
below $0$, which is finite (at most $\operatorname{rank}C\le4$, by the
min--max principle, since $\eps_D\ge0$). The statement about the
critical eigenvalue is the case of a zero eigenvalue of $X_0$, through
the congruence above.
\end{proof}

\begin{remark}[The critical set]\label{rem:critset}
Let $\Sigma^\sigma:=\{(\bK,\gamma):\ -1\in\spec(M_0^{1/2}CM_0^{1/2})\}$.
By Lemma~\ref{lem:char} below, for bosons $\Sigma^{3B}=\{\bK:P_\bK(1)=0\}$,
which is exactly the zone boundary, because the second positive root of
$P_\bK$ never exceeds $A<1$ (proof of Theorem~\ref{thm:zoneB}). For
$2+1$ fermions, $\Sigma^{2+1F}=\{(\bK,\gamma):P_\bK(-2/\gamma)=0\}$ is the
union of the graphs of $\gamma_c$ and of the second threshold
$\gamma_c^{(2)}$ (Corollary~\ref{cor:twostate}); in particular it
contains points inside the region where the first state already exists,
e.g.\ $(\bpi,\pi/(4-\pi))$, at which $\gamma>\gamma_c(\bpi)=1$. Version~1
described $\Sigma^{2+1F}$ incorrectly as ``exactly the zone boundary'';
Theorem~\ref{thm:count}(ii) now holds without excluding $\Sigma^\sigma$,
and at every point of $\Sigma^\sigma$ the critical channel contributes
\emph{no} bound state. Consequently all zone criteria below are strict
inequalities, valid on the whole parameter space.
\end{remark}

\begin{proposition}[One-sided approach at the symmetric critical points]\label{prop:onesided}
Let $\nu(\delta)$ denote the eigenvalue of $M(\delta)^{1/2}CM(\delta)^{1/2}$
that tends to $-1$ as $\delta\downarrow0$. Then, as $\delta\downarrow0$,
\begin{enumerate}[label=(\alph*)]
\item $s$-channel, $2+1$ fermions at $\bK=\bpi$, $\gamma=\gamma_c(\bpi)=1$:
\[
\nu(\delta)+1=\tfrac12\bigl(b_0(\delta)^{-1}-\delta\bigr)=\frac{\pi}{\ln(16/\delta)}\bigl(1+o(1)\bigr);
\]
\item $p$-channel (multiplicity two), $\bK=\mathbf0$, $\gamma=\gamma_c(\mathbf0)=\pi/(\pi-2)$:
\[
\nu(\delta)+1=\frac{\gamma_c}{2\pi}\,\delta\Bigl(\ln\frac1\delta+\kappa_0\Bigr)+o(\delta);
\]
\item $d$-channel, $\bK=\bpi$, $\gamma=\gamma_c^{(2)}(\bpi)=\pi/(4-\pi)$:
\[
\nu(\delta)+1=\frac\gamma2\Bigl(1-\frac2\pi\Bigr)\delta+o(\delta).
\]
\end{enumerate}
In each case $\nu(\delta)>-1$ for all $\delta>0$, in accordance with
Theorem~\ref{thm:count}(ii).
\end{proposition}

\begin{proof}
At $\bK\in\{\mathbf0,\bpi\}$ the matrices $M(\delta)$ and
$S_\bK=\pm S_{\mathbf0}$ are simultaneously diagonal in the $s,p,d$ basis
of Lemma~\ref{lem:M}, $S_{\mathbf0}$ acting as $+1,-1,+1$ on $s,p,d$;
hence the eigenvalues of $M^{1/2}CM^{1/2}$ with
$C=\frac\gamma2S_\bK$ are $\pm\frac\gamma2m_\bullet(\delta)$. (a)
$\nu=-\frac12m_s(\delta)$, and
$1-\frac12m_s=\frac12(b_0^{-1}-\delta)$, which is positive for small
$\delta$ and has the stated asymptotics by
Lemma~\ref{lem:constants}(ii). (b) $\nu=-\frac{\gamma_c}2A(\delta)$ and
$\frac{\gamma_c}2A=1$, so $\nu+1=\frac{\gamma_c}2(A-A(\delta))$; apply
Lemma~\ref{lem:constants}(v). (c) $\nu=-\frac\gamma2m_d(\delta)$ with
$\frac\gamma2m_d(0)=1$; apply the expansion of $m_d$ in
Lemma~\ref{lem:M}.
\end{proof}

The three rates, $1/\ln\frac1\delta$, $\delta\ln\frac1\delta$ and
$\delta$, are those of the $s$-, $p$- and $d$-channels; they are shown in
\autoref{fig:critical}. They are the mechanism behind the different
onset laws of Section~\ref{sec:critical}.

\begin{figure}[t]\centering
\includegraphics[width=\textwidth]{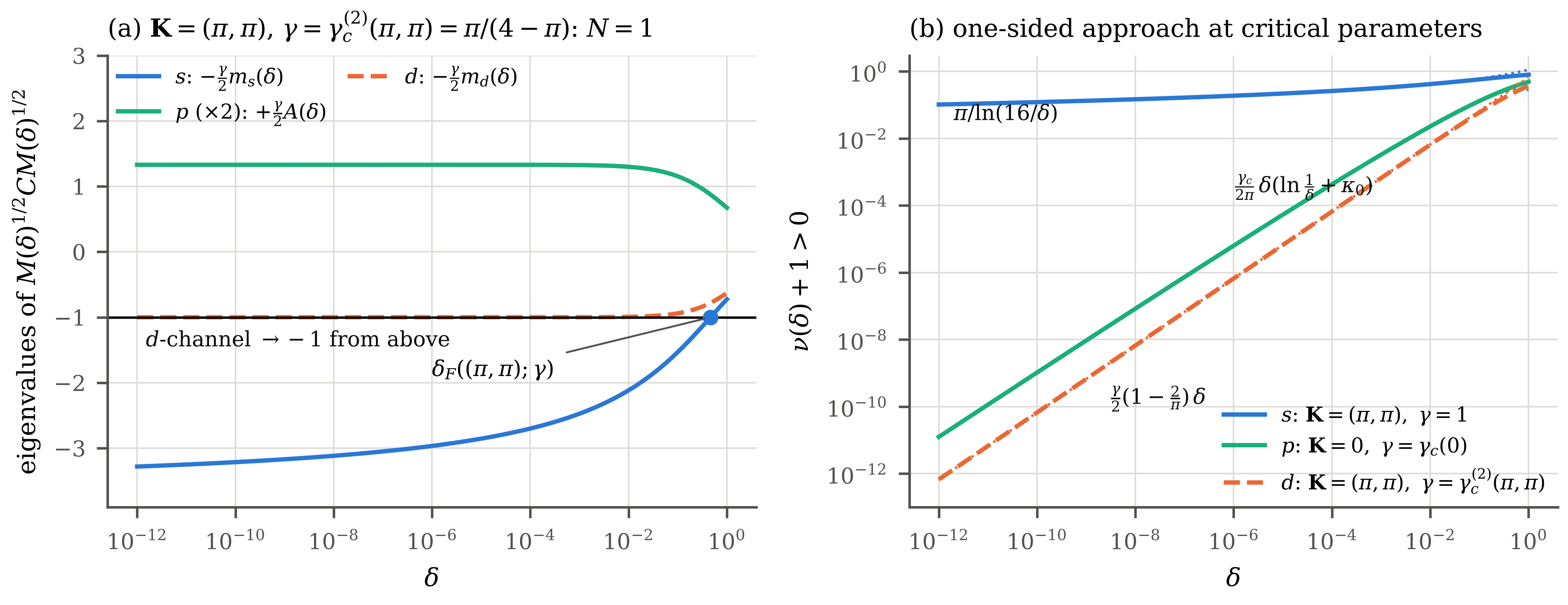}
\caption{Counting at critical parameters (Theorem~\ref{thm:count}(ii),
Proposition~\ref{prop:onesided}), from the exact infinite-lattice matrix
$M(\delta)$. \textbf{(a)} $2+1$ fermions at $\bK=\bpi$ and
$\gamma=\gamma_c^{(2)}(\bpi)=\pi/(4-\pi)$, a point of the critical set
inside the one-state region: the $s$-eigenvalue crosses $-1$ at the
binding energy $\delta_F(\bpi;\gamma)$ of the first state, while the
$d$-eigenvalue tends to $-1$ from above and never crosses it, so that
$N=1$. \textbf{(b)} $\nu(\delta)+1>0$ for the three critical channels
on a logarithmic scale, with the exact asymptotics of
Proposition~\ref{prop:onesided} shown as dotted lines: $\pi/\ln(16/\delta)$ ($s$),
$\frac{\gamma_c}{2\pi}\delta(\ln\frac1\delta+\kappa_0)$ ($p$) and
$\frac\gamma2(1-\frac2\pi)\delta$ ($d$).}
\label{fig:critical}
\end{figure}

\section{Brillouin-zone structure and critical mass ratios}\label{sec:zones}

The counting criterion is algebraic: it reduces to the distribution of
eigenvalues of a $4\times4$ matrix $M_0S_\bK$ depending on
$(u,v)=(\cos K_1,\cos K_2)$.

\subsection{Characteristic polynomial}

\begin{lemma}[Quartic $P_\bK$]\label{lem:char}
The eigenvalues of $M_0 S_\bK$ are the roots of
\[
P_\bK(x)=x^4+2(A-1)(u+v)x^3
+\bigl[(3A^2-8A+4)uv+2A^2-4A\bigr]x^2
+A(A^2-6A+4)(u+v)x+4A^2(1-A).
\]
Special values:
\[
P_\bK(x)\big|_{\bK=0}=(x-2)\bigl(x-(8/\pi-2)\bigr)(x+A)^2,
\]
\[
P_\bK(x)\big|_{\bK=\bpi}=(x+2)\bigl(x+(8/\pi-2)\bigr)(x-A)^2,
\]
\[
P_\bK(x)\big|_{\bK=(\pi,0)}=(x^2-A^2)\bigl(x^2-4(4/\pi-1)\bigr),
\]
\[
P_\bK(x)\big|_{\bK=(\pi/2,\pi/2)}=(x^2-2A)\bigl(x^2-2A(1-A)\bigr).
\]

Moreover, identically in $A,u,v$,
\begin{equation}\label{eq:PA}
P_\bK(\pm A)=A^2(A-2)(3A-2)(1\pm u)(1\pm v),
\end{equation}
\begin{equation}\label{eq:Pd}
P_\bK\bigl(-2(1-A)\bigr)=2(A-2)(A-1)(3A-2)\bigl[(3A-2)(1+u)(1+v)+A(1-uv)\bigr],
\end{equation}
and $P_{\bK+\bpi}(x)=P_\bK(-x)$. For $A=2-4/\pi\in(\frac23,1)$ the right-hand
side of \eqref{eq:PA} is $\le0$ and that of \eqref{eq:Pd} is $\ge0$ on
$[-1,1]^2$, with equality in \eqref{eq:Pd} only at $u=v=-1$. Finally, on the line $v=1$,
\begin{equation}\label{eq:PprimeA}
\partial_xP_\bK(x)\big|_{x=-A}=A(1-u)\,(4-4A-A^2),\qquad
4-4A-A^2=\displaystyle\tfrac{8(4\pi-\pi^2-2)}{\pi^2}>0,
\end{equation}
and the same holds on $u=1$ with $u$ replaced by $v$ ($P_\bK$ is symmetric in $u,v$);
by $P_{\bK+\bpi}(x)=P_\bK(-x)$, $\partial_xP_\bK(A)|_{v=-1}=-A(1+u)(4-4A-A^2)$.
\end{lemma}

\begin{proof}
With $U=\mathrm{diag}(e^{-iK_1/2},e^{iK_1/2},e^{-iK_2/2},e^{iK_2/2})$
one has $US_\bK U^*=S_{\mathbf 0}$ (the $\bK$-independent swap
$\pm\be_j\leftrightarrow\mp\be_j$), so
$\det(x-M_0S_\bK)=\det(x-UM_0U^*S_{\mathbf0})$, and all $\bK$-dependence
sits in the phases of the off-diagonal entries of $UM_0U^*$. Expanding
this $4\times4$ determinant, the phases combine into $u=\cos K_1$,
$v=\cos K_2$, and $A+2B=2$ eliminates $B$; this gives $P_\bK$. The
expansion and the four special factorisations were checked as
polynomial identities in $A,u,v$ (SymPy, \S\ref{sec:S43}).
\end{proof}

\subsection{Bosonic zone}

\begin{theorem}[Atom--dimer zone: bosons]\label{thm:zoneB}
For three identical bosons, $h_B(\bK)$ has at most one eigenvalue
below its threshold $4$, and it has one (equivalently,
$N^{3B}(\bK)=1$, by Corollary~\ref{cor:Nsigma}) if and only if
\[
F_B(u,v):=\bigl[(1-A)^4-A^4\bigr]
+(A^3-6A^2+6A-2)(u+v)
+(3A-2)(A-2)uv<0.
\]
Numerically,
\[
F_B(u,v)
=
-0.2734008215
-0.42466\ldots(u+v)
-0.22954\ldots uv.
\]
In
particular:
\begin{enumerate}
\item[(i)] $\bK=\mathbf 0$ belongs to the zone, with binding energy
$\delta_B(\mathbf 0)=\delta_\infty=0.0354199892\ldots$, reproducing
the transcendental constant of \cite{Trimer2026} by an independent
route;
\item[(ii)] $\bK=(\pi,0)$ belongs to the zone, but only marginally:
$x_{\max}=2\sqrt{4/\pi-1}=1.04545$, and by
Theorem~\ref{thm:boundary} the binding is exponentially small,
$\delta_B(\pi,0)=1.4535\cdot10^{-15}$ (see Remark~\ref{rem:rate});
\item[(iii)] $\bK=\bpi$ does \emph{not} belong to the zone
($F_B=+0.3464$): no atom--dimer state exists at $\bK=\bpi$;
\item[(iv)] on the diagonal $\bK=(K,K)$ the boundary is
$\cos K^*=2-3\pi/4$, $K^*=0.61593\pi$;
\item[(v)] at $\bK=\mathbf 0$ the $s$-channel binds, the $d$-channel
does not, and the $p$-channel is repulsive. Consequently, for every
$\mu\ge46$ with $3\eta_B(\mu)<\delta_\infty$ (i.e.\ $\mu>3060$), the
three-body operator $H_\mu(\mathbf0)$ has exactly two eigenvalues below
$-\mu+4-2\eta_B$: the trimer below $-\mu$ and the atom--dimer state in
$-\mu+4-\delta_\infty+[-\eta_B,\eta_B]$ (Theorem~\ref{thm:feshbach}(b),(c)
with $s=4-2\eta_B$).
\end{enumerate}
On the boundary $F_B(u,v)=0$ there is no bound state
(Theorem~\ref{thm:count}(ii)). For every $\bK$ with $F_B<0$ the
atom--dimer eigenvalue of $H_\mu(\bK)$ exists for all
$\mu>10+72/\delta_B(\bK)$, and uniformly on compact subsets of the open
zone (Corollary~\ref{cor:Nsigma}).
\end{theorem}

\begin{proof}
By Theorem~\ref{thm:count}(ii) with $C=-S_\bK$, the number of bound
states is the number of eigenvalues of $M_0^{1/2}S_\bK M_0^{1/2}$, i.e.\ of
roots of $P_\bK$ (Lemma~\ref{lem:char}), that are \emph{strictly}
greater than $1$; and $F_B=P_\bK(1)$ identically (\S\ref{sec:S43}). The numerical values in (i)--(iv)
follow by substituting the explicit $A=2-4/\pi$. That $P_\bK(1)<0$ is equivalent to $x_{\max}(\bK)>1$ uses that at most
one root of $P_\bK$ exceeds $1$ for every $\bK$: since $M_0>0$ and $S_\bK$ has signature $(2,2)$, $P_\bK$ has exactly
two positive roots $x_{\max}\ge x_2^+$; and since
$P_\bK(0)=4A^2(1-A)>0$ (as $A=2-4/\pi\in(2/3,1)$) while
$P_\bK(A)\le0$ by \eqref{eq:PA}, the smaller positive root satisfies
$x_2^+\le A=2-4/\pi<1$, with equality iff $K_1=\pi$ or $K_2=\pi$ (there $P_\bK(A)=0$, and by \eqref{eq:PprimeA} the other positive root lies strictly above $A$, so the root $A$ is $x_2^+$). For (iv), one has exactly
\[
F_B(u,u)
=
\bigl((A-2)u-2A+1\bigr)
\bigl((3A-2)u+2A^2-2A+1\bigr),
\] and the second factor has no zero in $[-1,1]$. Case (v): at~$\bK=\mathbf0$ the
eigenvalues of $M_0S_{\mathbf0}$ are $2$ ($s$), $8/\pi-2\approx0.546$
($d$) and $-A,-A$ ($p$); only $2>1$. With
$h_B(\mathbf0)$ having the single eigenvalue $4-\delta_\infty$ below $4$,
Theorem~\ref{thm:feshbach}(c) at $s=4-2\eta_B$ gives
$n_h(4-3\eta_B)\le n_H(s)\le n_h(4-\eta_B)$, and both bounds equal $1$
when $3\eta_B<\delta_\infty$; Theorem~\ref{thm:feshbach}(b) adds the
trimer.
\end{proof}

\subsection{Fermionic zone and closed-form critical mass ratios}

\begin{theorem}[Critical mass ratios]\label{thm:zoneF}
For the $2+1$ fermionic configuration with mass ratio $\gamma$, the
operator $h_F(\bK;\gamma)$ has an eigenvalue below its threshold
$2+2\gamma$ (equivalently $N^{2+1F}(\bK)\ge1$, Corollary~\ref{cor:Nsigma})
if and only if $\gamma>\gamma_c(\bK)=2/|x_{\min}(\bK)|$, where
$x_{\min}(\bK)$ is the smallest root of $P_\bK$. The number of bound
states equals the number of roots strictly less than $-2/\gamma$; in
particular there is none at $\gamma=\gamma_c(\bK)$.
The closed-form values are
\[
\boxed{\begin{gathered}
\gamma_c(\mathbf 0)=\frac{\pi}{\pi-2}\approx2.7519384,\qquad
\gamma_c(\bpi)=1,\\
\gamma_c(\pi,0)=\displaystyle\sqrt{\frac{\pi}{4-\pi}}\approx1.9130584,\qquad
\gamma_c\bigl(\tfrac\pi2,\tfrac\pi2\bigr)=\sqrt{\frac{\pi}{\pi-2}}\approx1.6588967.
\end{gathered}}
\]
Globally $\min_\bK\gamma_c=1$ (at $\bK=\bpi$),
$\max_\bK\gamma_c=\pi/(\pi-2)$ , attained only at $\bK=\mathbf 0$. Hence:
\begin{enumerate}
\item[(i)] for $\gamma\le1$ there is no atom--dimer state for any
$\bK$;
\item[(ii)] for $1<\gamma<\pi/(\pi-2)$ the zone
$\{\bK:\gamma_c(\bK)<\gamma\}$ is a non-empty open proper subset of
$\T^2$ containing $\bpi$ and not containing $\mathbf0$;
\item[(iii)] for $\gamma>\pi/(\pi-2)$ a state exists for every
$\bK$, and there is $\mu_0(\gamma)$ such that for all $\mu\ge\mu_0(\gamma)$
the three-body operator $H^{2+1F}_\mu(\bK)$ has a discrete eigenvalue
below its essential spectrum for \emph{every} $\bK\in\T^2$.
\end{enumerate}
\end{theorem}

\begin{proof}
Theorem~\ref{thm:count}(ii) with $C=\tfrac{\gamma}{2}S_\bK$: the
eigenvalues of $M_0^{1/2}CM_0^{1/2}$ are $\frac\gamma2x_i$, $x_i$ the
roots of $P_\bK$, and the strict condition $\tfrac{\gamma}{2}x_i<-1$
reads $x_i<-2/\gamma$. Lemma~\ref{lem:char}
gives $x_{\min}$ at each $\bK$. At $\bK=\mathbf 0$: $x_{\min}=-A$ and
$\gamma_c=2/A=2/(2-4/\pi)=\pi/(\pi-2)$. At $\bK=\bpi$: $x_{\min}=-2$,
$\gamma_c=1$. At $\bK=(\pi,0)$: $x_{\min}=-2\sqrt{4/\pi-1}$,
$\gamma_c=\sqrt{\pi/(4-\pi)}$. At $\bK=(\pi/2,\pi/2)$:
$x_{\min}=-\sqrt{2A}$, $\gamma_c=\sqrt{2/A}=\sqrt{\pi/(\pi-2)}$.
The lower bound is exact: $|x|\le\|M_0^{1/2}S_\bK M_0^{1/2}\|\le
\|M_0\|\,\|S_\bK\|=2$ for every root, so~$\gamma_c\ge1$, with equality
at $\bK=\bpi$. The upper bound $\gamma_c\le\pi/(\pi-2)$, i.e.\ $x_{\min}(\bK)\le-A$, follows
from \eqref{eq:PA}: for $x<0$ the factor of $P_\bK$ coming from its two
positive roots is positive, so~$P_\bK(-A)\le0$ places $-A$ between the two
negative roots, $x_{\min}\le-A\le x_2$. It is attained only at $\bK=\mathbf0$ (Corollary~\ref{cor:twostate}; numerical confirmation: \autoref{fig:gammacmap}, \autoref{fig:zonepath}(b),
\S\ref{sec:S44}). Statements (i)--(iii) for $h_F$ follow from the bounds
$1\le\gamma_c\le\pi/(\pi-2)$ and the continuity of $x_{\min}(\bK)$ in
$\bK$. For the three-body statement in (iii): the lowest eigenvalue of
$h_F(\bK;\gamma)$ is a~min--max value of a norm-continuous family, hence
continuous in $\bK$, and it lies below $2+2\gamma$ for every $\bK$; so
the largest binding energy $\delta_1(\bK)$ has a positive minimum on the
compact $\T^2$, and Corollary~\ref{cor:Nsigma}(iii) applies with
$\mathcal K=\T^2$.
\end{proof}

\begin{remark}[Consistency with earlier work]\label{rem:consistency_companion}
In \cite{Companion} the ratio at $\bK=\mathbf0$ is defined by the integral
equation (5) there, $\gamma_c^{-1}=\langle(\sin q_1+\sin q_2)^2/(2\eps(\bq))\rangle$.
The cross term is odd and $\langle2\sin^2q_1/\eps\rangle=R(2\be_1)=A$, so~the
integral equals $A/2=(\pi-2)/\pi$ exactly: the value $2.75194\pm0.00001$ of~\cite{Companion} is $\pi/(\pi-2)$, and the doubly degenerate level found there is
the $p$-wave atom--dimer pair of Theorem~\ref{thm:crit}(b). Two statements of
\cite{Companion} are refined: (i) at $\lambda=\infty$ the onset is
$\delta_F\simeq(2\pi/\gamma_c^2)\,h/\ln(1/h)$, $h=\gamma-\gamma_c$, not a pure
power law $Ch^\nu$; (ii) $\gamma>\gamma_c(\mathbf0)$ is the binding criterion at
$\bK=\mathbf0$ only: at~$\bK\ne\mathbf0$ binding requires $\gamma>\gamma_c(\bK)$,
and globally already $\gamma>1$ suffices (Theorem~\ref{thm:zoneF}).
\end{remark}

\begin{corollary}[Two-state structure of the fermionic zones]\label{cor:twostate}
For every $\bK\in\T^2$ the quartic $P_\bK$ has exactly two negative
roots $x_{\min}(\bK)\le x_2(\bK)<0$. Consequently
$N_F(\bK;\gamma)\le2$ for all $\bK,\gamma$, and the second state exists
iff $\gamma>\gamma_c^{(2)}(\bK):=2/|x_2(\bK)|$ (strictly: at
$\gamma=\gamma_c^{(2)}(\bK)$ the second channel does not bind,
Theorem~\ref{thm:count}(ii) and Remark~\ref{rem:critset}). On the lines $K_1=0$ or $K_2=0$ one has $x_2\equiv-A$ exactly, i.e.~$\gamma_c^{(2)}\equiv\pi/(\pi-2)$ (a transverse $p$-mode that decouples
there); off these lines $P_\bK(-A)<0$ by \eqref{eq:PA}, so~they are
exactly the set where $\gamma_c^{(2)}=\pi/(\pi-2)$;
on them $x_{\min}(\bK)<-A$ strictly except at $\bK=\mathbf 0$, so~the maximum of
$\gamma_c$ is attained only at~$\bK=\mathbf 0$; at~$\bK=\bpi$, $\gamma_c^{(2)}=2/(8/\pi-2)=\pi/(4-\pi)\approx3.6598$.
Moreover $\pi/(\pi-2)\le\gamma_c^{(2)}(\bK)\le\pi/(4-\pi)$ for every $\bK$, the
upper bound being attained only at $\bpi$; so for $\gamma>\pi/(4-\pi)$ there are exactly two fermionic
atom--dimer states at every $\bK$.
\end{corollary}

\begin{proof}
$M_0>0$, so $M_0^{1/2}S_\bK M_0^{1/2}$ is congruent to $S_\bK$, whose
eigenvalues are $+1,+1,-1,-1$ (two $2\times2$ blocks
$\bigl(\begin{smallmatrix}0&e^{-iK_j}\\e^{iK_j}&0\end{smallmatrix}\bigr)$);
Sylvester's law gives exactly two negative eigenvalues, which are the
negative roots of $P_\bK$. For $v=\cos K_2=1$, $P_\bK(-A)=0$ identically
in $u$ (computer algebra; the~eigenvector is $(0,0,1,-1)$, on which
$S_\bK$ acts as $-1$ and $M_0$ as $A$). Writing $P_\bK=(x+A)Q$ on $v=1$: for $x<0$ the factor of $Q$ coming from the two
positive roots is positive, and $Q(-A)=\partial_xP_\bK(-A)>0$ for $u<1$ by
\eqref{eq:PprimeA}; hence the other negative root lies strictly below $-A$, i.e.\
$x_{\min}<-A=x_2$. Off the lines $u=1$, $v=1$ one has $P_\bK(-A)<0$ by \eqref{eq:PA},
so $x_{\min}<-A<x_2$ there. The lower bound is $x_2\ge-A$ (proof of Theorem~\ref{thm:zoneF}). For the upper
bound, $-2(1-A)=-(8/\pi-2)\in(-A,0)$ and $P_\bK<0$ on $(x_{\min},x_2)$; since
$P_\bK(-2(1-A))\ge0$ by \eqref{eq:Pd}, $x_2\le-2(1-A)$, i.e.\
$\gamma_c^{(2)}\le1/(1-A)=\pi/(4-\pi)$, with equality only at $\bK=\bpi$.
The $601\times601$ scan of \S\ref{sec:S44} confirms both bounds.
\end{proof}

The resulting fermionic zone structure is shown in
\autoref{fig:fphase} along the high-symmetry path and in
\autoref{fig:zonemaps}(b)--(f) over the whole Brillouin zone: the zone
opens at $M=\bpi$ at $\gamma=1$, reaches $X$ at
$\gamma=\sqrt{\pi/(4-\pi)}$, reaches $\Gamma$ (and thus covers $\T^2$)
at $\gamma=\pi/(\pi-2)$,
and a second band appears first on the $\Gamma$--$X$ lines (hatched in
\autoref{fig:zonemaps}(e)).

\begin{figure}[t]\centering
\figinclude{0.94\textwidth}{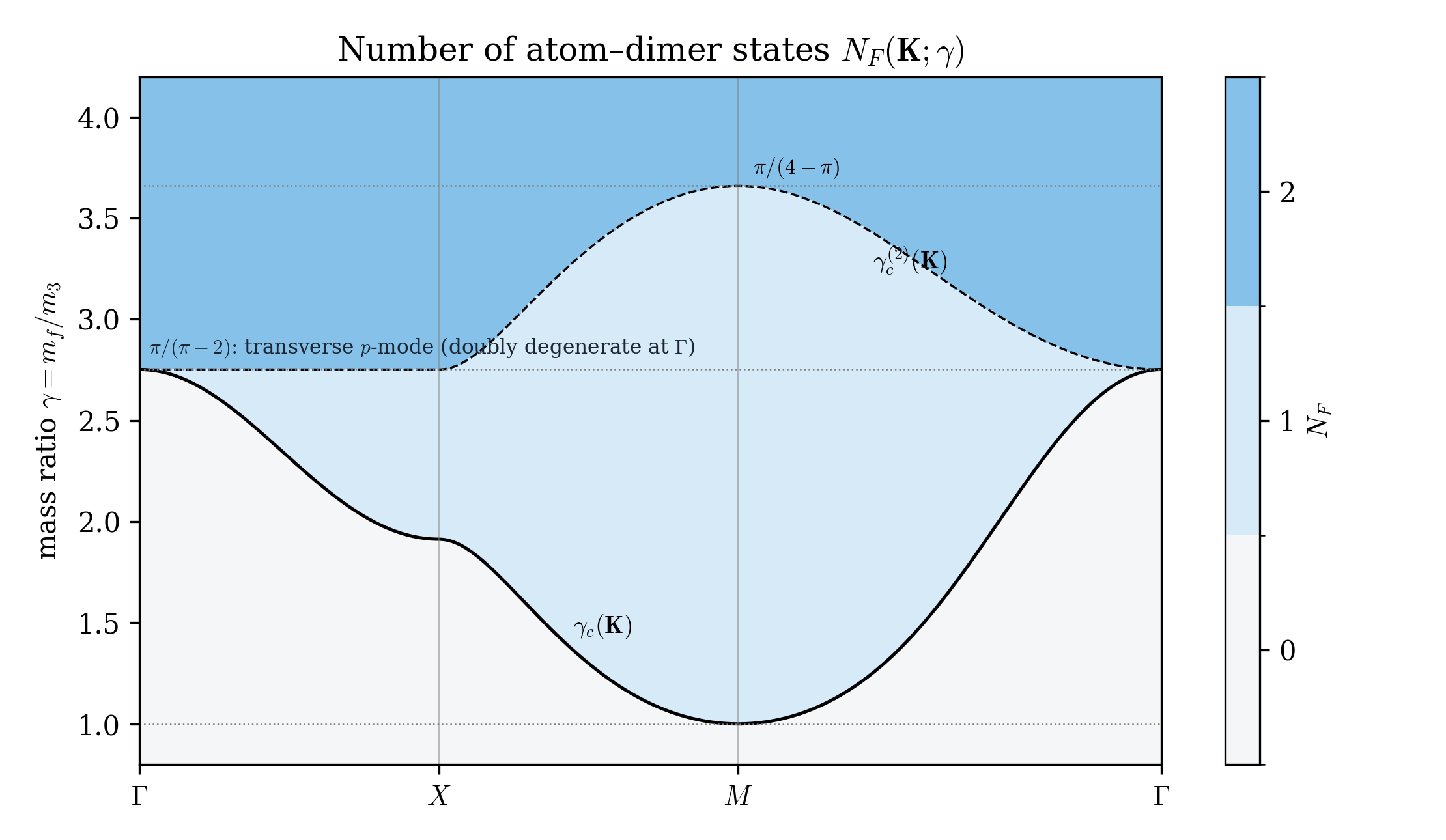}
\caption{Number $N_F(\bK;\gamma)$ of $2+1$ fermionic atom--dimer states
along $\Gamma$--$X$--$M$--$\Gamma$ as a function of the mass ratio
$\gamma=m_f/m_3$, from the counting criterion (Theorem~\ref{thm:count})
with $C=\frac\gamma2S_\bK$. Solid line: the critical mass ratio
$\gamma_c(\bK)$ of Theorem~\ref{thm:zoneF}; dashed line: the second
threshold $\gamma_c^{(2)}(\bK)$ of Corollary~\ref{cor:twostate}, flat
at $\pi/(\pi-2)$ on $\Gamma$--$X$ and maximal, $\pi/(4-\pi)$, at $M$.
$N_F\le2$ everywhere.}
\label{fig:fphase}
\end{figure}

\begin{figure}[t]\centering
\figinclude{\textwidth}{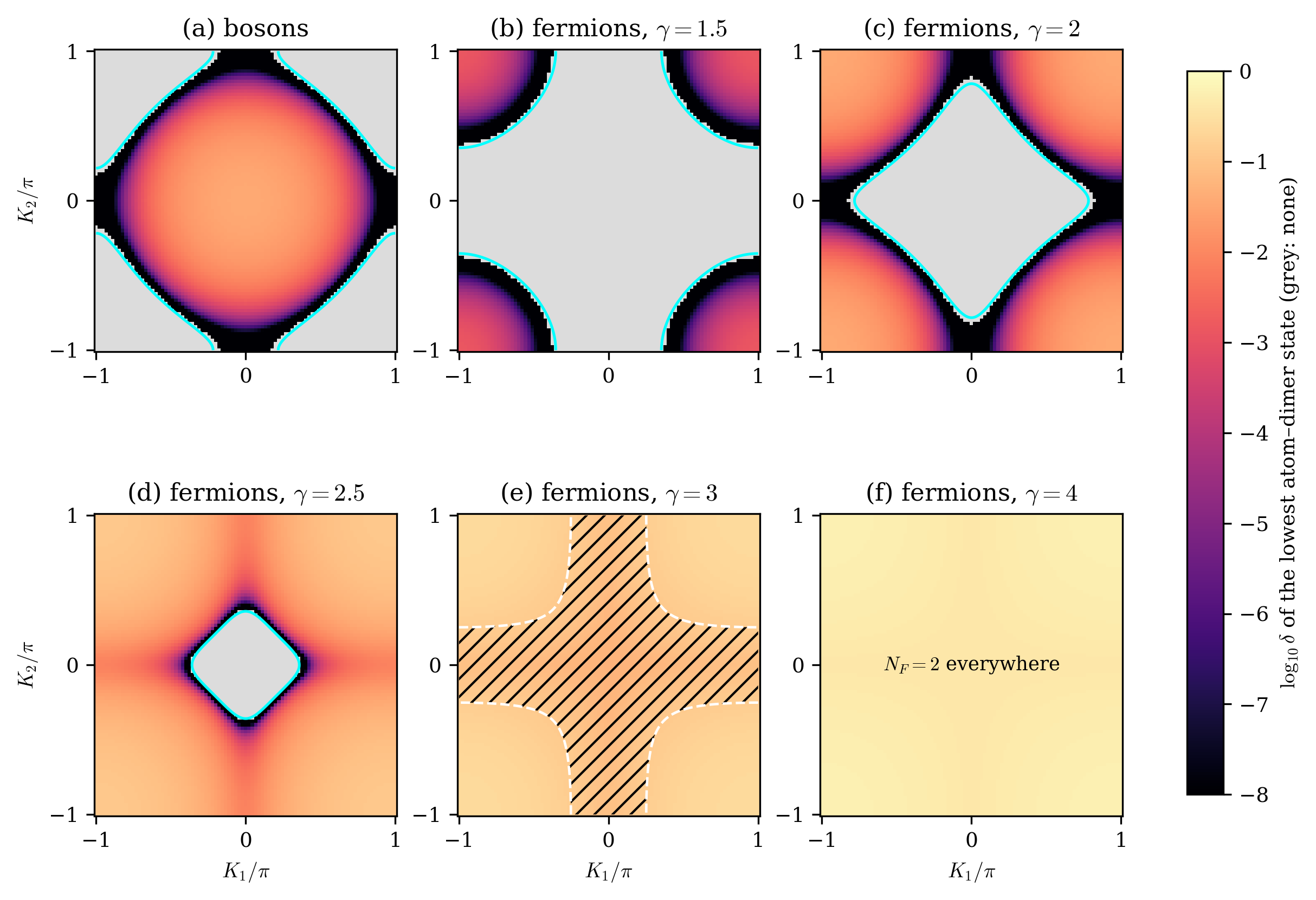}
\caption{Brillouin-zone maps of the binding energy of the lowest
atom--dimer state, $\log_{10}\delta$, computed in the infinite lattice
from the exact threshold matrix $M(\delta)$ of Lemma~\ref{lem:M}
($91\times91$ grid; grey: no state). \textbf{(a)} Bosons, $h_B$.
\textbf{(b)--(f)} $2+1$ fermions, $h_F$, at $\gamma=1.5,2,2.5,3,4$.
Cyan: exact zone boundary from $P_\bK$ ($P_\bK(1)=0$ for bosons,
$\gamma=\gamma_c(\bK)$ for fermions); white dashed and hatched: region
with a second fermionic state, $\gamma>\gamma_c^{(2)}(\bK)$. Panel (c)
is panel (a) translated by $\bpi$ up to a constant shift in energy, as required by Theorem~\ref{thm:dual} (the duality gives $h_F(\bK+\bpi;\gamma=2) = h_B(\bK)+2$, so binding energies coincide but thresholds differ by $2$). The black
band inside each boundary, where $\delta<10^{-8}$, is the essential
singularity of Theorem~\ref{thm:boundary}.}
\label{fig:zonemaps}
\end{figure}

\section{Critical behaviour at the threshold}\label{sec:critical}

The onset of binding near $\gamma=\gamma_c(\bK)$ has different
qualitative character in the $s$-wave and $p$-wave channels.

\begin{theorem}[Critical onset]\label{thm:crit}
(a) \emph{$s$-wave at $\bK=\bpi$.} As $\gamma\downarrow1$,
\begin{equation}\label{eq:onset-s}
\delta_F(\bpi,\gamma)=16\,\exp\!\Bigl(-\,\frac{\pi\gamma}{\gamma-1}\Bigr)
\Bigl(1+O\bigl((\gamma-1)^{-2}e^{-\pi/(\gamma-1)}\bigr)\Bigr),
\end{equation}
i.e.\ the binding energy vanishes faster than any power of
$(\gamma-1)$: an onset of \emph{infinite order}.

(b) \emph{$p$-wave at $\bK=\mathbf 0$.} As $\delta\downarrow0$,
$A(\delta)=A-\tfrac{\delta}\pi\bigl(\ln\tfrac1\delta+\kappa_0\bigr)+o(\delta)$
with the constant $\kappa_0=1+4\ln2-\pi=0.6309960686\ldots$ of
Lemma~\ref{lem:constants}(v). As $\gamma$ decreases to
$\gamma_c=\pi/(\pi-2)$, the (doubly degenerate) binding energy satisfies
\begin{equation}\label{eq:onset-p}
\delta_F\Bigl(\ln\frac1{\delta_F}+\kappa_0\Bigr)
=\frac{2\pi(\gamma-\gamma_c)}{\gamma\gamma_c}\bigl(1+o(1)\bigr),
\end{equation}
hence
\[
\delta_F=\frac{2\pi}{\gamma_c^2}\,
\frac{\gamma-\gamma_c}{\ln\bigl(1/(\gamma-\gamma_c)\bigr)}\bigl(1+o(1)\bigr),
\]
where $2\pi/\gamma_c^2=2(\pi-2)^2/\pi\approx0.8297$. The onset is linear up to a logarithm, the two-dimensional $p$-wave ($\ell=1$)
threshold law; in particular it is neither the square-root law ($\nu=\tfrac12$)
anticipated in~\cite{Companion} nor the quadratic law of Fig.~1 of~\cite{Companion}.
\end{theorem}

\begin{proof}[Proof sketch; details in Supplementary Material, \S\ref{sec:S23}]
(a) At $\bK=\bpi$ the $s$-vector $\mathbf1$ is an eigenvector of
$S_{\bpi}$ with eigenvalue $-1$ and of $M(\delta)$ with eigenvalue
$m_s(\delta)$ (Lemma~\ref{lem:M}); by Theorem~\ref{thm:count}(i) the
binding energies in this channel are the roots of
$\tfrac\gamma2m_s(\delta)=1$, i.e.\ $1/b_0(\delta)=g_*+\delta$ with
$g_*:=2(\gamma-1)/\gamma$. Since $m_s$ is strictly decreasing from
$m_s(0^+)=2$, there is exactly one root iff $\gamma>1$. With
$2\pi b_0=\ln\frac{16}\delta+O(\delta\ln\frac1\delta)$
(Lemma~\ref{lem:constants}(ii)) this gives
$\ln\frac{16}\delta=\frac{2\pi}{g_*}+O\bigl(\frac{\delta}{(\gamma-1)^2}\bigr)$,
because $\ln\frac1\delta=O(\frac1{\gamma-1})$; since
$\frac{2\pi}{g_*}=\frac{\pi\gamma}{\gamma-1}$ and
$\delta=O(e^{-\pi/(\gamma-1)})$, \eqref{eq:onset-s} follows.
(b) The $p$-vectors are eigenvectors of $S_{\mathbf0}$ with eigenvalue
$-1$ and of $M(\delta)$ with eigenvalue $A(\delta)$, so the condition is
$A(\delta)=2/\gamma$, with multiplicity two; since $A(\delta)$ is
strictly decreasing from $A=2/\gamma_c$, there is a root iff
$\gamma>\gamma_c$, and it is unique. The expansion of $A(\delta)$ is
Lemma~\ref{lem:constants}(v), proved in \S\ref{sec:S25} with the closed
form of $\kappa_0$, and $A-2/\gamma=2(\gamma-\gamma_c)/(\gamma\gamma_c)$
gives the implicit law in \eqref{eq:onset-p}. Inverting it,
$\ln\frac1{\delta_F}=\ln\frac1h+O(\ln\ln\frac1h)$ with $h=\gamma-\gamma_c$,
which gives the explicit form with relative error
$O(\ln\ln\frac1h/\ln\frac1h)$.
\end{proof}

Table~\ref{tab:crit} reports an independent numerical verification of
the two asymptotic laws \eqref{eq:onset-s} and \eqref{eq:onset-p}: on~the left the exact roots of $\tfrac\gamma2m_s(\delta)=1$ are compared
with $16\displaystyle\,e^{-\pi\gamma/(\gamma-1)}$, on the right the ratios
$\rho_{\rm impl}$ and $\rho_{\rm expl}$ quantify the convergence of the
implicit and explicit leading-log laws for the $p$-wave onset.

\begin{table}[h]\centering\small
\caption{Numerical check of Theorem~\ref{thm:crit}. Left: $s$-wave at
$\bK=\bpi$, exact root of $\tfrac\gamma2m_s(\delta)=1$ with $b_0$ in
closed elliptic form (40 digits). Right: $p$-wave at $\bK=\mathbf0$;
$\rho_{\rm impl}$ is the ratio of the two sides of the implicit law in
\eqref{eq:onset-p}, $\rho_{\rm expl}$ that of the explicit leading-log
law, which converges only as $1-O(\ln\ln\frac1h/\ln\frac1h)$,
$h=\gamma-\gamma_c$.}
\label{tab:crit}
\begin{tabular}{cccc}
\toprule
$\gamma$ & $\delta_F(\bpi,\gamma)$ & $16\,e^{-\pi\gamma/(\gamma-1)}$ & ratio \\
\midrule
1.50 & $1.31169\cdot10^{-3}$ & $1.29119\cdot10^{-3}$ & 1.015875 \\
1.30 & $1.95905\cdot10^{-5}$ & $1.95804\cdot10^{-5}$ & 1.000516 \\
1.20 & $1.04199\cdot10^{-7}$ & $1.04199\cdot10^{-7}$ & 1.0000054 \\
1.15 & $5.54500\cdot10^{-10}$ & $5.54500\cdot10^{-10}$ & 1.00000005 \\
\bottomrule
\end{tabular}
\quad
\begin{tabular}{cccc}
\toprule
$h$ & $\delta_F(\mathbf0,\gamma_c+h)$ & $\rho_{\rm impl}$ & $\rho_{\rm expl}$ \\
\midrule
$10^{-2}$ & $1.1120\cdot10^{-3}$ & 0.9998 & 0.617 \\
$10^{-3}$ & $8.2674\cdot10^{-5}$ & 1.0000 & 0.688 \\
$10^{-4}$ & $6.6061\cdot10^{-6}$ & 1.0000 & 0.733 \\
$10^{-5}$ & $5.5158\cdot10^{-7}$ & 1.0000 & 0.765 \\
\bottomrule
\end{tabular}
\end{table}

The mechanism behind (a) is general. Put $\xi(\bK):=x_{\max}(\bK)$ for
bosons and $\xi(\bK):=\tfrac\gamma2|x_{\min}(\bK)|$ for fermions, so
that the zone is $\{\xi>1\}$, and let $v_\bK$ be the corresponding
eigenvector of $M_0S_\bK$ (resp.\ $M_0(-S_\bK)$). Define
\[
c(\bK):=\xi(\bK)\,
\frac{|\langle\mathbf1,M_0^{-1}v_\bK\rangle|^2}{\langle v_\bK,M_0^{-1}v_\bK\rangle}\ \ge0 .
\]

\begin{theorem}[Essential singularity at the zone boundary]\label{thm:boundary}
Let $\bK_0$ lie on the zone boundary, $\xi(\bK_0)=1$, with
$c(\bK_0)>0$ and $\xi(\bK_0)$ a simple eigenvalue. Then, as
$\bK\to\bK_0$ from inside the zone,
\[
\ln\delta(\bK)=-\frac{\pi\,c(\bK)}{2\,(\xi(\bK)-1)}+O(1).
\]
In particular $\delta$ vanishes faster than any power of the distance
to the boundary. Here $\delta(\bK)$ is the binding energy of the state
whose matrix eigenvalue is $\xi(\bK)$. If $c(\bK_0)=0$ (threshold mode
orthogonal to the $s$-wave, as for the $p$-modes at $\bK=\mathbf0$) the
marginal $s$-wave mechanism is absent and the onset is governed by the
regular $O(\delta\ln\frac1\delta)$ ($p$) or $O(\delta)$ ($d$) terms of
Lemma~\ref{lem:M}; this case is treated here only at the symmetric
points (Theorem~\ref{thm:crit}(b), Proposition~\ref{prop:onesided}).
\end{theorem}

\begin{proof}[Proof sketch; details in Supplementary Material, \S\ref{sec:S24}]
By \eqref{eq:Mlimit}, up to $O(\delta\ln\frac1\delta)$ the binding
condition is that $M_\epsilon:=M_0-\epsilon\mathbf1\mathbf1^T$, with
$\epsilon=1/(4b_0)=\pi/(2\ln(16/\delta))$, has $\xi_\epsilon=1$. Writing
$M_\epsilon Sv=\xi v$ as the Hermitian pencil $Sv=\xi M_\epsilon^{-1}v$
($M_\epsilon>0$), the Hellmann--Feynman formula gives
$\partial_\epsilon\xi|_{\epsilon=0}=-c(\bK)$. Hence the root is
$\epsilon^*=(\xi-1)/c+O((\xi-1)^2)$ and
$\ln(16/\delta)=\pi/(2\epsilon^*)=\pi c/(2(\xi-1))+O(1)$.
\end{proof}

\begin{remark}\label{rem:boundary-numbers}
At $\bK=\bpi$ for fermions $v=\mathbf1$ and
Theorem~\ref{thm:boundary} reduces to \eqref{eq:onset-s}. For bosons at
$\bK=(\pi,0)$: $\xi=1.045446$, $c=1.0454$, and the exact root (with $b_0$
in elliptic form) is $\ln(16/\delta)=36.937$, i.e.\
$\delta_B(\pi,0)=1.4535\cdot10^{-15}$; the leading term
$\pi c/(2(\xi-1))=36.13$ is within the stated $O(1)$. At~$\bK=(\tfrac\pi2,\tfrac\pi2)$: $\xi=1.20562$ and $\delta_B=6.5049\cdot10^{-4}$.
\end{remark}

\begin{remark}[Renormalization-group interpretation; not used in the proofs]\label{rem:RG}
The threshold asymptotics admit a renorma\-li\-za\-tion-group interpretation, which we record
because it identifies the mechanism behind the essential singularity and places it in a
wider methodological context. Items (a)--(d) are interpretive: the
mathematical content of each is a statement already proved above, which
is indicated; the comparisons with other physical mechanisms quote the
cited works and are not used anywhere.

\emph{(a) The reduction is an exact decimation step.} The Schur
complement $F(z)$ of Theorem~\ref{thm:feshbach} is the elementary step
of the operator-theoretic RG of Bach, Fr\"ohlich and Sigal \cite{BFS}:
the modes in $\Ran\Pi_0\oplus\Ran\Pi_3$, at spectral distance at least
$\mu-10$, are integrated out exactly, in the spirit of Wilson's
decimation \cite{WilsonKogut} and of the Schrieffer--Wolff
transformation \cite{SchriefferWolff}.

\emph{(b) The $s$-wave coupling is marginal; binding is dimensional
transmutation.} Put $\ell=\ln(16/\delta)$ and $g(\ell)=1/b_0(\delta)$.
By Lemma~\ref{lem:M}, $g=2\pi/\ell+O(\delta\ln\frac1\delta)$, hence
\[
\beta(g):=\frac{dg}{d\ell}=-\frac{g^2}{2\pi}+O(\delta\ln\tfrac1\delta),
\]
the one-loop flow of a marginal coupling, as for a contact interaction
in two dimensions, where scale invariance is broken by a quantum
anomaly \cite{PitaevskiiRosch,Olshanii2010}. The binding condition of
Theorem~\ref{thm:crit}(a), $\frac\gamma2m_s(\delta)=1$, reads
$g(\ell)=g_*+O(\delta)$ with the bare coupling
$g_*=2(\gamma-1)/\gamma$: the binding energy is the scale at which the
running coupling reaches its bare value, and
\[
\delta=16\,e^{-2\pi/g_*}\bigl(1+o(1)\bigr),
\]
which is \eqref{eq:onset-s}. This is dimensional transmutation of the
same type as the weak-coupling bound state in two dimensions
\cite{Simon1976}, the BCS gap and the Kondo scale; the number $16$ is
the lattice scale parameter (the analogue of $\Lambda$ in a given
scheme), fixed by $\Z^2$ rather than by a scattering length. At~a
general boundary point Theorem~\ref{thm:boundary} states the same with
$g_*=4(\xi(\bK)-1)/c(\bK)+O((\xi-1)^2)$.

\emph{(c) Channels are classified by relevance.} In the $p$- and
$d$-channels the threshold matrix elements have finite limits with corrections $O(\delta\ln\frac1\delta)$ ($p$) and
$O(\delta)$ ($d$); no marginal coupling is present, and the onset is linear up to
a logarithm in the $p$-channel (Theorem~\ref{thm:crit}(b)) and linear in the
$d$-channel (e.g.\ the second fermionic state at $\bpi$ as $\gamma\downarrow\pi/(4-\pi)$).
 The constant $c(\bK)$ of
Theorem~\ref{thm:boundary} is the weight of the unique marginal
direction $\mathbf 1$ in the threshold mode. Thus the onset law is
decided by one structural datum, the~overlap with the marginal
operator, and not by the details of the dispersion.

\emph{(d) The zone boundary is not of BKT type.} Kaplan, Lee, Son and
Stephanov \cite{KLSS2009} showed that loss of conformality by
annihilation of an ultraviolet and an infrared fixed point produces
the Berezinskii--Kosterlitz--Thouless scaling
$\ln(\Lambda_{\rm IR}/\Lambda_{\rm UV})\sim-c/\sqrt{\alpha-\alpha_*}$;
according to the cited literature, the onset of the Efimov effect for
the $3$D $2+1$ fermionic problem at $M/m=13.607$
\cite{Efimov1973,Petrov2003}, where two real scaling exponents merge and
become complex, is of this kind; we quote this only as a point of
comparison, and no property of the Efimov effect is used or claimed
here. In the present model, by contrast, $\ln\delta\sim-c/(\xi-1)$ with a
\emph{linear} denominator: a single marginal coupling whose bare value
passes through zero, with no fixed-point merger. The two classes are
distinguished directly by the data of Table~\ref{tab:crit} and
\autoref{fig:fonset}(a), where $\ln\delta_F$ is linear in
$1/(\gamma-1)$ over more than fifteen decades of $\delta_F$.
\end{remark}

Both onset laws are displayed in \autoref{fig:fonset} over many decades.

\begin{figure}[t]\centering
\figinclude{\textwidth}{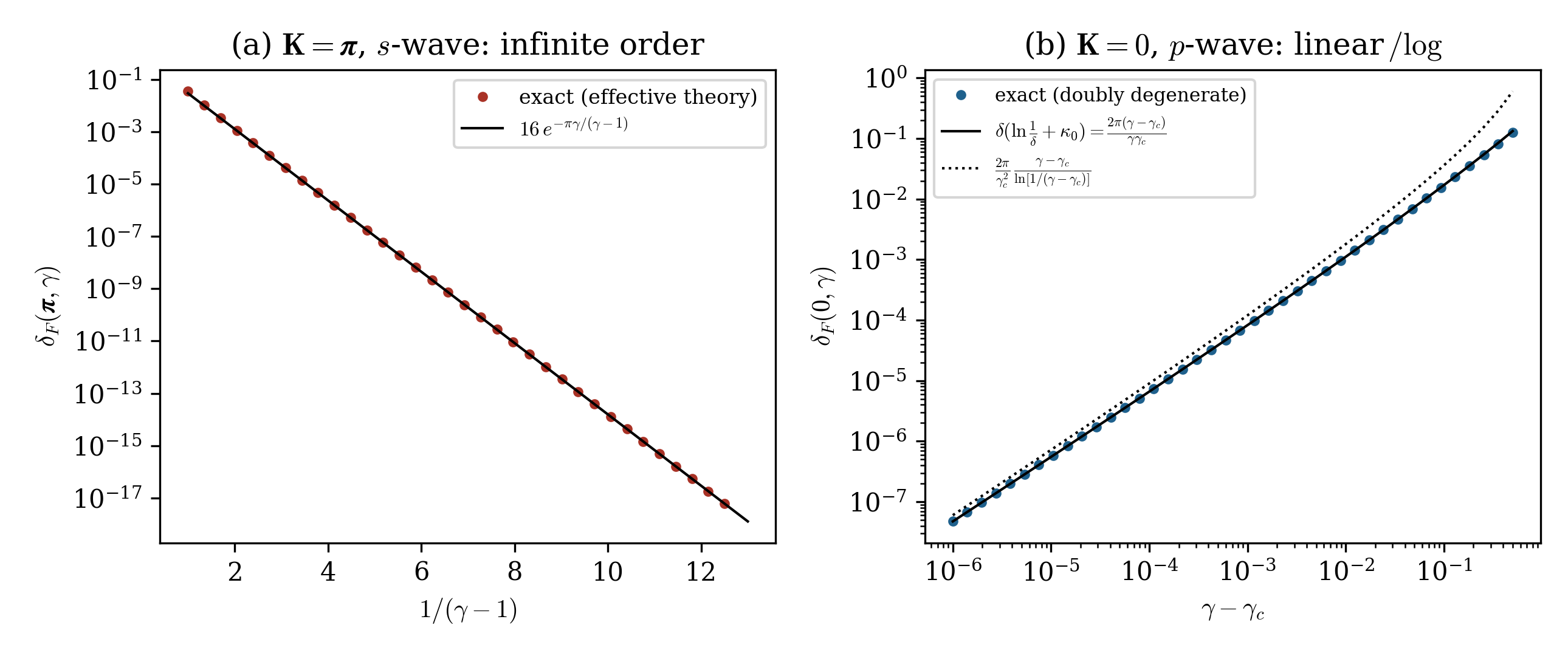}
\caption{Critical onset of the $2+1$ fermionic atom--dimer state
(Theorem~\ref{thm:crit}); dots are exact infinite-lattice roots of the
counting criterion. \textbf{(a)} $\bK=\bpi$, $s$-wave: against
$1/(\gamma-1)$ the data fall on the straight line
$16\,e^{-\pi\gamma/(\gamma-1)}$ of \eqref{eq:onset-s} down to
$\delta\sim10^{-18}$, an~onset of infinite order. \textbf{(b)}
$\bK=\mathbf0$, $p$-wave (doubly degenerate): the~implicit law
\eqref{eq:onset-p} (solid) is exact to plotting accuracy over six
decades of $\gamma-\gamma_c$, while its explicit leading-log form
(dotted) converges only logarithmically (cf.\ Table~\ref{tab:crit}).}
\label{fig:fonset}
\end{figure}

\section{Duality between statistics and quasimomentum}\label{sec:duality}

The exchange operator $S_\bK$ is a Hermitian involution on the
first-shell space $\ell^2(\{\pm\be_1,\pm\be_2\})$; the~shift
$\bK\mapsto\bK+\bpi$ reverses its sign.

\begin{theorem}[Effective statistics-quasimomentum duality]\label{thm:dual}
For $\gamma=2$ the fermionic atom--dimer Hamiltonian at quasimomentum
$\bK+\bpi$ coincides with the bosonic one at $\bK$, up to a constant:
\[
h_F(\bK+\bpi;\gamma{=}2)-(2+2\gamma)
=h_B(\bK)-4=\eps_D-S_\bK.
\]
Consequently the binding energies satisfy
$\delta_F(\bK+\bpi;\gamma{=}2)=\delta_B(\bK)$.
\end{theorem}

\begin{proof}
$S_{\bK+\bpi}=-S_\bK$; the fermionic Hamiltonian at $\gamma=2$ is
$2+2\gamma+\eps_D+\tfrac{\gamma}{2}S_\bK$; substituting $\bK\mapsto
\bK+\bpi$ and $\gamma=2$ gives $\eps_D-S_\bK$.
\end{proof}

\begin{remark}[Physical interpretation]
The sign reversal $S_{\bK+\bpi}=-S_\bK$ is a lattice counterpart of
the phase factor $e^{i\bpi\cdot\be}=-1$ acquired when the quasimomentum
is shifted by half a reciprocal-lattice vector. The~value $\gamma=2$
is the one at which the fermionic exchange amplitude $\gamma/2$ equals
in modulus the bosonic one, $1$ (two exchange processes with
amplitude $\frac12$ each); the remaining difference of sign is exactly
the factor $e^{i\bpi\cdot\be}=-1$. Thus bosonic and fermionic
atom--dimer states have identical binding energies at quasimomenta
differing by $\bpi$.
\end{remark}

\autoref{fig:fpath} shows the fermionic binding energies along the
high-symmetry path for several mass ratios. At $\gamma=2$ the
fermionic curve coincides point by point with the bosonic curve
translated by $\bpi$, as Theorem~\ref{thm:dual} requires; for $\gamma=3$
and $\gamma=4$ the second state of Corollary~\ref{cor:twostate}
(dashed) is visible, degenerate with the first at $\Gamma$ for
$\gamma=3$ ($p$-wave pair, $\delta=0.0519$) and separated at $M$ for
$\gamma=4$ ($\delta=0.5868$ and $0.1473$, Table~\ref{tab:effective}).

\begin{figure}[t]\centering
\figinclude{0.94\textwidth}{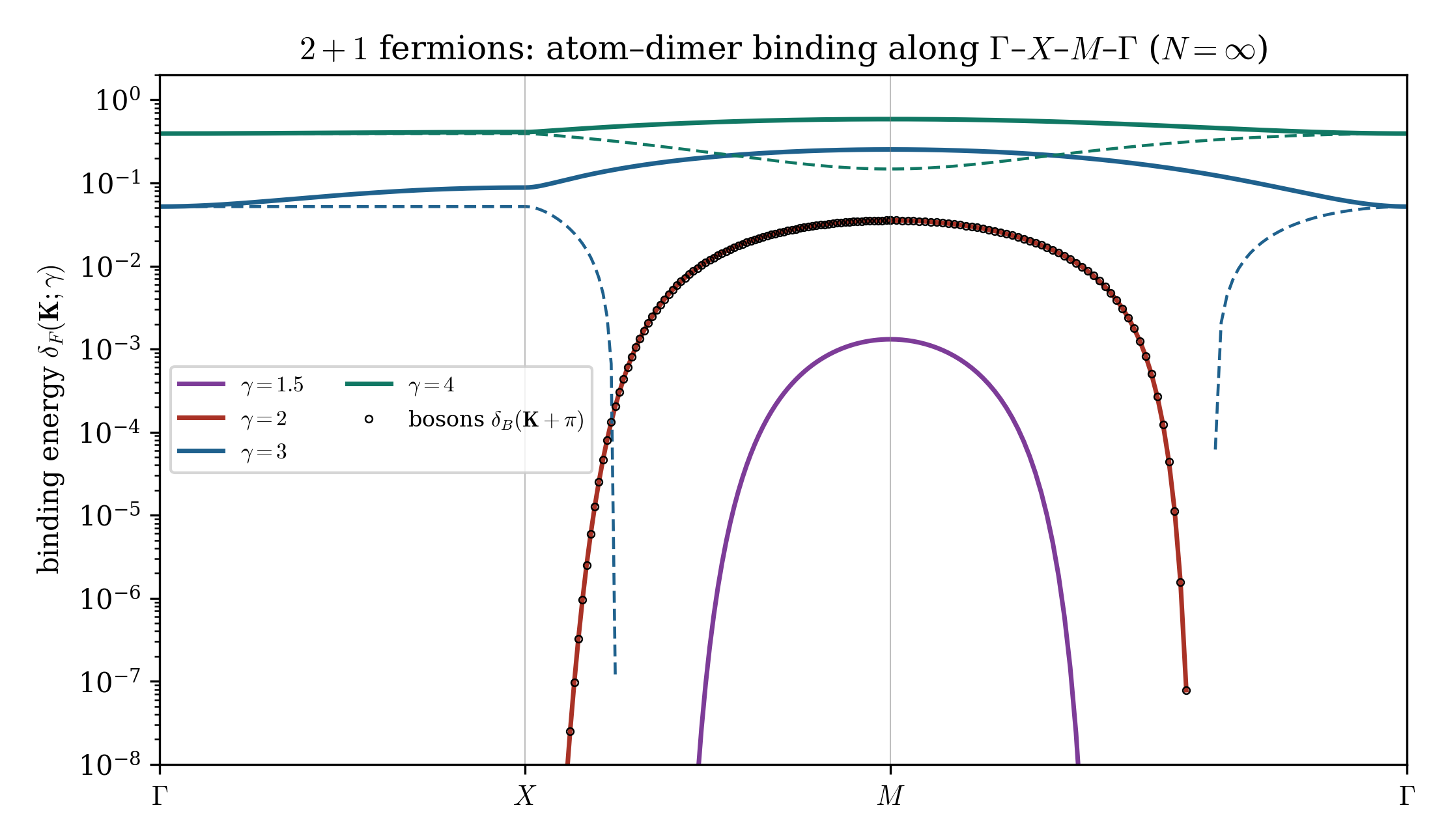}
\caption{Binding energies $\delta_F(\bK;\gamma)$ of the $2+1$ fermionic
atom--dimer states along $\Gamma$--$X$--$M$--$\Gamma$ in the infinite
lattice, for $\gamma=1.5,2,3,4$ (solid: lowest state; dashed: second
state). Open circles: bosonic $\delta_B(\bK+\bpi)$, which coincides with
the $\gamma=2$ curve (Theorem~\ref{thm:dual}). The steep descent at the
zone edges is the essential singularity of Theorem~\ref{thm:boundary}.
The values at $\Gamma$ and $M$ reproduce Table~\ref{tab:effective}.}
\label{fig:fpath}
\end{figure}

\begin{mainprop}[Mixed statistics: two bosons and one distinguishable particle]\label{prop:mixed}
Let particles $1,2$ be identical bosons with hopping $1$ and particle
$3$ be distinguishable with hopping $\gamma$, and let only the two
interspecies pairs interact. Then the first-order effective operator is
\[
h_{2+1B}(\bK;\gamma)=2+2\gamma+\eps_D-\tfrac\gamma2S_\bK
=h_F(\bK+\bpi;\gamma),
\]
so that for every $\gamma$ and every $\bK$
\[
\delta_{2+1B}(\bK;\gamma)=\delta_F(\bK+\bpi;\gamma),\qquad
\gamma_c^{2+1B}(\bK)=\frac{2}{x_{\max}(\bK)}=\gamma_c(\bK+\bpi).
\]
In particular $\gamma_c^{2+1B}(\mathbf0)=1$ and
$\gamma_c^{2+1B}(\bpi)=\pi/(\pi-2)$: the $2+1$ bosonic mixture binds
most easily at $\bK=\mathbf0$, exactly where the $2+1$ fermionic system
binds with the greatest difficulty.
\end{mainprop}

\begin{proof}[Proof sketch; details in \S\ref{sec:S8}]
As in the fermionic case of Lemma~\ref{lem:heff}, only the two
branches $\mathcal B_a$, $\mathcal B_b$ occur (the pair $\{1,2\}$ does
not interact), the diagonal is $4+2\gamma$, particle~$1$ hops give
$\eps_D-2$, and the single particle-$3$ exchange contributes
\[
-\tfrac\gamma2\,e^{i\bK\cdot\bx}\psi(\mathbf0,-\bx),\qquad
\bx\in\{\pm\be_1,\pm\be_2\}.
\]
Bose symmetry gives
$\psi(\mathbf0,\by)=+\varphi(\by)$ instead of $-\varphi(\by)$, which
flips the sign of the exchange term relative to $h_F$. The identity
$S_{\bK+\bpi}=-S_\bK$ then gives $h_{2+1B}(\bK)=h_F(\bK+\bpi)$.
\end{proof}

\begin{table}[t]\centering\footnotesize
\caption{The three statistics cases of the $\Z^2$ three-body problem at
strong coupling. Everywhere $\eps_D$ is the Dirichlet Laplacian on
$\Z^2\setminus\{0\}$, $S_\bK$ the rank-four exchange operator,
$x_{\max/\min}(\bK)$ the extreme roots of $P_\bK$ and
$\delta$ the binding below the atom--dimer threshold.}
\label{tab:statistics}
\setlength{\tabcolsep}{4pt}
\begin{tabular}{L{2.9cm}L{3.1cm}L{4.3cm}L{4.3cm}}
\toprule
 & Three identical bosons & $2+1$ fermions & $2+1$ bosons (mixture) \\
\midrule
Branches in $\Ran\Pi_1$ & three ($\mathcal B_a,\mathcal B_b,\mathcal B_c$) & two & two \\
\hline
Deep trimer channel & $\nu=3$, $\approx-3\mu$ & absent (Pauli) & $\nu=2$ (triple occupancy), $\approx-2\mu$ \\
\hline
Effective operator &
$4+\eps_D-S_\bK$ &
$2+2\gamma+\eps_D+\frac\gamma2S_\bK$ &
$2+2\gamma+\eps_D-\frac\gamma2S_\bK$ \\
\hline
Exchange amplitude & $-1$ (two processes) & $+\gamma/2$ & $-\gamma/2$ \\
\hline
Zone criterion & $x_{\max}(\bK)>1$ & $\frac\gamma2|x_{\min}(\bK)|>1$ & $\frac\gamma2x_{\max}(\bK)>1$ \\
\hline
Critical ratio & --- (no free parameter) & $\gamma_c=2/|x_{\min}|\in[1,\frac{\pi}{\pi-2}]$ & $2/x_{\max}\in[1,\frac{\pi}{\pi-2}]$ \\
\hline
Easiest $\bK$ & $\mathbf0$ ($\delta_B=0.03542$) & $\bpi$ ($\gamma_c=1$) & $\mathbf0$ ($\gamma_c=1$) \\
\hline
Hardest $\bK$ & $\bpi$ (no state) & $\mathbf0$ ($\gamma_c=\frac{\pi}{\pi-2}$) & $\bpi$ ($\gamma_c=\frac{\pi}{\pi-2}$) \\
\hline
Max.\ number of states & $1$ & $2$ (Corollary~\ref{cor:twostate}) & $2$ \\
\hline
Onset at the boundary & essential singularity & essential singularity ($s$), linear$/\log$ ($p$ at $\mathbf0$) & essential singularity ($s$), linear$/\log$ ($p$ at $\bpi$) \\
\hline
Map & \autoref{fig:zonemaps}(a) & \autoref{fig:zonemaps}(b)--(f) & (b)--(f) shifted by $\bpi$; at $\gamma=2$ equal to (a) \\
\bottomrule
\end{tabular}
\end{table}

\autoref{tab:statistics} collects the three statistics cases side by
side. The pattern is uniform: the statistics enters the effective
theory \emph{only} through the sign and the magnitude of the exchange
amplitude in front of $S_\bK$, and any change of that sign is
equivalent to the shift $\bK\to\bK+\bpi$. This is why the bosonic and
the mixed problems, which are physically distinct, share one and the
same zone geometry up to a translation of the Brillouin zone
(\autoref{fig:zonemaps}).

\section{Effect of external couplings}\label{sec:external}

The effective Hamiltonians of Lemma~\ref{lem:heff} can be generalised
to include nearest-neighbour attractions, three-body forces, and
periodic (Floquet) modulation.

\subsection{Nearest-neighbour attraction}

Let $U_{nn}$ be the operator of multiplication by the number of pairs
of particles on nearest-neighbour sites, and add $-\lambda U_{nn}$
($\lambda=O(1)$ in units of the hopping) to $H_\mu(\bK)$.

\begin{proposition}[Nearest-neighbour attraction]\label{prop:nn}
For three identical bosons the operator $\Pi_1U_{nn}\Pi_1$ is
unitarily equivalent to $2P_{\rm sh}$ on $\ell^2(\Z^2\setminus\{0\})$,
where $P_{\rm sh}$ projects onto the first shell. For the $2+1$
systems, if the $F$--$I$ and $F$--$F$ nearest-neighbour couplings are
$\lambda_{FI}$ and $\lambda_{FF}$, the effective shell term is
$-(\lambda_{FI}+\lambda_{FF})P_{\rm sh}$; hence the atom--dimer
operators become
\[
h_B^\lambda=4+\eps_D-S_\bK-2\lambda P_{\rm sh},\qquad
h_F^{\Lambda}=2+2\gamma+\eps_D+\tfrac\gamma2S_\bK-\Lambda P_{\rm sh},
\quad \Lambda=\lambda_{FI}+\lambda_{FF}.
\]
The perturbation is \emph{not} a constant shift: it acts only when the
atom is adjacent to the dimer. At~$\bK=\bpi$ the bosonic thresholds are
\[
\lambda_c^{p}=\frac{4-\pi}{4(\pi-2)}\approx0.18798,\qquad
\lambda_c^{s}=\frac34,\qquad
\lambda_c^{d}=\frac{8-\pi}{4(4-\pi)}\approx1.41495,
\]
and for $2+1$ fermions at $\bK=\bpi$, $\gamma=1$, every $\Lambda>0$ binds
the $s$-channel, with $\ln(16/\delta_F)=\pi(1+2\Lambda)/(2\Lambda)+o(1)$ as $\Lambda\downarrow0$.
\end{proposition}

\begin{proof}
An atom on the shell is a nearest neighbour of both dimer constituents
(two pairs); elsewhere on $\Ran\Pi_1$ no pair is at distance one. The
counting criterion (Theorem~\ref{thm:count}) with $C=-S_\bK-2\lambda
P_{\rm sh}$ is diagonal at $\bK=\bpi$ in the $s,p,d$ basis, because
$S_{\bpi}=-S_{\mathbf0}$ acts as $-1,+1,-1$ on $s,p,d$ and $M_0$ as
$2,A,\frac8\pi-2$: the conditions $2(2\lambda-1)>1$,
$A(1+2\lambda)>1$, $(\frac8\pi-2)(2\lambda-1)>1$ give the three values.
For fermions at $\gamma=1$ the $s$-condition is
$m_s(\delta)\bigl(\frac12+\Lambda\bigr)=1$, i.e.\
$1/b_0\simeq4\Lambda/(1+2\Lambda)$, and $b_0\simeq\frac1{2\pi}\ln\frac{16}\delta$
gives the stated law.
\end{proof}

\subsection{Three-body force}

Adding $-w\mu\Pi_3$, with $\Pi_3$ the rank-one projection onto the
constant function, modifies the trimer band but affects the
atom--dimer sector only at order $\mu^{-1}$.

\subsection{Floquet modulation}

Under high-frequency lattice shaking along $(1,1)$, entering the dispersion of
species $j$ through the Peierls substitution
$\eps(\bp)\to\eps\bigl(\bp+\tfrac{\mathcal A_j}{\omega}\sin(\omega t)\,(1,1)\bigr)$,
the effective tunnelling is renormalised by $\tau_j=J_0(\mathcal A_j/\omega)$
\cite{Eckardt2017};
we assume $\omega\gg\mu$, so that no photon-assisted
resonance with the interaction scale occurs.

\begin{proposition}[Floquet-renormalised zones]\label{prop:floqAD}
At leading order the atom--dimer effective Hamiltonian is
$h_B(\bK;\tau)\cong6-6|\tau|+|\tau|\,h_B(\bK)$ for $\tau>0$ and,
after the gauge transformation
$\varphi(\bx)\mapsto(-1)^{x_1+x_2}\varphi(\bx)$,
$h_B(\bK;\tau)\cong6-6|\tau|+|\tau|\,h_B(\bK+\bpi)$ for $\tau<0$,
with $\tau=J_0(\mathcal A/\omega)$.
Hence $\delta_B(\bK;\tau)=\tau
\delta_B(\bK)$ for $\tau>0$ and $\delta_B(\bK;\tau)=|\tau|
\delta_B(\bK+\bpi)$ for $\tau<0$. In~particular, a sign change
$J_0<0$ inverts the effective band, and, for $2+1$ fermions,
$\gamma_{\rm eff}=\gamma\tau_3/\tau_f$ allows tuning through
$\gamma_c$ without changing atomic species.
\end{proposition}

\begin{remark}
The formula follows from $\Pi_1E_\bK^{(\tau)}\Pi_1=6-2\tau+\tau\eps_D-\tau
S_\bK$ (all hops multiplied by $\tau$); for $\tau<0$ the gauge
transformation $\varphi(\bx)\mapsto(-1)^{x_1+x_2}\varphi(\bx)$ flips the
sign of the atom hopping and leaves $S_\bK$ invariant (both shell
partners $\pm\be$ are odd sites), which yields the $|\tau|$ form. As a
consistency check, at $\tau=1$ it simplifies algebraically to $4+\eps_D-S_\bK=h_B(\bK)$ of
Lemma~\ref{lem:heff}, and at $\tau=-1$ to $4+\eps_D+S_\bK=
4+\eps_D-S_{\bK+\bpi}=h_B(\bK+\bpi)$, matching the stated
$\bK\to\bK+\bpi$ inversion. A worked example is given in Supplementary Material \S\ref{sec:S62}: at
$\bK=\bpi$ the undriven bosonic system has no atom--dimer state, while
for $\mathcal A/\omega$ just beyond the first zero $2.4048$ of $J_0$
($\tau<0$) a state appears with $\delta_B=|\tau|\,\delta_B(\mathbf0)
=0.03542\,|\tau|$. Corrections of higher order in $\omega^{-1}$ are not
analysed here.
\end{remark}

\section{Numerical verification}\label{sec:numerical}

The effective theory of Sections~\ref{sec:effective}--\ref{sec:critical}
is verified by two independent numerical schemes, plus a~fully
independent audit of the closed-form threshold constants ($A$, $B$,
$\gamma_c$, $\delta_\infty$, $K^\ast$) carried out from scratch~--
by direct Brillouin-zone quadrature and symbolic algebra, without
reference to the derivations of Sections~\ref{sec:threshold}--\ref{sec:zones}
-- summarised in \S\ref{ssec:selfcontained} and detailed in
Supplementary Material \S\ref{sec:S4}.

\subsection{Independent audit of the threshold constants}\label{ssec:selfcontained}

Because the counting criterion of \autoref{thm:count} and the
critical mass ratios of \autoref{thm:zoneF} rest entirely on the two
lattice constants $A$, $B$ and on the quartic $P_\bK$ of
Lemma~\ref{lem:char}, these were re-derived independently, without
using the Dirichlet-elimination argument of Lemma~\ref{lem:M}:
\begin{itemize}
\item $A=2-4/\pi=0.72676046\ldots$ and $B=2/\pi=0.63661977\ldots$
were recovered to $8$ significant figures from the direct
double integral (square-lattice effective resistance / lattice Green-function integral)
$R(\bx)=\tfrac1{(2\pi)^2}\iint_{\T^2}\bigl[1-\cos(\bp\cdot\bx)\bigr]/
\eps(\bp)\,d\bp$ at $\bx=(2,0)$ and $\bx=(1,1)$ -- the classical
two-point resistances of the infinite square lattice
\cite{Cserti,Atk} -- with $A+2B=2$ confirmed to machine precision.
\item The general quartic $P_\bK(x)$ of Lemma~\ref{lem:char} was
checked by computer algebra to reduce \emph{identically in $A$}
(not merely numerically) to each of its four stated special-point
factorisations at $\bK=\mathbf0,\bpi,(\pi,0),(\tfrac\pi2,\tfrac\pi2)$.
\item All four closed-form values of $\gamma_c$ in \autoref{thm:zoneF}
were re-obtained from the roots of these factorisations.
\item $\kappa_0$ was evaluated from $(A-A(\delta))\pi/\delta-\ln\frac1\delta$
at $\delta=10^{-4},\dots,10^{-8}$, converging to
$0.6309960\ldots=1+4\ln2-\pi$, the value proved in
Lemma~\ref{lem:constants}(v) (\S\ref{sec:S42}).
\item $\delta_\infty=0.0354199892\ldots$ was reproduced to $9$
digits as the root of $b_0(\delta)=1/(1+\delta)$, with $b_0(\delta)$
computed by direct quadrature (no series expansion used).
\item The diagonal zone-boundary point of \autoref{thm:zoneB}(iv),
$\cos K^\ast=2-3\pi/4$, was reproduced to $12$ digits by root-finding
on $P_{(K,K)}(1)=0$.
\end{itemize}
The checks fall into three categories, which we keep distinct:
(i)~exact algebraic identities (e.g.\ $A+2B=2$, the quartic
factorisations of Lemma~\ref{lem:char}), verified by symbolic
algebra; (ii)~rigorous spectral estimates (the Feshbach--Schur
bound of Theorem~\ref{thm:feshbach}, the counting criterion of
Theorem~\ref{thm:count}), proved analytically; (iii)~numerical
solutions of the exact infinite-lattice threshold equation and
finite-volume cross-checks, which confirm the consistency of the
implementation but do not by themselves prove the $\mu\to\infty$
asymptotics.
No discrepancy was found at the level of precision used. Full code
and numerical output are given in Supplementary Material \S\ref{sec:S4}.

\subsection{Exact finite-volume three-body solver}\label{ssec:STM}

The exact solver is based on the following lattice form of the
Skornyakov--Ter-Martirosyan (STM) reduction. The method is classical
for zero-range interactions in the continuum
\cite{Skornyakov1956,Faddeev1963} and has been used for lattice
three-body operators (see, e.g., \cite{LakaevAbdullaev2003,ALM2004});
since none of the hypotheses of those works needs to be verified here,
we state and prove the precise version used.

Let $\T^2_\star$ denote either $\T^2$ with the normalised Haar measure
or the discrete torus $\T^2_N=\{2\pi\mathbf j/N:\mathbf j\in\Z_N^2\}$
with the normalised counting measure (the momentum space of the
periodic $N\times N$ lattice, $\bK\in\T^2_N$); $\langle\cdot\rangle_\bq$
is the corresponding average and $\mathcal H^\sigma$ the corresponding
fibre space. For $z<\min E^\sigma_\bK$ (the minimum over
$\T^2_\star\times\T^2_\star$) put
\[
\Delta^\sigma_\mu(\bp,z):=1-\mu\Bigl\langle\frac1{E^\sigma_\bK(\bp,\bq)-z}\Bigr\rangle_\bq,
\qquad
(L^\sigma_\mu(z)\varphi)(\bp):=\Delta^\sigma_\mu(\bp,z)\varphi(\bp)
-\kappa_\sigma\,\mu\Bigl\langle\frac{\varphi(\bq)}{E^\sigma_\bK(\bp,\bq)-z}\Bigr\rangle_\bq ,
\]
on $L^2(\T^2_\star)$, with $\kappa_{3B}=2$, $\kappa_{2+1F}=-1$,
$\kappa_{2+1B}=1$.

\begin{lemma}[Lattice STM reduction]\label{lem:STM}
Let $z<\min E^\sigma_\bK$. The map $f\mapsto\varphi:=V_1f$, i.e.\
$\varphi(\bp)=\langle f(\bp,\cdot)\rangle$, is a linear isomorphism of
$\ker(H^\sigma_\mu(\bK)-z)$ onto $\ker L^\sigma_\mu(z)$, with inverse
\begin{equation}\label{eq:STM-inverse}
f=\mu\,(E^\sigma_\bK-z)^{-1}\Phi,\qquad
\Phi(\bp,\bq)=\begin{cases}
\varphi(\bp)+\varphi(\bq)+\varphi(\bK-\bp-\bq) & (3B),\\
\varphi(\bp)-\varphi(\bq) & (2{+}1F),\\
\varphi(\bp)+\varphi(\bq) & (2{+}1B).
\end{cases}
\end{equation}
In particular $z$ is an eigenvalue of $H^\sigma_\mu(\bK)$ if and only if
$\ker L^\sigma_\mu(z)\ne\{0\}$, and the multiplicities coincide.
\end{lemma}

\begin{proof}
We give the bosonic case; the others differ only in the bookkeeping of
signs. Since $z<\min E_\bK$, $(E_\bK-z)^{-1}$ is a bounded
multiplication operator. Let $f\in\Ltwo$ with $H_\mu(\bK)f=zf$, i.e.\
$(E_\bK-z)f=\mu Vf$, and put $\varphi:=V_1f$. By the symmetry of $f$
(\S\ref{sec:setup}), $(V_2f)(\bp,\bq)=\varphi(\bq)$ and
$(V_3f)(\bp,\bq)=\varphi(\bK-\bp-\bq)$, so $Vf=\Phi$ and
$f=\mu(E_\bK-z)^{-1}\Phi$. Applying $V_1$, i.e.\ averaging over $\bq$,
and using the measure-preserving change of variables
$\bq\mapsto\bK-\bp-\bq$, under which $E_\bK(\bp,\bq)$ is invariant, in the
third term, gives
$\varphi(\bp)=\mu\langle(E_\bK-z)^{-1}\rangle_\bq\varphi(\bp)
+2\mu\langle(E_\bK-z)^{-1}\varphi(\bq)\rangle_\bq$, i.e.\
$L_\mu(z)\varphi=0$. Moreover $\|\varphi\|\le\|f\|$, and $\varphi=0$
would give $Vf=0$, hence $(E_\bK-z)f=0$ and $f=0$: the map is
injective. Conversely, let $\varphi\in\ker L_\mu(z)$ and define $f$ by
\eqref{eq:STM-inverse}. Each of the three terms of $\Phi$ is in
$L^2((\T^2_\star)^2)$ with norm $\|\varphi\|$ (the maps
$(\bp,\bq)\mapsto\bp,\bq,\bK-\bp-\bq$ are measure preserving), so
$f\in L^2$; $f$ is invariant under the permutations of the three
momenta, since $E_\bK$ and $\Phi$ are. The computation above, read
backwards, gives $V_1f=\varphi$, hence $Vf=\Phi$ and
$(E_\bK-z)f=\mu Vf$, i.e.\ $H_\mu(\bK)f=zf$; and $V_1f=\varphi$ shows that
the two maps are mutually inverse. For $2{+}1F$, antisymmetry gives
$V_2f=-\varphi(\bq)$ and the single exchange term carries the factor
$-1$; for $2{+}1B$ it carries $+1$; there is no $V_3$ term.
\end{proof}

\begin{remark}[Use of Lemma~\ref{lem:STM}]\label{rem:STMuse}
(i) The lemma applies to all eigenvalues used in this article: by
Theorem~\ref{thm:feshbach}(a),(b) they lie below
$-\mu+w_\sigma<0\le\min E^\sigma_\bK$. (ii) On $\T^2_N$ the fibre is
finite-dimensional, $H^\sigma_\mu(\bK)$ is a Hermitian matrix of size
$O(N^4)$, and Lemma~\ref{lem:STM} reduces its eigenvalue problem below
$\min E^\sigma_\bK$ \emph{exactly} to the $N^2\times N^2$ problem
$\ker L^\sigma_\mu(z)\ne\{0\}$; no approximation is involved, only the
choice of $N$. (iii) For the computations we use the symmetric form:
where $\Delta^\sigma_\mu(\cdot,z)>0$, $\ker L^\sigma_\mu(z)\ne\{0\}$ is
equivalent to $1\in\spec\mathcal A^\sigma(z)$, where $\mathcal
A^\sigma(z)=\kappa_\sigma\mu\,\Delta^{-1/2}(E^\sigma_\bK-z)^{-1}\Delta^{-1/2}$ is
the Hermitian matrix with kernel
$\kappa_\sigma\mu\,\Delta(\bp)^{-1/2}(E(\bp,\bq)-z)^{-1}\Delta(\bq)^{-1/2}$,
and the eigenvalues are located as the crossings of the eigenvalue
branches of $\mathcal A^\sigma(z)$ with $1$ by Brent's method at
$10^{-12}$ precision. The finite-volume threshold, used in
Table~\ref{tab:kappa} and \autoref{fig:transfer}, is
$\theta_N:=\sup\{z:\ \min_\bp\Delta^\sigma_\mu(\bp,z)>0\}$, the lowest
energy of a dimer plus a free atom on $\T^2_N$.
\end{remark}

\subsection{Direct diagonalisation}

Direct diagonalisation of the sparse finite-volume Hamiltonian on the
relative-coordinate torus $N^4$ with projection onto the $S_3$-symmetric
sector is used as an independent check. The two schemes agree to
$10^{-8}$ at $N=8$, $\mu=20$:
$z_1(0)=-54.07693698$, $z_2(0)=-16.28398325$,
$z_1(\bpi)=-54.07318058$. By Lemma~\ref{lem:STM} the two schemes
compute the same finite-volume eigenvalues, so their agreement checks
the implementations, not the reduction, which is exact. The trimer levels lie at $-3\mu+6+O(\mu^{-1})$, in agreement with \cite{Trimer2026}; higher-order trimer coefficients are not used in the present atom--dimer analysis.

\subsection{Comparison of effective theory and exact solver}

Table~\ref{tab:effective} compares the atom--dimer binding energies of
the effective theory with the exact finite-volume solver. Because of
the slow $s$-channel convergence (Remark~\ref{rem:slow}), the
comparison must be made at equal volume: the column ``eff., $N=32$''
evaluates $h_B$, $h_F$ on the same $32\times32$ torus as the exact
solver, and the column ``eff., $N=\infty$'' gives the thermodynamic
limit. The agreement at equal volume is at the level expected from
the $O(\mu^{-1})$ remainder of Theorem~\ref{thm:feshbach}; weak
bindings are strongly volume dependent ($\delta_B(\pi,0)$ decreases as
$N^{-2}$: $2.09\cdot10^{-4}$, $5.47\cdot10^{-5}$, $1.44\cdot10^{-5}$,
$3.80\cdot10^{-6}$ for $N=32,64,128,256$).

\begin{table}[h]\centering\small
\caption{Atom--dimer binding energies $\delta$ (in units of the hopping,
energy $-\mu+{\rm threshold}-\delta$). $\delta$ in the exact column is measured from the exact finite-$\mu$ atom--dimer threshold (as in Table~\ref{tab:kappa}). Exact: finite-volume three-body
solver, $N=32$, $\mu=400$. For~$\gamma=3$, the threshold shift of the order of $16/\mu$ is comparable to $\delta$.}
\label{tab:effective}
\begin{tabular}{llccc}
\toprule
System & $\bK$ & eff., $N=\infty$ & eff., $N=32$ & exact, $N=32$ \\
\midrule
Bosons & $(\pi,0)$ & $1.45\cdot10^{-15}$ & $0.000209$ & $0.000208$ \\
Bosons & $(\pi/2,\pi/2)$ & $0.000650$ & $0.002144$ & $0.002258$ \\
Bosons & $(3\pi/4,3\pi/4)$, $\bpi$ & none & none & none \\
Fermions $\gamma=2$ & $\bpi$ & $0.035420$ & $0.035450$ & $0.035665$ \\
Fermions $\gamma=2$ & $\mathbf 0$ & none & none & none \\
Fermions $\gamma=2$ & $(\pi/2,\pi/2)$ & $0.000650$ & $0.002144$ & $0.002131$ \\
Fermions $\gamma=3$ & $\mathbf 0$ & $0.051929$ ($\times2$) & $0.051926$ ($\times2$) & $0.047661$ ($\times2$) \\
Fermions $\gamma=4$ & $\bpi$ & $0.586815;\ 0.147265$ & $0.586815;\ 0.147265$ & $0.587256;\ 0.155191$ \\
\bottomrule
\end{tabular}
\end{table}

\subsection{First correction at $\bK=0$}

The $O(\mu^{-1})$ correction to $z_2^s(\mu)=-\mu+C$ is
$\kappa_B/\mu$ with $\kappa_B\approx-4.35966$, obtained from the
second-order term of the Schur complement $F(z)$ of
Theorem~\ref{thm:feshbach} evaluated on the normalised $h_B$
eigenvector, and confirmed by the exact solver (Table~\ref{tab:kappa}).

\begin{table}[h]\centering\small
\caption{Bosons at $\bK=0$, exact solver ($N=48$). Here
$\delta(\mu)=\theta(\mu)-z_2(\mu)$ is measured from the exact
atom--dimer threshold $\theta(\mu)=-\mu+4-4/\mu+O(\mu^{-2})$ (dimer at
rest), so that $\mu[z_2+\mu-C]=-4-\mu[\delta(\mu)-\delta_\infty]+O(\mu^{-1})$.
Both columns converge, to~$0.35966$ and to
$\kappa_B=-4-0.35966=-4.35966$.}
\label{tab:kappa}
\begin{tabular}{rcccc}
\toprule
$\mu$ & $\delta(\mu)$ exact & $\mu[\delta(\mu)-\delta_\infty]$ & $\mu[z_2+\mu-C]$ \\
\midrule
25  & $0.0505554$ & $0.3784$ & $-4.3847$ \\
50  & $0.0428313$ & $0.3705$ & $-4.3721$ \\
100 & $0.0390749$ & $0.3655$ & $-4.3659$ \\
200 & $0.0372335$ & $0.3626$ & $-4.3627$ \\
400 & $0.0363232$ & $0.3612$ & $-4.3612$ \\
800 & $0.0358708$ & $0.3604$ & $-4.3604$ \\
$\infty$ & $\delta_\infty$ & $0.3597$ & $-4.3597$ \\
\bottomrule
\end{tabular}
\end{table}

For reference, \autoref{fig:zonepath} displays the extreme
eigenvalues $x_{\max}(\bK)$ and $-x_{\min}(\bK)$ of the quartic
$P_\bK$ along the high-symmetry path, together with the critical mass
ratio $\gamma_c(\bK)$, and \autoref{fig:gammacmap} shows
$\gamma_c(\bK)$ over the full Brillouin zone. Both figures follow
directly from Lemma~\ref{lem:char} and are reproduced here for
convenience.

\begin{figure}[h]\centering
\figinclude{0.97\textwidth}{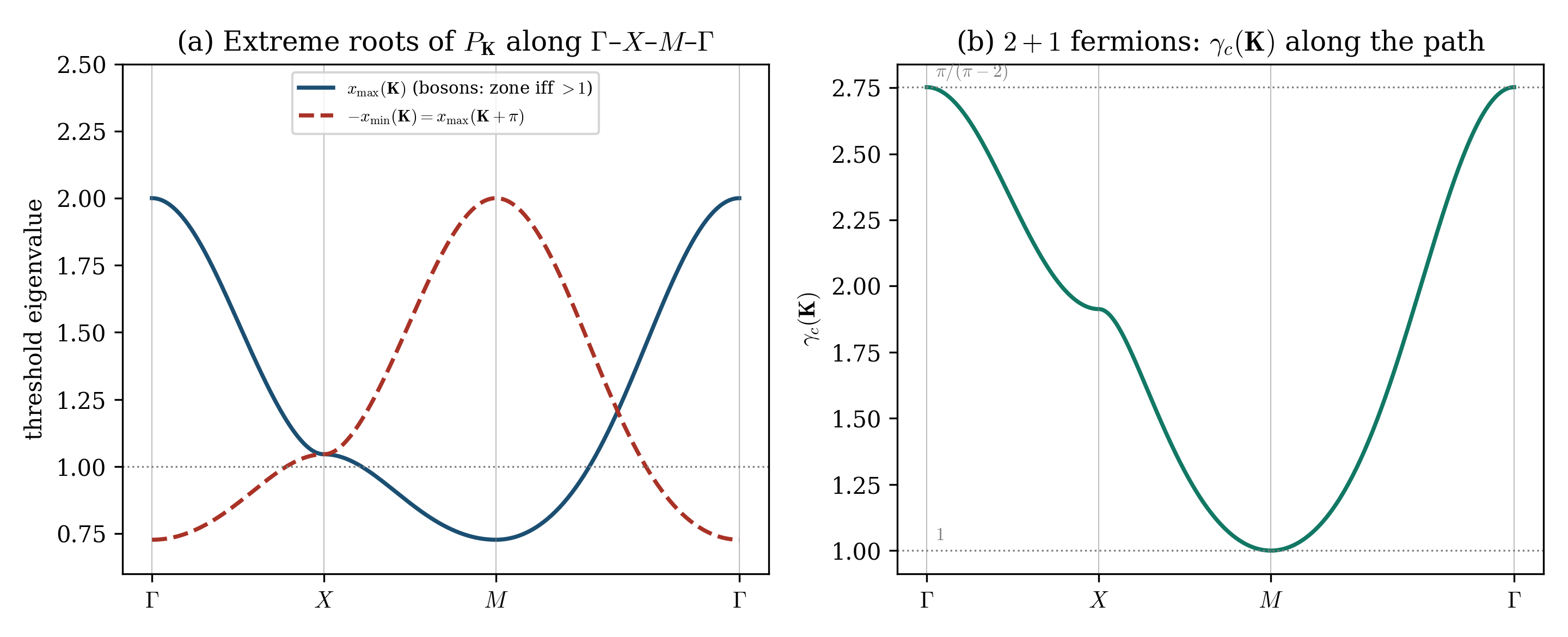}
\caption{Zone-indicator eigenvalues along the high-symmetry path
$\Gamma(\mathbf0)$--$X(\pi,0)$--$M(\bpi)$--$\Gamma$, computed from
the closed-form quartic $P_\bK$ of Lemma~\ref{lem:char} with
$A=2-4/\pi$. \textbf{(a)} $x_{\max}(\bK)$, the bosonic zone
indicator of Theorem~\ref{thm:zoneB} (an atom--dimer state exists
iff $x_{\max}(\bK)>1$), together with $-x_{\min}(\bK)$; the two
coincide with each other's $\bpi$-translate,
$-x_{\min}(\bK)=x_{\max}(\bK+\bpi)$, which is the duality of
Theorem~\ref{thm:dual} made visible. \textbf{(b)} The critical mass
ratio $\gamma_c(\bK)=2/|x_{\min}(\bK)|$ along the same path,
attaining $\pi/(\pi-2)$ at $\Gamma$ and $1$ at $M$ as in
Theorem~\ref{thm:zoneF}. The bosonic zone boundary on the
$M$--$\Gamma$ diagonal, $\cos K^\ast=2-3\pi/4$
($K^\ast=0.61593\,\pi$), is reproduced by panel (a) to
$2\cdot10^{-5}\pi$ at the plotted resolution; root-finding gives it to
$10^{-15}$ (Supplementary Material \S\ref{sec:S43}).}
\label{fig:zonepath}
\end{figure}

\begin{figure}[h]\centering
\figinclude{0.7\textwidth}{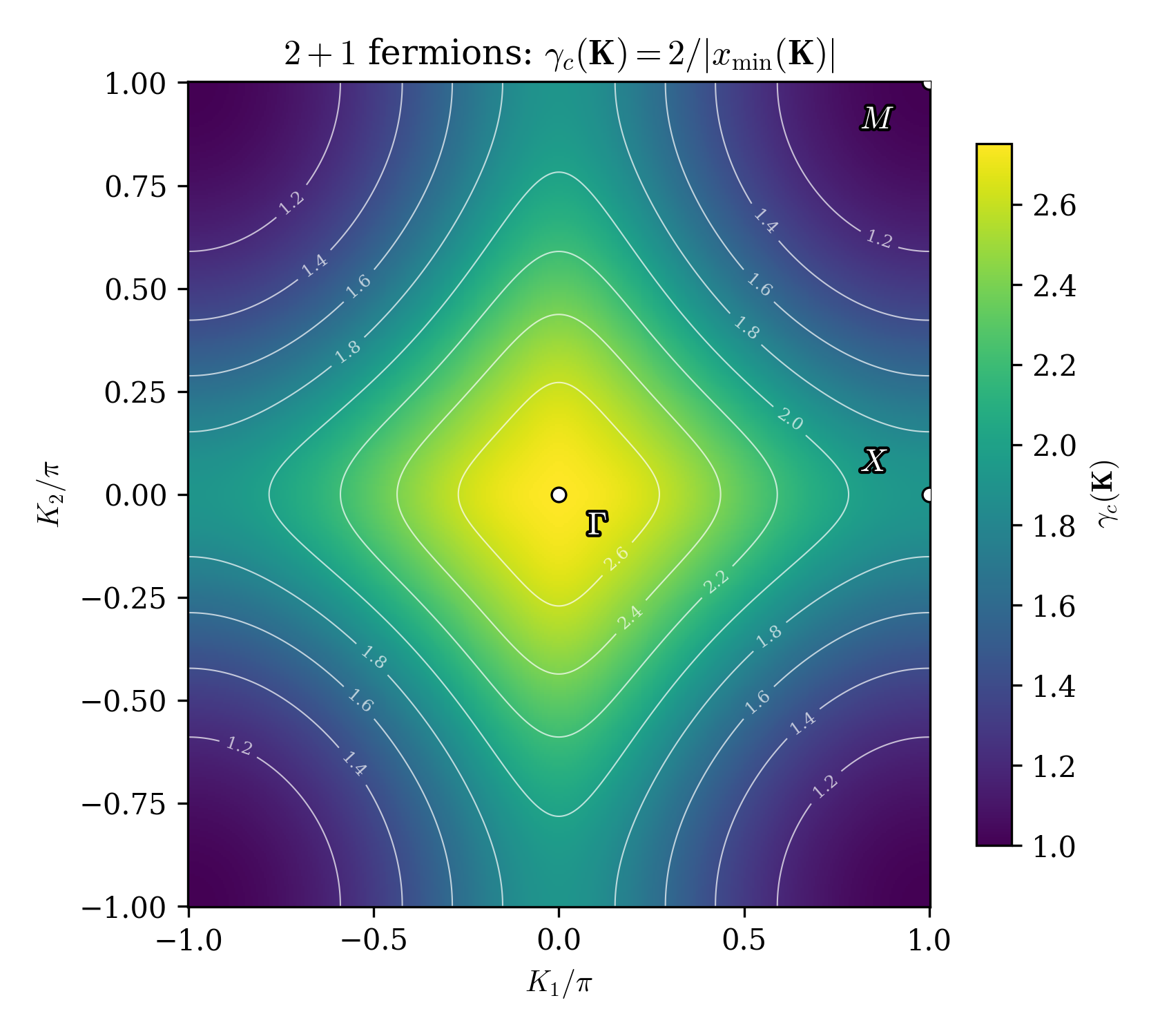}
\caption{The critical mass ratio $\gamma_c(\bK)=2/|x_{\min}(\bK)|$ of
Theorem~\ref{thm:zoneF} over the full Brillouin zone $\bK\in\T^2$,
obtained by finding the smallest root of $P_\bK(x)$ at each $\bK$ on a~$400\times400$ grid. The maximum $\pi/(\pi-2)\approx2.752$ sits at
$\Gamma$ and the minimum $1$ at the zone-boundary point $M=\bpi$;
white contours are drawn at $\gamma_c=1.2,1.4,\dots,2.6$. For
$\gamma$ below the global minimum ($\gamma<1$) no $\bK$ supports an
atom--dimer state; for~$\gamma$ above the global maximum
($\gamma>\pi/(\pi-2)$) every $\bK$ does, in agreement with
Theorem~\ref{thm:zoneF}(i),(iii).}
\label{fig:gammacmap}
\end{figure}

\section{Conclusion}\label{sec:conclusion}

We have solved problems (P1)--(P5) of \S\ref{ssec:problem} for the
atom--dimer sector of three particles on $\Z^2$ at strong coupling, for
three identical bosons and for both $2+1$ mixtures, at every total
quasimomentum. Earlier work treated the trimer \cite{Trimer2026}
and the bosonic
atom--dimer constant $C$ at
$\bK=\mathbf0$
\cite{Trimer2026},
an elliptic reduction and an order-matching criterion \cite{arXiv2608}, and the
$2+1$ fermionic
threshold at $\bK=\mathbf0$ \cite{Companion} (interpreted there as a trimer); the results below are
new.

\emph{Rigorous results.} (1)~The limit operators $h_B$, $h_F$, $h_{2+1B}$ are derived exactly
(Lemma~\ref{lem:heff}, Proposition~\ref{prop:mixed}). The essential
spectrum of the three-body fibre operator begins at an atom--dimer
threshold $\theta_\mu(\bK)$ with $|\theta_\mu(\bK)+\mu-t_\sigma|\le\eta_\sigma$,
and below it the eigenvalue counting function is squeezed between those
of $-\mu+h^\sigma(\bK)$ shifted by $\pm\eta_\sigma$, with
$\eta_{3B}=36/(\mu-10)$, $\eta_{2+1F}=(4+2\gamma)^2/(\mu-7-3\gamma)$,
$\eta_{2+1B}=(4+2\gamma)^2/(\mu-8-4\gamma)$
(Theorem~\ref{thm:feshbach}); hence the number of atom--dimer states
bound uniformly in $\mu$ is the number of bound states of $h^\sigma(\bK)$
(Corollary~\ref{cor:Nsigma}). (2)~Existence and number of atom--dimer
states reduce to a $4\times4$ matrix whose limit is built from the
square-lattice resistances $A=2-\frac4\pi$, $B=\frac2\pi$
(Lemmas~\ref{lem:constants}, \ref{lem:M}, Theorem~\ref{thm:count}),
for all parameters including the critical ones, where the critical
eigenvalue is approached from above (Proposition~\ref{prop:onesided}). (3)~The critical mass
ratio is an explicit function on the Brillouin zone with
$\gamma_c\ge1$, $\gamma_c(\bpi)=1$ and closed forms at all symmetry
points; in particular $\gamma_c(\mathbf0)=\pi/(\pi-2)$ gives the closed
form of the constant $2.75194$ of \cite{Companion} and identifies that
state as the $p$-wave atom--dimer state (Theorem~\ref{thm:zoneF}).
(4)~There are at most two atom--dimer states at any $\bK$ and
$\gamma$ (Corollary~\ref{cor:twostate}). (5)~On the lattice a change
of statistics is equivalent to the shift $\bK\to\bK+\bpi$
(Theorem~\ref{thm:dual}, Proposition~\ref{prop:mixed}). (6)~The onset
is an essential singularity wherever the threshold mode has an
$s$-wave component and linear-over-logarithmic in the pure $p$-wave
(Theorems~\ref{thm:crit}, \ref{thm:boundary}). The $p$- and $d$-channels are regular at threshold by symmetry; in the
$s$-channel the Dirichlet condition turns the divergent two-dimensional
logarithm $b_0(\delta)$ into the finite limit $m_s\to2$, approached as
$-1/(4b_0(\delta))$; this single term explains both the
essential-singularity onset and the strong finite-volume dependence of
weak atom--dimer bindings.

\emph{Numerical confirmations.} The bound $\gamma_c\le\pi/(\pi-2)$, the fact that
at most one bosonic root exceeds $1$, and the range of the second fermionic
threshold, first observed on $601\times601$ grids, are~proved by the identities
\eqref{eq:PA}--\eqref{eq:Pd}; the grids serve as independent checks.

The trimer energy is $\bK$-independent at the two leading orders, $-3\mu+6+O(\mu^{-1})$, established in~\cite{Trimer2026}, whereas the~atom--dimer binding energy depends on $\bK$ already at order
$\mu^0$ (through the exchange operator $S_\bK$), ranging from
$\delta_B=0.0354$ at $\Gamma$ to zero beyond $\cos K^\ast=2-3\pi/4$ on
the diagonal; no~higher-order trimer bandwidth estimate is
used in the present atom--dimer analysis. The atom--dimer states are therefore far more mobile and
more sensitive to $\bK$ than trimers, which makes the critical mass
ratio $\gamma_c(\bK)$ of Figure~\ref{fig:gammacmap} a~natural target for~quasimomentum-resolved spectroscopy ($\gamma=t_3/t_f$, and Section \ref{sec:duality}) of mass-imbalanced mixtures in
optical lattices, with Floquet driving (Proposition~\ref{prop:floqAD})
providing, for~species-selective driving ($\tau_3\ne\tau_f$), a way of tuning through $\gamma_c$.

The dualities of Theorem~\ref{thm:dual} and Proposition~\ref{prop:mixed}
are exact unitary equivalences generated by the sublattice identity
$S_{\bK+\bpi}=-S_\bK$; they are not strong--weak or holographic
dualities, nor do they involve a~large-$N$ limit, and they hold for
every $\gamma>0$ and every $\bK\in\T^2$.

\emph{Broader significance.} The results above are not confined to
the model studied here. The exact Brillouin-zone dependence of the
atom--dimer thresholds and the closed-form critical mass ratios
provide benchmark quantities for quantum simulation of
mass-imbalanced three-body systems in optical lattices, where the
hopping ratio $\gamma=t_3/t_f$ is tunable by species-selective
lattice depths and quasimomentum-resolved spectroscopy is available.
The statistics--quasimomentum duality of Theorem~\ref{thm:dual} and
Proposition~\ref{prop:mixed} identifies a lattice-specific
equivalence that can be exploited in Floquet engineering, allowing one
to tune through the critical ratio without changing atomic species.
The onset laws of Theorems~\ref{thm:crit} and~\ref{thm:boundary}
admit an interpretation in terms of a marginal coupling and dimensional
transmutation (Remark~\ref{rem:RG}); the comparison with the
fixed-point-annihilation scenario recorded there is contextual. Finally, the explicit $4\times4$
threshold matrix of Lemma~\ref{lem:M} and the quartic $P_\bK$ of
Lemma~\ref{lem:char} offer exact algebraic benchmarks for
finite-volume and infinite-lattice solvers, and the method extends to
anisotropic lattices, mixed statistics, and higher dimensions.

\begin{table}[!ht]\centering\footnotesize
\caption{Applied content of the effective theory: closed-form
quantities, their values and the experimental quantity they control.
Energies are in units of the hopping amplitude; $\delta$ is measured
from the atom--dimer threshold. ``Exact'' means a closed-form
expression; ``$10^{-k}$'' means the value is known numerically to that
relative accuracy (Supplementary Material \S\ref{sec:S4}).}
\label{tab:applied}
\setlength{\tabcolsep}{4pt}
\begin{tabular}{L{3.4cm}L{4.6cm}L{1.9cm}L{4.9cm}}
\toprule
Quantity & Closed form & Accuracy & Experimental handle \\
\midrule
Critical mass ratio at $\bK=\mathbf0$ &
$\gamma_c=\pi/(\pi-2)=2.75194$ & exact &
choice of species pair in a~mass-imbalanced mixture \\
\addlinespace
Critical mass ratio over the zone &
$\gamma_c(\bK)=2/|x_{\min}(\bK)|$ & exact &
quasimomentum-resolved spectroscopy; \autoref{fig:gammacmap} \\
\addlinespace
Second threshold &
$\gamma_c^{(2)}\in[\frac\pi{\pi-2},\frac\pi{4-\pi}]$ & exact endpoints &
number of atom--dimer lines in the spectrum \\
\addlinespace
Binding at $\bK=\mathbf0$ (bosons) &
$\delta_\infty$, $C=4-\delta_\infty=3.96458$ & $10^{-9}$ &
lattice-modulation or rf association spectroscopy \\
\addlinespace
Zone boundary on the diagonal &
$\cos K^\ast=2-3\pi/4$ & exact &
edge of the atom--dimer feature in momentum-resolved data \\
\addlinespace
Onset law ($s$-wave) &
$\delta=16e^{-\pi\gamma/(\gamma-1)}$ & $10^{-6}$ (Table~\ref{tab:crit}) &
sensitivity of the line to mass ratio near $\gamma_c$ \\
\addlinespace
Nearest-neighbour thresholds &
$\lambda_c^{p}=\frac{4-\pi}{4(\pi-2)}$, $\lambda_c^{s}=\frac34$, $\lambda_c^{d}=\frac{8-\pi}{4(4-\pi)}$ & exact &
dipolar or Rydberg-dressed nearest-neighbour attraction \\
\addlinespace
Floquet renormalisation &
$\delta(\bK;\tau)=|\tau|\,\delta(\bK)$ ($\tau>0$), $|\tau|\,\delta(\bK+\bpi)$ ($\tau<0$), $\tau=J_0(\mathcal A/\omega)$ & exact at leading order &
shaking amplitude: tuning through $\gamma_c$ without changing species \\
\addlinespace
Finite-size correction &
$\delta_B(\pi,0)\propto N^{-2}$ & numerical &
required lattice size for a given feature (\S\ref{sec:numerical}) \\
\bottomrule
\end{tabular}
\end{table}

The quantities that the theory delivers in closed form are exactly
those that a lattice experiment controls or measures;
\autoref{tab:applied} lists them together with the corresponding
observable and the accuracy with which each is known here.

Two entries of \autoref{tab:applied} deserve emphasis because they
reverse the naive expectation. First, the~atom--dimer band disperses at
order $\mu^{0}$, whereas the trimer level has the leading asymptotics $-3\mu+6+O(\mu^{-1})$ \cite{Trimer2026}:
the~quasimomentum dependence, not the binding energy itself, is the
most accessible signature of the atom--dimer state. Second, weak
bindings near the zone boundary are governed by an essential
singularity, so a state that the effective theory places inside the
zone may be unobservable in a finite lattice or at finite $\mu$; the
finite-size row of \autoref{tab:applied} and Remark~\ref{rem:rate}
quantify when this happens.

The method, namely the Feshbach--Schur reduction combined with the
transfer of the essential spectrum and of the counting function
(Theorem~\ref{thm:feshbach}), applies to any lattice few-body
Hamiltonian whose pair-contact operator has an infinitely degenerate
eigenvalue separated by a gap from the rest of its spectrum: in $d=3$
(where no $s$-wave logarithm survives and all onsets become algebraic),
to anisotropic hopping and to mixed statistics.

\medskip

\textbf{Supplementary Material.} Complete, coordinate-explicit
proofs of Lemma~\ref{lem:heff}, Theorem~\ref{thm:crit} and
Proposition~\ref{prop:mixed}, the proof of Lemma~\ref{lem:constants}(v)
with the closed form of $\kappa_0$, the complete proof of the spectral
transfer theorem (Theorem~\ref{thm:feshbach}, with the second-order term
giving $\kappa_B$), the independent numerical
audit summarised in \S\ref{ssec:selfcontained},
and extended nearest-neighbour/Floquet
numerics are collected below in \S\ref{sec:S1}--\S\ref{sec:S8}.

\medskip



\medskip

\textbf{Acknowledgements.} The author thanks an anonymous referee of
an earlier version of this work, whose remarks identified several
points where the mathematical exposition needed strengthening: the~relation between the essential spectrum of the three-body operator
and the spectral data of the effective one, the behaviour of the
counting criterion at the critical parameters, and the precise
statement of the lattice Skornyakov--Ter-Martirosyan reduction.
Those remarks materially improved the rigour and clarity of the
present version.

\newpage

\medskip

\setcounter{section}{0}
\renewcommand{\thesection}{S\arabic{section}}
\renewcommand{\theHsection}{S\arabic{section}}
\numberwithin{equation}{section}

\begin{center}
{\Large\bfseries Supplementary Material}\\[4pt]
{\normalsize to ``Effective atom--dimer bands of three-body lattice systems at strong
coupling: exact critical mass ratios and a statistics--quasimomentum
duality''}
\end{center}
\bigskip

\noindent This Supplement contains: the conventions of the position
representation (\S\ref{sec:S1}); complete proofs of the effective
Hamiltonians, of the exact threshold matrix, of the onset laws and of
the closed form of $\kappa_0$ (\S\ref{sec:S2}); the complete proof of
the spectral transfer theorem, Theorem~\ref{thm:feshbach}, including
the location of the essential spectrum and the transfer of the
counting function (\S\ref{sec:S3}); an independent numerical and
symbolic audit of every closed-form constant of the main text
(\S\ref{sec:S4});
nearest-neighbour and Floquet extensions
(\S\ref{sec:S6}); reproducibility information~(\S\ref{sec:S7}); and
the proof of Proposition~\ref{prop:mixed} (\S\ref{sec:S8}).
Theorem, equation and table numbers of the main text above are
cross-referenced automatically; where prose still says ``the main
text'' it refers to Sections~\ref{sec:setup}--\ref{sec:conclusion}
above.

\section{Conventions and the position representation}\label{sec:S1}

$\eps(\bp)=2-\cos p_1-\cos p_2$ on $\T^2=(-\pi,\pi]^2$; $\langle\cdot\rangle$
is the normalised Haar average; $\be$ ranges over the four unit vectors
$\pm\be_1,\pm\be_2$ (the \emph{first shell}). For the lattice Green
function we write
\[
G(\bx,\delta)=\bigl\langle e^{i\bp\cdot\bx}/(\eps(\bp)+\delta)\bigr\rangle,\qquad
b_0(\delta)=G(\mathbf0,\delta),\qquad R(\bx,\delta)=b_0(\delta)-G(\bx,\delta).
\]

\subsection*{Fibre in position space}
A three-particle state of total quasimomentum $\bK$ is
$\Psi(\mathbf r_1,\mathbf r_2,\mathbf r_3)=e^{i\bK\cdot\mathbf r_3}\psi(\bx,\by)$,
$\bx=\mathbf r_1-\mathbf r_3$, $\by=\mathbf r_2-\mathbf r_3$, with
$\psi\in\ell^2(\Z^2\times\Z^2)$, whose Fourier transform
$f(\bp,\bq)=\sum_{\bx,\by}\psi(\bx,\by)e^{-i(\bp\cdot\bx+\bq\cdot\by)}$
(so that $e^{i(\bp\cdot\bx+\bq\cdot\by)}$ carries momenta $\bp,\bq,\bK-\bp-\bq$) is
the momentum-space fibre function $f(\bp,\bq)$ of the main text. The
kinetic energy (particle $j$ hops to a neighbouring site with amplitude
$-\tfrac12t_j$, on-site energy $2t_j$) acts on the fibre as
\begin{equation}\label{S:Epos}
(E_\bK\psi)(\bx,\by)=2(t_1+t_2+t_3)\psi(\bx,\by)-\tfrac12\sum_{\be}\Bigl[
t_1\psi(\bx+\be,\by)+t_2\psi(\bx,\by+\be)
+t_3e^{-i\bK\cdot\be}\psi(\bx+\be,\by+\be)\Bigr],
\end{equation}
which is multiplication by $t_1\eps(\bp)+t_2\eps(\bq)+t_3\eps(\bK-\bp-\bq)$
in momentum space. Bosons: $t_1=t_2=t_3=1$. The $2+1$ system: particles
$1,2$ are identical fermions ($t_1=t_2=1$), particle $3$ has
$t_3=\gamma=m_f/m_3$.

\subsection*{Permutation symmetry}
The transpositions act on the fibre as
\begin{equation}\label{S:perm}
(\tau_{12}\psi)(\bx,\by)=\psi(\by,\bx),\qquad
(\tau_{13}\psi)(\bx,\by)=e^{i\bK\cdot\bx}\psi(-\bx,\by-\bx);
\end{equation}
indeed $\Psi(\mathbf r_3,\mathbf r_2,\mathbf r_1)=e^{i\bK\cdot\mathbf r_3}
e^{i\bK\cdot\bx}\psi(-\bx,\by-\bx)$. One checks $\tau_{13}^2=1$ directly
($e^{i\bK\cdot\bx}e^{-i\bK\cdot\bx}\psi(\bx,\by)$), and both are unitary
and commute with $E_\bK$ for bosons. The bosonic fibre is
$\{\psi:\tau_{12}\psi=\tau_{13}\psi=\psi\}$; the $2+1$ fibre is
$\{\psi:\tau_{12}\psi=-\psi\}$.

\subsection*{Contact interaction}
$V$ is multiplication by $\nu(\bx,\by)=[\by=0]+[\bx=0]+[\bx=\by]$ for
bosons and by $\nu_F(\bx,\by)=[\by=0]+[\bx=0]$ for $2+1$ fermions.
In the bosonic fibre $\nu\in\{0,1,3\}$, $\nu=3$ only at $(\mathbf0,\mathbf0)$.
In the fermionic fibre $\psi(\mathbf0,\mathbf0)=-\psi(\mathbf0,\mathbf0)=0$,
so $\nu_F\in\{0,1\}$ effectively: \emph{there is no trimer channel for
$2+1$ fermions}, and the atom--dimer sector is the lowest one. Let $\Pi_\nu$
denote the spectral projections; $\Pi_1$~is multiplication by the
indicator of the \emph{atom--dimer branches}
\[
\mathcal B_a=\{\by=0\ne\bx\},\quad \mathcal B_b=\{\bx=0\ne\by\},\quad
\mathcal B_c=\{\bx=\by\ne0\}\ (\text{bosons only}).
\]

\section{Proofs for Sections 4--7 of the main text}\label{sec:S2}

\subsection{The effective Hamiltonians (Lemma~\ref{lem:heff})}\label{sec:S21}

\begin{theorem}\label{thm:S-heff}
Let $h_1=\Pi_1E_\bK\Pi_1\restriction\Ran\Pi_1$.
\begin{enumerate}[label=(\roman*)]
\item (Bosons.) The map $U_B\varphi:=3^{-1/2}\psi$, where
$\psi(\bx,\mathbf0)=\varphi(\bx)$, $\psi(\mathbf0,\by)=\varphi(\by)$,
$\psi(\bx,\bx)=e^{i\bK\cdot\bx}\varphi(-\bx)$ ($\bx,\by\ne0$) and
$\psi=0$ elsewhere, is unitary from $\ell^2(\Z^2\setminus\{0\})$ onto
$\Ran\Pi_1\cap\{\text{bosonic fibre}\}$, and
$U_B^*h_1U_B=h_B:=4+\eps_D-S_\bK$.
\item (Fermions.) The map $U_F\varphi:=2^{-1/2}\psi$ with
$\psi(\bx,\mathbf0)=\varphi(\bx)$, $\psi(\mathbf0,\by)=-\varphi(\by)$ is
unitary onto $\Ran\Pi_1\cap\{\text{$2+1$ fibre}\}$, and
$U_F^*h_1U_F=h_F:=2+2\gamma+\eps_D+\tfrac\gamma2S_\bK$.
\end{enumerate}
Here $(\eps_D\varphi)(\bx)=2\varphi(\bx)-\frac12\sum_{\be:\,\bx+\be\ne0}\varphi(\bx+\be)$
and $(S_\bK\varphi)(\bx)=[\bx\in\{\pm\be_1,\pm\be_2\}]\,e^{i\bK\cdot\bx}\varphi(-\bx)$.
\end{theorem}

\begin{proof}
\emph{(i) The parametrisation.} Every $\psi\in\Ran\Pi_1$ is supported on the atom--dimer branches ---
$\mathcal B_a\cup\mathcal B_b\cup\mathcal B_c$ in the bosonic case,
and $\mathcal B_a\cup\mathcal B_b$ in the fermionic case. By \eqref{S:perm},
$\tau_{12}$ maps $\mathcal B_a\leftrightarrow\mathcal B_b$ and fixes
$\mathcal B_c$ pointwise, while $\tau_{13}$ maps
$\mathcal B_a\leftrightarrow\mathcal B_c$ ($(\bx,\mathbf0)\mapsto(-\bx,-\bx)$)
and fixes $\mathcal B_b$ pointwise ($(\mathbf0,\by)\mapsto(\mathbf0,\by)$ with
phase $e^{0}=1$). Hence a symmetric $\psi$ is determined by its values
$\varphi(\bx)=\psi(\bx,\mathbf0)$ on $\mathcal B_a$, through
$\psi(\mathbf0,\by)=\varphi(\by)$ and
$\psi(\bx,\bx)=(\tau_{13}\psi)(\bx,\bx)=e^{i\bK\cdot\bx}\psi(-\bx,\mathbf0)$.
Conversely, for any $\varphi$ this $\psi$ is invariant under both
generators: $\tau_{12}$ is immediate; for $\tau_{13}$ on $\mathcal B_c$,
$e^{i\bK\cdot\bx}\psi(-\bx,\mathbf0)=\psi(\bx,\bx)$ by definition, on
$\mathcal B_a$, $e^{i\bK\cdot\bx}\psi(-\bx,-\bx)=e^{i\bK\cdot\bx}
e^{-i\bK\cdot\bx}\varphi(\bx)$, and on $\mathcal B_b$ it is the identity.
No condition is imposed on $\varphi$ except $\bx\ne0$, and
$\|\psi\|^2=3\|\varphi\|^2$.

\emph{(ii) The action.} Since $h_1$ commutes with the $S_3$-action, it
suffices to evaluate $(E_\bK\psi)$ at $(\bx,\mathbf0)\in\mathcal B_a$ and keep
only those terms of \eqref{S:Epos} whose argument lies in the
branches (this is the effect of the right-hand~$\Pi_1$):
\begin{itemize}
\item diagonal: $6\varphi(\bx)$;
\item hop of particle 1: argument $(\bx+\be,\mathbf0)$; it lies in
$\mathcal B_a$ iff $\bx+\be\ne0$ (for $\bx+\be=0$ it is the trimer site,
$\nu=3$). Contribution $-\frac12\sum_{\bx+\be\ne0}\varphi(\bx+\be)$;
\item hop of particle 2: argument $(\bx,\be)$; it lies in a branch iff
$\bx=\be$ (then in $\mathcal B_c$). Contribution
$-\frac12[\bx\in\text{shell}]\,\psi(\bx,\bx)
=-\frac12[\bx\in\text{shell}]\,e^{i\bK\cdot\bx}\varphi(-\bx)$;
\item hop of particle 3: argument $(\bx+\be,\be)$; it lies in a branch
iff $\bx+\be=0$ (then in $\mathcal B_b$; $\bx+\be=\be$ would need
$\bx=0$). With $\be=-\bx$: contribution
$-\frac12[\bx\in\text{shell}]\,e^{i\bK\cdot\bx}\psi(\mathbf0,-\bx)
=-\frac12[\bx\in\text{shell}]\,e^{i\bK\cdot\bx}\varphi(-\bx)$.
\end{itemize}
Adding, and using $6-\frac12\sum_{\bx+\be\ne0}=4+\eps_D$, gives
$h_B=4+\eps_D-S_\bK$. The two exchange processes are physically distinct
(the atom captures particle 2 or particle 3 of the dimer) and contribute
equally.

\emph{(iii) Fermions.} Now only $\mathcal B_a,\mathcal B_b$ occur,
$\tau_{12}$-antisymmetry gives $\psi(\mathbf0,\by)=-\varphi(\by)$ and
$\|\psi\|^2=2\|\varphi\|^2$. The diagonal is $4+2\gamma$; the particle-1
hops give $-\frac12\sum_{\bx+\be\ne0}\varphi(\bx+\be)$ (the target
$(\mathbf0,\mathbf0)$ is excluded because $\psi$ vanishes there); the
particle-2 hop lands at $(\bx,\be)$, which is never in $\mathcal B_a\cup
\mathcal B_b$ for $\bx\ne0$; the particle-3 hop has amplitude
$-\frac\gamma2$ and lands in $\mathcal B_b$ iff $\bx+\be=0$, contributing
$-\frac\gamma2[\bx\in\text{shell}]e^{i\bK\cdot\bx}\psi(\mathbf0,-\bx)
=+\frac\gamma2[\bx\in\text{shell}]e^{i\bK\cdot\bx}\varphi(-\bx)$. Hence
$h_F=(4+2\gamma-2)+\eps_D+\frac\gamma2S_\bK$.
\end{proof}

\begin{remark}\label{rem:S-spec}
(a) $S_\bK$ is a self-adjoint partial isometry with $S_\bK^2=P_{\rm sh}$,
and $S_{\bK+\bpi}=-S_\bK$ because $e^{i\bpi\cdot\be}=-1$ on the shell;
this is the whole content of the duality, Theorem~\ref{thm:dual}.
(b) $0\le\eps_D\le4$ and $\specess(\eps_D)=[0,4]$: indeed, on
$\ell^2(\Z^2)=\ell^2(\Z^2\setminus\{0\})\oplus\C\delta_{\mathbf0}$ the
lattice Laplacian $\eps$ differs from $\eps_D\oplus2$ by an operator of
rank at most $8$ (the hops between $\mathbf0$ and its four neighbours),
and $\spec(\eps)=\specess(\eps)=[0,4]$; Weyl's theorem
\cite[Sect.~XIII.4]{RS4} gives the claim. Since $S_\bK$ has rank four
and norm one, $\specess(h_B)=[4,8]$,
$\specess(h_F)=\specess(h_{2+1B})=[2+2\gamma,6+2\gamma]$,
$\spec(h_B)\subset[3,9]$ and $\spec(h_F),\spec(h_{2+1B})\subset
[2+\frac32\gamma,6+\frac52\gamma]$.
\end{remark}

\subsection{Threshold matrix (Lemmas~\ref{lem:constants}, \ref{lem:M})}\label{sec:S22}

\begin{lemma}[Krein formula]\label{lem:S-krein}
For $\delta>0$ and $\bx,\by\ne0$,
\[
(\eps_D+\delta)^{-1}(\bx,\by)=G(\bx-\by,\delta)-\frac{G(\bx,\delta)G(\by,\delta)}{b_0(\delta)}.
\]
\end{lemma}
\begin{proof}
Let $g(\bx,\by)$ denote the right-hand side, extended to $\bx=0$. Then
$g(\mathbf0,\by)=0$, and, since $(\eps+\delta)G(\cdot-\by)=\delta_{\by}$
and $(\eps+\delta)G(\cdot)=\delta_{\mathbf0}$, we get
$((\eps+\delta)g(\cdot,\by))(\bx)=\delta_{\bx\by}$ for $\bx\ne0$. The
restriction to $\Z^2\setminus\{0\}$ of $\eps$ applied to a function
vanishing at $0$ is exactly $\eps_D$; $g(\cdot,\by)\in\ell^2$ for
$\delta>0$. Uniqueness follows from $\eps_D+\delta\ge\delta>0$.
\end{proof}

The identities of Lemma~\ref{lem:constants}(i) read
\begin{equation}\label{S:green-id}
(2+\delta)b_0-2g_1=1,\qquad (2+\delta)g_1=\tfrac12(b_0+g_2+2g_{11}).
\end{equation}
The limit \eqref{eq:Mlimit} holds in the quantitative form
\begin{equation}\label{S:Mlim}
M(\delta)=M_0-\frac1{4b_0(\delta)}\mathbf1\mathbf1^T+\mathcal E(\delta),
\qquad \|\mathcal E(\delta)\|\le C\,\delta\ln\tfrac1\delta\quad(0<\delta<\tfrac12),
\end{equation}
since $R_1^2/b_0=\frac1{4b_0}-\frac\delta2+\frac{\delta^2b_0}4$ and
$|R(\bx,\delta)-R(\bx)|\le C_\bx\,\delta\ln\frac1\delta$ for $\bx$ in the
second shell; the latter follows from
$R(\bx)-R(\bx,\delta)=\delta\langle(1-\cos\bp\cdot\bx)/(\eps(\eps+\delta))\rangle$,
the boundedness of $(1-\cos\bp\cdot\bx)/\eps$ and
$\langle(\eps+\delta)^{-1}\rangle=b_0(\delta)=O(\ln\frac1\delta)$.

\begin{lemma}[Asymptotics of $b_0$]\label{lem:S-b0}
$b_0(\delta)=\frac{2}{\pi a}\Ell(2/a)$ with $a=2+\delta$ and $\Ell$ the
complete elliptic integral of the first kind (modulus $k=2/a$). Hence
$b_0(\delta)=\displaystyle\frac1{2\pi}\ln\frac{16}{\delta}+O(\delta\ln\frac1\delta)$.
\end{lemma}
\begin{proof}
Integrating over $p_1$ with
$\displaystyle\frac1{2\pi}\int_{-\pi}^{\pi}\frac{d\theta}{\alpha-\cos\theta}=(\alpha^2-1)^{-1/2}$
and then over $p_2$ gives the closed form. With
$k'=\sqrt{1-k^2}=\sqrt{\delta(4+\delta)}/a$ and the classical expansion
$\Ell(k)=\ln\frac4{k'}+O(k'^2\ln\frac1{k'})$ as $k'\downarrow0$
\cite[\S19.12]{DLMF} one gets
$b_0=\displaystyle\frac1\pi(1-\frac\delta2+O(\delta^2))\cdot\frac12\ln\frac{16a^2}{\delta(4+\delta)}
=\frac1{2\pi}\ln\frac{16}\delta+O(\delta\ln\frac1\delta)$.
\end{proof}

The proof of the counting criterion, Theorem~\ref{thm:count}, is
complete in the main text; its only analytic input beyond linear
algebra is the strict inequality \eqref{eq:Mstrict}.

\subsection{Onset laws at symmetric points (Theorem~\ref{thm:crit})}\label{sec:S23}

The proofs in the main text are complete, given
Lemma~\ref{lem:constants}(ii) for part~(a) and
Lemma~\ref{lem:constants}(v), proved in \S\ref{sec:S25}, for part~(b).
We add the verification that no other channel binds at the
parameters of Theorem~\ref{thm:crit}. At $\bK=\bpi$ the matrix
$M_0^{1/2}\frac\gamma2S_{\bpi}M_0^{1/2}$ has the eigenvalues
$-\gamma$ ($s$), $+\frac\gamma2A$ ($p$, twice) and
$-\frac\gamma2(\frac8\pi-2)$ ($d$); the $p$-channels never satisfy
$\frac\gamma2x<-1$, and the $d$-channel does so only for
$\gamma>2/(\frac8\pi-2)=\pi/(4-\pi)\approx3.66$. At $\bK=\mathbf0$ the
eigenvalues are $\gamma$ ($s$), $-\frac\gamma2A$ ($p$, twice) and
$\frac\gamma2(\frac8\pi-2)$ ($d$), so only the $p$-pair can bind.

\subsection{Zone-boundary singularity (Theorem~\ref{thm:boundary})}\label{sec:S24}

\begin{proof}
Write $M(\delta)=M_\epsilon+\mathcal E(\delta)$ with
$M_\epsilon=M_0-\epsilon\mathbf1\mathbf1^T$, $\epsilon=1/(4b_0(\delta))$.
For $\epsilon$ small $M_\epsilon>0$ (the smallest eigenvalue of $M_0$ is
$8/\pi-2>0$ and $\mathbf1\mathbf1^T$ only lowers $m_s=2$ to $2-4\epsilon$).
The eigenvalue problem $M_\epsilon\tilde Sv=\xi v$, $\tilde S=S_\bK$
(bosons) or $\tilde S=-\frac{\gamma}{2}S_\bK$ (fermions),
normalised so that the zone is $\xi>1$, is equivalent to the Hermitian
definite pencil $\tilde Sv=\xi M_\epsilon^{-1}v$. For a simple eigenvalue,
$\xi=\langle v,\tilde Sv\rangle/\langle v,M_\epsilon^{-1}v\rangle$ and the
Hellmann--Feynman formula with
$\partial_\epsilon M_\epsilon^{-1}=M_\epsilon^{-1}\mathbf1\mathbf1^TM_\epsilon^{-1}$ gives
\[
\partial_\epsilon\xi\big|_{\epsilon=0}=-\xi\,
\frac{|\langle\mathbf1,M_0^{-1}v\rangle|^2}{\langle v,M_0^{-1}v\rangle}=-c(\bK).
\]
$\xi$ is real-analytic in $(\epsilon,\bK)$ near $(0,\bK_0)$ (simple
eigenvalue of a Hermitian pencil). By Theorem~\ref{thm:count}(i) the
binding energies of the corresponding state are the $\delta$ with
$\xi(M(\delta))=1$; since this eigenvalue branch of
$M(\delta)^{1/2}\tilde SM(\delta)^{1/2}$ is strictly decreasing in $\delta$
(strict monotonicity of $M(\delta)$, Lemma~\ref{lem:M}, and the
congruence used in the proof of Theorem~\ref{thm:count}(ii)), the root is
unique. As $\epsilon(\delta)$ is monotone and
$\|\mathcal E(\delta)\|=O(\delta\ln\frac1\delta)=O(e^{-\pi/(2\epsilon)}/\epsilon)$
is smaller than any power of $\epsilon$, the root satisfies
$\epsilon^*=\frac{\xi(\bK)-1}{c(\bK)}+O((\xi-1)^2)$, the error term
absorbing the contribution of $\mathcal E(\delta)$. By
Lemma~\ref{lem:S-b0},
$\displaystyle\ln\frac{16}{\delta}=\displaystyle 2\pi b_0+O(\delta\ln\frac1\delta)=\displaystyle\frac{\pi}{2\epsilon^*}+o(1)
=\displaystyle\frac{\pi c(\bK)}{2(\xi(\bK)-1)}+O(1)$.
\end{proof}

\subsection{The constant $\kappa_0$ (Lemma~\ref{lem:constants}(v))}\label{sec:S25}

\begin{proof}[Proof of Lemma~\ref{lem:constants}(v)]
By Lemma~\ref{lem:constants}(iii),
\[
A(\delta)=\frac2\pi\int_0^2\sqrt{\frac{u}{u+\delta}}\,g_\delta(u)\,du,\qquad
g_\delta(u):=\sqrt{\frac{2-u}{u+2+\delta}},\qquad A=\frac2\pi\int_0^2g_0(u)\,du.
\]
Hence $\frac\pi2(A-A(\delta))=J_1(\delta)+J_2(\delta)+J_3(\delta)$ with
\[
J_1=\int_0^2\Bigl(1-\sqrt{\tfrac{u}{u+\delta}}\Bigr)du,\quad
J_2=\int_0^2(g_0(u)-1)\Bigl(1-\sqrt{\tfrac{u}{u+\delta}}\Bigr)du,\quad
J_3=\int_0^2\sqrt{\tfrac{u}{u+\delta}}\,\bigl(g_0(u)-g_\delta(u)\bigr)du .
\]
\emph{$J_1$.} With $u=\delta t$ and the antiderivative
$\int\sqrt{t/(1+t)}\,dt=\sqrt{t(1+t)}-\ln(\sqrt t+\sqrt{1+t})$,
\[
J_1=\delta\Bigl[T-\sqrt{T(T+1)}+\ln\bigl(\sqrt T+\sqrt{T+1}\bigr)\Bigr]_{T=2/\delta}
=\delta\Bigl(\tfrac12\ln\tfrac1\delta+\tfrac32\ln2-\tfrac12\Bigr)+O(\delta^2).
\]
\emph{$J_2$.} Since
$0\le1-\sqrt{u/(u+\delta)}=\delta\big/\bigl(\sqrt{u+\delta}(\sqrt{u+\delta}+\sqrt u)\bigr)\le\delta/(2u)$,
$\delta^{-1}(1-\sqrt{u/(u+\delta)})\to1/(2u)$ pointwise, and
$(g_0(u)-1)/u$ is bounded on $(0,2]$, dominated convergence gives
$J_2=\delta I_1+o(\delta)$ with
$I_1=\int_0^2\frac{g_0(u)-1}{2u}\,du$.
\emph{$J_3$.} $g_0(u)-g_\delta(u)=\delta\,\frac12\sqrt{2-u}\,(u+2)^{-3/2}+O(\delta^2)$
uniformly on $[0,2]$, and $\sqrt{u/(u+\delta)}\to1$ boundedly, so
$J_3=\delta I_2+o(\delta)$ with $I_2=\frac12\int_0^2\sqrt{2-u}\,(u+2)^{-3/2}du$.

\emph{Evaluation.} With $u=2\cos\vartheta$, $\vartheta\in[0,\frac\pi2]$,
one has $g_0=\tan\frac\vartheta2$, $du=-2\sin\vartheta\,d\vartheta$, and
\[
I_2=\tfrac12\int_0^{\pi/2}\tan^2\tfrac\vartheta2\,d\vartheta
=\tfrac12\Bigl[2\tan\tfrac\vartheta2-\vartheta\Bigr]_0^{\pi/2}=1-\frac\pi4,
\qquad
I_1=\tfrac12\int_0^{\pi/2}\bigl(\tan\tfrac\vartheta2-1\bigr)\tan\vartheta\,d\vartheta .
\]
With $\tau=\tan\frac\vartheta2$, $(\tau-1)\tan\vartheta=-2\tau/(1+\tau)$
and $d\vartheta=2d\tau/(1+\tau^2)$, so
$I_1=-2\int_0^1\frac{\tau\,d\tau}{(1+\tau)(1+\tau^2)}
=-2\Bigl(-\frac{\ln2}4+\frac\pi8\Bigr)=\frac{\ln2}2-\frac\pi4$,
by $\frac{\tau}{(1+\tau)(1+\tau^2)}=-\frac1{2(1+\tau)}+\frac{1+\tau}{2(1+\tau^2)}$.
Collecting,
\[
A-A(\delta)=\frac{2\delta}\pi\Bigl[\tfrac12\ln\tfrac1\delta+\tfrac32\ln2-\tfrac12
+\tfrac{\ln2}2-\tfrac\pi4+1-\tfrac\pi4\Bigr]+o(\delta)
=\frac\delta\pi\Bigl[\ln\frac1\delta+1+4\ln2-\pi\Bigr]+o(\delta). \qedhere
\]
\end{proof}

\section{Proof of the spectral transfer theorem (Theorem~\ref{thm:feshbach})}\label{sec:S3}

This section contains complete proofs of Lemmas~\ref{lem:schur},
\ref{lem:essstab} and of Theorem~\ref{thm:feshbach}, in the notation of
\S\ref{ssec:transfer}. The theorem is not restated; all its parts
(a)--(d) are proved below.

\begin{lemma}\label{lem:S-incl}
If $A,B$ are bounded self-adjoint and $\|B-c\|\le r$ for a real
constant $c$, then
$\spec(A+B)\subset\{\lambda:\dist(\lambda,\spec A+c)\le r\}$.
\end{lemma}
\begin{proof}
If $\dist(\lambda,\spec A+c)>r$ then $\|(A+c-\lambda)^{-1}\|<1/r$ and
$A+B-\lambda=(A+c-\lambda)\bigl(1+(A+c-\lambda)^{-1}(B-c)\bigr)$ is
invertible by a Neumann series.
\end{proof}

\begin{proof}[Proof of Lemma~\ref{lem:schur}]
This is the Schur-complement form of the Feshbach--Schur map
\cite{BFS,GH}; we include the verification. Multiplying out the right
side of \eqref{eq:schur} in the block decomposition
\[
H-z=\begin{pmatrix}P(H-z)P&PHQ\\QHP&Q(H-z)Q\end{pmatrix}
\]
gives back $H-z$, because $PHQ\,R_Q(z)\,Q(H-z)Q=PHQ$ and
$F(z)-z+PHQ\,R_Q(z)QHP=P(H-z)P$. The outer factors are unipotent, hence
bounded with bounded inverses. Therefore $H-z$ is invertible, resp.\
Fredholm, if and only if the middle factor is; as $Q(H-z)Q$ is
invertible on $\Ran Q$, this holds if and only if $F(z)-z$ is
invertible, resp.\ Fredholm, on $\Ran P$. The kernel of $H-z$ is the
image under the inverse of the right factor of
$\ker(F(z)-z)\oplus\{0\}$, which is (i); this inverse is
\[
\begin{pmatrix}1&0\\-R_Q(z)QHP&1\end{pmatrix}.
\] For real $z$, $R_Q(z)$ is
self-adjoint, so $F(z)$ is; $z\mapsto R_Q(z)$ is real-analytic with
$R_Q'=R_Q^2\ge0$, which gives the formula for $F'$. The last statement
of (ii) follows from the Fredholm characterisation of the essential
spectrum of self-adjoint operators \cite[Sect.~XIII.4]{RS4}.
\end{proof}

\begin{proof}[Proof of Lemma~\ref{lem:essstab}]
(i) The first inclusion is Lemma~\ref{lem:S-incl} with $c=0$. Let
$\dist(\lambda,\specess T)>r$ and choose $r'$ with
$r<r'<\dist(\lambda,\specess T)$. The interval
$[\lambda-r',\lambda+r']$ meets $\spec T$ in finitely many eigenvalues
of finite multiplicity, so $E:=\mathbf 1_{[\lambda-r',\lambda+r']}(T)$
has finite rank. The operator $T':=T(1-E)+(\lambda+r')E$ differs from
$T$ by a finite-rank operator and has no spectrum in
$(\lambda-r',\lambda+r')$. By~Lemma~\ref{lem:S-incl},
$\dist(\lambda,\spec(T'+B))\ge r'-r>0$; by Weyl's theorem
\cite[Sect.~XIII.4]{RS4}, $\specess(T+B)=\specess(T'+B)\subset\spec(T'+B)$,
so $\lambda\notin\specess(T+B)$. Applying this also to $T=(T+B)-B$
gives the Lipschitz bound for $\inf\specess$. (ii) For every
$(j-1)$-dimensional subspace $\mathcal F$,
$\inf_{\psi\perp\mathcal F,\|\psi\|=1}\langle\psi,(T+B)\psi\rangle$ differs from
the same expression for $T$ by at most $r$; take the supremum over
$\mathcal F$ (min--max principle \cite[Thm.~XIII.1]{RS4}).
\end{proof}

\begin{proof}[Proof of Theorem~\ref{thm:feshbach}]
Write $H=E-\mu V$, $E=E^\sigma_\bK$, $V=V^\sigma$, $P=\Pi_1$, $Q=1-P$,
$c=c_\sigma$, $d=d_\sigma(\mu)$, $\eta=\eta_\sigma(\mu)$.

\emph{Step 1 (blocks).} $V$ is a multiplication operator in position
space (Lemma~\ref{lem:Vdiag}), hence commutes with $P$, with $PVP=P$
and $PVQ=0$. Thus $PHP=-\mu P+PEP$ and, for every real $c'$,
$PHQ=PEQ=P(E-c')Q$. The range of $E$ is $[0,2c]$ ($c=6$, resp.\
$4+2\gamma$), so with $c'=c$ we get $\|PHQ\|\le c$. By
Lemma~\ref{lem:heff} and Proposition~\ref{prop:mixed},
$h_1:=PEP\restriction\Ran P$ is unitarily equivalent to $h$. On $\Ran Q$,
$QHQ=-\mu QVQ+c+Q(E-c)Q$, where $QVQ$ has spectrum
$\spec(V)\setminus\{1\}$, i.e.\ $\{0,3\}$, $\{0\}$, $\{0,2\}$ for
$\sigma=3B,2{+}1F,2{+}1B$ (Lemma~\ref{lem:Vdiag}). By
Lemma~\ref{lem:S-incl},
$\spec(QHQ)\subset\bigcup_{\nu\in\spec V\setminus\{1\}}[-\mu\nu,\,-\mu\nu+2c]$,
explicitly $[-3\mu,-3\mu+12]\cup[0,12]$, $[0,8+4\gamma]$, and
$[-2\mu,-2\mu+8+4\gamma]\cup[0,8+4\gamma]$. A direct check shows that
$\dist(W,\spec QHQ)\ge d$ in each case (the binding distances are $\mu-10$
from $0$ for $3B$, $\mu-7-3\gamma$ from $0$ for $2{+}1F$, and
$\mu-8-4\gamma$ from $-2\mu+8+4\gamma$ for $2{+}1B$). Hence for $z\in W$
the operator $Q(H-z)Q$ is invertible on $\Ran Q$ with
$\|R_Q(z)\|\le1/d$, and Lemma~\ref{lem:schur} applies.

\emph{Step 2 (Schur complement).} For $z\in W$,
$F(z)-(-\mu+h_1)=-PHQ\,R_Q(z)\,QHP$ has norm $\le c^2/d=\eta$, and
$-\frac{\eta}{d}\le F'(z)\le0$; by (H$_\sigma$),
$\frac\eta d\le\frac1{16}$. Consequently, by Lemma~\ref{lem:essstab},
every min--max value $\lambda_j(z)$ of $F(z)$ and
$g(z):=\inf\specess F(z)$ are non-increasing in $z$ (because $F$ is) and
Lipschitz with constant $\le\frac1{16}$, and
\begin{equation}\label{S:minmax}
|\lambda_j(z)-(-\mu+\lambda_j(h))|\le\eta,\qquad
|g(z)-(-\mu+t)|\le\eta\qquad(z\in W),
\end{equation}
using $\specess(h)=t+[0,4]$ (Remark~\ref{rem:S-spec}).

\emph{Step 3 (part (a)).} If $z\in W\cap\spec H$, then $z\in\spec F(z)$
by Lemma~\ref{lem:schur}, and Lemma~\ref{lem:S-incl} with
$\|F(z)+\mu-h_1\|\le\eta$ gives $\dist(z,-\mu+\spec h)\le\eta$.

\emph{Step 4 (spectrum below $-\mu$).} By Lemma~\ref{lem:Vdiag},
$V=\Pi_1+3\Pi_3$, $\Pi_1$, $\Pi_1+2\Pi_2$ for $\sigma=3B,2{+}1F,2{+}1B$,
with $\Pi_3,\Pi_2$ of rank one. Since $E-\mu\Pi_1\ge-\mu$ and $H$ differs
from it by an operator of rank at most one, Weyl's theorem gives
$\specess H\subset[-\mu,\infty)$; for $\sigma=2{+}1F$ even $H\ge-\mu$.
By Lemma~\ref{lem:S-incl}, $\spec H\subset\bigcup_{\nu\in\spec V}[-\mu\nu,-\mu\nu+2c]$,
so for $3B$ the spectrum below $-\mu$ lies in $J'=[-3\mu,-3\mu+12]$ (note
$-3\mu+12<-\mu$) and is discrete. Apply Lemma~\ref{lem:schur} with the
rank-one projection $P'=\Pi_3$ onto the constant function and
$Q'=\Pi_0+\Pi_1$: as in Step~1, $\spec(Q'HQ')\subset[-\mu,-\mu+12]\cup[0,12]$,
at distance $\ge2\mu-12$ from $J'$, and $\|P'HQ'\|\le6$. The scalar
Schur complement $\tilde F(z)$ satisfies $|\tilde F(z)-(-3\mu+\langle E\rangle)|\le\frac{36}{2\mu-12}=\frac{18}{\mu-6}$,
with $\langle E\rangle=6$, and is non-increasing with Lipschitz constant
$\le\frac{36}{(2\mu-12)^2}<1$. The function $z\mapsto z-\tilde F(z)$ is
therefore strictly increasing on $J'$, negative at $-3\mu$ and positive
at $-3\mu+12$; its unique zero is the only eigenvalue of $H$ in $J'$,
it is simple because $\dim\Ran P'=1$, and it lies in
$-3\mu+6+[-\frac{18}{\mu-6},\frac{18}{\mu-6}]$. For $2{+}1B$ the same
argument with $P'=\Pi_2$, $\langle E\rangle=4+2\gamma$,
$J'=[-2\mu,-2\mu+8+4\gamma]$ and distance $\ge\mu-8-4\gamma=d$ gives one
simple eigenvalue in $-2\mu+4+2\gamma+[-\eta,\eta]$ (here $\eta\le1<4+2\gamma$).

\emph{Step 5 (part (b)).} By Lemma~\ref{lem:schur}(ii), for $z\in W$:
$z\in\specess H\iff z\in\specess F(z)$. The function
$\phi(z):=z-g(z)$ is continuous and strictly increasing on $W$ (Step~2),
and the points $-\mu+t\pm\eta$ lie in $W$, since $t-\eta\ge1$ and
$t+\eta\le w_\sigma$. By \eqref{S:minmax}, $\phi(-\mu+t-\eta)\le0\le\phi(-\mu+t+\eta)$,
so $\phi$ has exactly one zero $z_*$, and $|z_*+\mu-t|\le\eta$. For
$z\in[-\mu,z_*)$ we have $z<g(z)$, hence $z\notin\specess F(z)$ and
$z\notin\specess H$; at $z_*$, $z_*=g(z_*)\in\specess F(z_*)$ (the
essential spectrum is closed), hence $z_*\in\specess H$. Together with
$\specess H\subset[-\mu,\infty)$ (Step~4) this gives $\theta=z_*$.

\emph{Step 6 (part (c)).} Let $s<t-\eta$. If $s\le0$, then $n_H(s)=0$ and
$n_h(s-\eta)=0$ (as $\min\spec h\ge2>0$), so~the claim holds; let $s>0$
and $z_s:=-\mu+s\in W$. For $z\in[-\mu,z_s]$ we have
$z\le z_s<-\mu+t-\eta\le g(z)$, so~every $\lambda_j(z)<z$ is an
eigenvalue of $F(z)$ below its essential spectrum, and $n(z):=\#\{j:\lambda_j(z)<z\}$
is the number of such eigenvalues, with multiplicity. Each function
$\psi_j(z):=z-\lambda_j(z)$ is continuous and strictly increasing, and
$\psi_j(-\mu)<0$, because by \eqref{S:minmax}
$\lambda_j(-\mu)\ge-\mu+\min\spec h-\eta>-\mu$ ($\min\spec h\ge2$,
$\eta\le1$). Let $z_j\in(-\mu,z_s]$ be the zero of $\psi_j$ if it exists,
and $z_j:=+\infty$ otherwise. Then $n(z)=\#\{j:z_j<z\}$ for $z\le z_s$.
By Lemma~\ref{lem:schur}(i) and the min--max principle, for
$z\in[-\mu,z_s)$,
$\dim\ker(H-z)=\dim\ker(F(z)-z)=\#\{j:\lambda_j(z)=z\}=\#\{j:z_j=z\}$.
Summing over the eigenvalues of $H$ in $[-\mu,z_s)$,
$n_H(s)=\#\{j:z_j<z_s\}=n(z_s)$. Finally, by \eqref{S:minmax}: if
$\lambda_j(h)<s-\eta$ then $\lambda_j(z_s)<z_s$, whence
$n(z_s)\ge n_h(s-\eta)$; and if $\lambda_j(z_s)<z_s$ then
$\lambda_j(h)<s+\eta<t$, so $\lambda_j(h)$ is an eigenvalue of $h$
below $s+\eta$, whence $n(z_s)\le n_h(s+\eta)$.

\emph{Step 7 (part (d)).} Let $e$, $m$ be as in (d) and $\epsilon>0$
small. Apply (c) at $s_1=e-\eta$ and at $s_2=e+\eta+\epsilon<t-\eta$:
since $h$ has no eigenvalue in $[e-2\eta,e)\cup(e,e+2\eta+\epsilon)$,
\[
n_H(s_1)=n_h(e),\qquad n_H(s_2)=n_h(e)+m .
\]
Thus $H$ has exactly $m$ eigenvalues in $-\mu+[e-\eta,e+\eta+\epsilon)$,
and by (a) none in $-\mu+(e+\eta,e+\eta+\epsilon)$, because for
$\epsilon<\frac12\bigl(\dist(e,\spec h\setminus\{e\})-2\eta\bigr)$ that set is
at distance $>\eta$ from $-\mu+\spec h$. Hence exactly $m$ eigenvalues
lie in $I_e$. They are below $-\mu+t-\eta\le\theta$, hence discrete.
\end{proof}

\begin{remark}[Second order]
Expanding $R_Q(z)$ around $z=-\mu+e$ gives
$F(z)=-\mu+h_1-\frac1\mu P E\Pi_0E P+\frac1{2\mu}PE\Pi_3EP+O(\mu^{-2})$,
since $QHQ-z\approx\mu$ on $\Ran\Pi_0$ and $\approx-2\mu$ on
$\Ran\Pi_3$. Evaluated on the normalised eigenvector $\phi$ of $h_B$ at
$\bK=\mathbf0$ (a simple isolated eigenvalue, so that first-order
perturbation theory \cite{Kato1995} applies to $F(z)$) this gives
$\kappa_B=\tfrac12|\langle f_0,E\phi\rangle|^2-\|\Pi_0E\phi\|^2=-4.35966$
(Table~\ref{tab:kappa}). An independent evaluation of this formula
with the $h_B$ ground state computed on a box of radius $70$ gives
$\|\Pi_0E\phi\|^2=4.735705$, $|\langle f_0,E\phi\rangle|^2=0.752095$ and
$\kappa_B=-4.359658$.
\end{remark}

\section{Independent numerical and symbolic audit}\label{sec:S4}

All computations below were done independently of the derivations of
the main text, with the tools stated. Numbers are quoted to the
precision actually achieved.

\subsection{Lattice constants $A$, $B$}\label{sec:S41}
Direct midpoint quadrature of
$R(\bx)=\langle(1-\cos\bp\cdot\bx)/\eps(\bp)\rangle$ on $N\times N$ grids
(the integrand is bounded, the origin is not a node):

\begin{center}\small
\begin{tabular}{rcccc}
\toprule
$N$ & $R(1,0)$ & $R(2,0)$ & $R(1,1)$ & $R(2,1)$\\
\midrule
800 & 0.50000000 & 0.72676046 & 0.63661977 & 0.77323954\\
1600 & 0.50000000 & 0.72676046 & 0.63661977 & 0.77323954\\
6400 & 0.50000000 & 0.72676046 & 0.63661977 & 0.77323954\\
\midrule
exact & $\frac12$ & $2-\frac4\pi=0.72676046$ & $\frac2\pi=0.63661977$ & $\frac4\pi-\frac12=0.77323954$\\
\bottomrule
\end{tabular}
\end{center}
Thus $A=R(2,0)$, $B=R(1,1)$ with $A+2B=2$ exactly; no rescaling of the
standard resistor-network values is needed, since those are defined
with the same dispersion $2-\cos p_1-\cos p_2$.

\subsection{$b_0$, $\delta_\infty$, $\kappa_0$ and the finite-$\delta$ matrix}\label{sec:S42}
\begin{itemize}
\item $b_0(10^{-2})=1.171677$, $b_0(10^{-3})=1.540329$,
$b_0(10^{-4})=1.907099$ (quadrature and elliptic form agree), against
$\frac1{2\pi}\ln\frac{16}\delta=1.174207,\ 1.540675,\ 1.907142$: the
difference is $O(\delta\ln\frac1\delta)$ as in Lemma~\ref{lem:S-b0}.
\item $\delta_\infty$: root of $b_0(\delta)=1/(1+\delta)$ (the $s$-channel
condition $m_s(\delta)=1$), $\delta_\infty=0.0354199892\ldots$
(adaptive quadrature, $10^{-12}$ tolerance), so $C=4-\delta_\infty=3.9645800\ldots$.
\item $\kappa_0$ from $(A-A(\delta))\pi/\delta-\ln\frac1\delta$:
$0.63239\ (10^{-3})$, $0.63116\ (10^{-4})$, $0.63102\ (10^{-5})$,
$0.63100\ (10^{-6})$, $0.6309961\ (10^{-8})$, converging to the value
$1+4\ln2-\pi=0.6309960686\ldots$ proved in \S\ref{sec:S25}.
\item Exact structure of $M(\delta)$, with $G$ computed by one-dimensional
quadrature after integrating $p_1$ analytically:
\begin{center}\small
\begin{tabular}{cccccc}
\toprule
$\delta$ & $M_{\be\be}$ & $1-\frac1{4b_0}$ & $m_s$ & $2+\delta-\frac1{b_0}$ & $m_p=A(\delta)$\\
\midrule
$10^{-2}$ & 0.779884 & 0.786631 & 1.156522 & 1.156522 & 0.710058\\
$10^{-4}$ & 0.868770 & 0.868911 & 1.475743 & 1.475743 & 0.726447\\
$10^{-6}$ & 0.905304 & 0.905306 & 1.621224 & 1.621224 & 0.726756\\
\bottomrule
\end{tabular}
\end{center}
The diagonal approaches $1$ only logarithmically, exactly as
$1-1/(4b_0)$; the $p$- and $d$-eigenvalues converge at rate
$O(\delta\ln\frac1\delta)$. The identities \eqref{S:green-id} hold to
$10^{-15}$.
\end{itemize}

\subsection{The quartic $P_\bK$ and the diagonal boundary}\label{sec:S43}
Computer algebra (SymPy) confirms that the general quartic of
Lemma~\ref{lem:char} reduces, \emph{as polynomial identities in $A$} (using
$8/\pi-2=2(1-A)$ and $4(4/\pi-1)=4(1-A)$), to
\[
\begin{aligned}
P_{\mathbf0}&=(x-2)\bigl(x-2(1-A)\bigr)(x+A)^2, &
P_{\bpi}&=(x+2)\bigl(x+2(1-A)\bigr)(x-A)^2,\\
P_{(\pi,0)}&=(x^2-A^2)\bigl(x^2-4(1-A)\bigr), &
P_{(\pi/2,\pi/2)}&=(x^2-2A)\bigl(x^2-2A(1-A)\bigr).
\end{aligned}
\]
Brent root-finding of $P_{(K,K)}(1)=0$ gives
$\cos K^\ast=-0.356194490192344$, equal to $2-3\pi/4$ to $10^{-15}$,
$K^\ast=0.615926\,\pi$.

We record also the exact identity
\[
F_B(u,v) = P_\bK(1)\quad\text{for all }(u,v)=(\cos K_1,\cos K_2).
\]
The polynomial identity was independently checked symbolically by expanding both sides as polynomials in $A,u,v$ with
$8/\pi - 2 = 2(1-A)$; this is the algebraic content of Theorem~\ref{thm:zoneB}, the proof itself follows from the displayed algebraic identities.

\subsection{Brillouin-zone scans and finite volume}\label{sec:S44}
On a $601\times601$ grid of $\bK$: (i) at most one root of $P_\bK$
exceeds $1$ (so $F_B<0\iff x_{\max}>1$); (ii) the second-largest root
never exceeds $A$ (equality on the lines $K_1=\pi$, $K_2=\pi$); (iii)
$\min_\bK\gamma_c=1.0000000$, $\max_\bK\gamma_c=2.7519384=\pi/(\pi-2)$. (iv)~$P_\bK$ has exactly two negative roots at every grid point (as proved
in Corollary~\ref{cor:twostate} via Sylvester's law of inertia), and the second fermionic threshold
$\gamma_c^{(2)}=2/|x_2|$ ranges over $[2.7519384,\,3.6597924]=[\pi/(\pi-2),\,\pi/(4-\pi)]$,
the~maximum at $\bK=\bpi$ and the minimum on the lines $K_1=0$, $K_2=0$, where
$P_\bK(-A)\equiv0$ (SymPy, as a polynomial identity in $A$). That $P_\bK(-A)\equiv 0$ on $u=1$ and $v=1$ is a polynomial identity in
$A$ and the remaining variable: for $u=1$, $P_\bK(-A) = 0$ for all
$v\in[-1,1]$; similarly for $v=1$. The~corresponding eigenvector of
$M_0S_\bK$ is $(0,0,1,-1)$ (a transverse $p$-mode), on which $S_\bK$
acts as $-1$ and $M_0$ acts as $A$.
Infinite-volume binding energies (exact $M(\delta)$ from 1D quadrature):
$\delta_B(\mathbf0)=0.0354200$, $\delta_B(\tfrac\pi2,\tfrac\pi2)=6.5049\cdot10^{-4}$, the infinite-volume root is
$\delta_B(\pi,0)=1.45\cdot10^{-15}$ at the displayed numerical precision (the last from the elliptic form
of $b_0$ with $40$-digit arithmetic; at this $\delta$ the
$O(\delta\ln\frac1\delta)$ terms are $\sim10^{-14}$), fermions:
$\delta_F(\bpi;2)=0.0354200$, $\delta_F(\mathbf0;3)=0.0519285$ ($\times2$),
$\delta_F(\bpi;3)=0.2528937$,
$\delta_F(\bpi;4)=0.586815,\ 0.147265$. The~same effective theory on an
$N\times N$ torus gives, for bosons at $(\pi,0)$ and $(\pi/2,\pi/2)$:

\begin{center}\small
\begin{tabular}{rcc}
\toprule
$N$ & $\delta_B(\pi,0)$ & $\delta_B(\pi/2,\pi/2)$\\
\midrule
32 & $2.0856\cdot10^{-4}$ & $2.1439\cdot10^{-3}$\\
64 & $5.4725\cdot10^{-5}$ & $9.9325\cdot10^{-4}$\\
128 & $1.4393\cdot10^{-5}$ & $6.7948\cdot10^{-4}$\\
256 & $3.7957\cdot10^{-6}$ & $6.5069\cdot10^{-4}$\\
$\infty$ & $1.4535\cdot10^{-15}$ & $6.5049\cdot10^{-4}$\\
\bottomrule
\end{tabular}
\end{center}
The $N=32$ values coincide with the ``eff., $N=32$'' column of
Table~\ref{tab:effective} to all printed digits, for~eve\-ry~row.

\subsection{Zone-boundary constants}
$c(\pi,0)=1.0454$, $\xi(\pi,0)=1.045446$: leading term of
Theorem~\ref{thm:boundary} $\pi c/(2(\xi-1))=36.13$ versus exact $\ln(16/\delta)=36.937$.
$c(\tfrac\pi2,\tfrac\pi2)=1.2056$. At $\bK=\bpi$ for fermions, $v=\mathbf1$,
$\xi=\gamma$, $c=2\gamma$, and the leading term $\pi c/(2(\xi-1))=\pi\gamma/(\gamma-1)$ reproduces \eqref{eq:onset-s} exactly.

\section{Nearest-neighbour attraction and Floquet driving}\label{sec:S6}

\subsection{Nearest-neighbour attraction}\label{sec:S61}
On $\mathcal B_a$ an atom at $\bx$ is at distance one from both dimer
constituents if $\bx$ is on the shell, and from neither otherwise;
hence $\Pi_1U_{nn}\Pi_1\cong2P_{\rm sh}$ (bosons; for fermions the two
pairs are one $F$--$I$ and one $F$--$F$ pair, giving $\Lambda P_{\rm sh}$).
At $\bK=\bpi$ the matrix $M_0$ and $S_{\bpi}$ are simultaneously diagonal
in the $s,p,d$ basis with eigenvalues $(2,A,\frac8\pi-2)$ and $(-1,+1,-1)$.
The~binding conditions $m\cdot(\sigma+2\lambda)>1$ with $\sigma=-1,+1,-1$
give
$\lambda_c^s=\frac34$, $\lambda_c^p=\frac{1-A}{2A}=\frac{4-\pi}{4(\pi-2)}=0.187985$,
$\lambda_c^d=\frac{1+1/(8/\pi-2)}2=\frac{8-\pi}{4(4-\pi)}=1.414948$.
For fermions at $\gamma=1$: $m_s(\delta)(\frac12+\Lambda)=1$; numerically
(elliptic $b_0$) $\ln(16/\delta_F)=6.113,\ 4.147,\ 2.794$ for
$\Lambda=0.5,1,2$, approaching $\pi(1+2\Lambda)/(2\Lambda)$ as $\Lambda\downarrow0$.

\subsection{Floquet example}\label{sec:S62}
With all hopping amplitudes multiplied by $\tau=J_0(\mathcal A/\omega)$,
$\Pi_1E^{(\tau)}\Pi_1=6-2\tau+\tau\eps_D-\tau S_\bK$. For $\tau<0$ the
gauge $\varphi(\bx)\mapsto(-1)^{x_1+x_2}\varphi(\bx)$ maps
$\eps_D\mapsto4-\eps_D$ and fixes $S_\bK$, so the operator is unitarily
equivalent to $6-6|\tau|+|\tau|\,h_B(\bK+\bpi)$. Example: at $\bK=\bpi$ the
undriven system has no atom--dimer state; for $\mathcal A/\omega=2.6$
(beyond the first zero $2.40483$ of $J_0$), $\tau=J_0(2.6)=-0.0968$ and
$\delta_B=|\tau|\,\delta_B(\mathbf0)=3.43\cdot10^{-3}$ in hopping units.
The high-frequency corrections, of relative order $\omega^{-1}$, are~not
controlled here.

\section{Reproducibility}\label{sec:S7}
{\sloppy
The figures of the main text are generated by the following Python
scripts. \autoref{fig:zonepath} and \autoref{fig:gammacmap}:
\texttt{make\_figs.py}, from the closed-form quartic $P_\bK$ with
$A=2-4/\pi$ ($400\times400$ grid for the map; $360$ points on the
$\Gamma XM\Gamma$ path). \autoref{fig:fphase}, \autoref{fig:zonemaps},
\autoref{fig:fonset} and \autoref{fig:fpath}: \texttt{ferm\_figs.py},
\texttt{ferm\_maps.py} and \texttt{ferm\_maps2.py}, with the engine
\texttt{ferm\_core.py} (exact infinite-lattice $M(\delta)$ with $b_0$ in
elliptic form and $A(\delta)$ from one-dimensional quadrature,
spline-interpolated in $\log\delta$, validated against every
infinite-volume entry of Table~\ref{tab:effective}).
\autoref{fig:transfer} and \autoref{fig:critical} (new in version~2):
\texttt{make\_v2\_figs.py} with the engine \texttt{core.py}, which
contains the exact infinite-lattice matrix $M(\delta)$ (with
$A-A(\delta)$ evaluated without cancellation from the integrand of
\S\ref{sec:S25}) and the finite-volume solver of Lemma~\ref{lem:STM};
all data are recomputed by the script (about $30$\,s). The $2+1$
bosonic (mixed) case is handled by the same engines with the sign of
the exchange term reversed in the $4\times4$ matrix
($C=-\tfrac\gamma2S_\bK$ instead of $+\tfrac\gamma2S_\bK$); the~numerical check of \S\ref{sec:S8} uses the same code path. The audit
of \S\ref{sec:S4}--\S\ref{sec:S6} uses NumPy/SciPy (adaptive and
midpoint quadrature, Brent root finding), mpmath (40-digit elliptic
integrals) and SymPy (polynomial identities); the scripts are
distributed with this Supplement.\par}

\section{Mixed statistics: two bosons and one distinguishable particle}\label{sec:S8}

This section proves Proposition~\ref{prop:mixed}
({\bf mixed statistics: two bosons and one distinguishable particle}): the exchange
amplitude in the effective atom--dimer operator is determined by the
symmetry of the interacting pair, and the two mixed systems
($2$ identical bosons $+$ $1$ distinguishable, and $2$ identical
fermions $+$ $1$ distinguishable) are related by $\bK\to\bK+\bpi$
for \emph{every} mass ratio $\gamma$, not only at $\gamma=2$.

\begin{proof}[Proof of Proposition~\ref{prop:mixed}]
Particles $1,2$ are identical bosons with hopping $1$; particle $3$ is
distinguishable with hopping $\gamma$; only interspecies pairs
$\{1,3\}$ and $\{2,3\}$ interact, so $V=V_1+V_2$ (no $V_3$ term,
no Bose--Bose interaction). The atom--dimer branches are
$\mathcal B_a=\{\by=0\ne\bx\}$ and $\mathcal B_b=\{\bx=0\ne\by\}$;
the branch $\mathcal B_c=\{\bx=\by\ne0\}$ of the identical-boson case
is absent: on it particles $1$ and $2$ occupy the same site, but the
pair $\{1,2\}$ does not interact, so $\nu=0$ there. The triple-occupancy site $(\mathbf0,\mathbf0)$ ($\nu=2$) is allowed by Bose
symmetry; it carries a deep state near $-2\mu$ and is removed by $\Pi_1$,
which is the origin of the Dirichlet condition in $h_{2+1B}$.

The bosonic symmetry of $\psi$ under the transposition $\tau_{12}$
gives $\psi(\mathbf0,\by)=+\varphi(\by)$ (with the plus sign, in~contrast to the fermionic case where the antisymmetry gives
$\psi(\mathbf0,\by)=-\varphi(\by)$). Repeating the computation of
Supplementary \S\ref{sec:S21}:
\begin{itemize}
\item diagonal: $(4+2\gamma)\varphi(\bx)$;
\item particle-$1$ hops: $(\eps_D-2)\varphi(\bx)$;
\item particle-$2$ hops: land on $(\bx,\be)$ which is never in
$\mathcal B_a\cup\mathcal B_b$ for $\bx\ne0$;
\item particle-$3$ hops: amplitude $-\gamma/2$, landing on
$(\mathbf0,-\bx)$ iff $\bx$ on the shell, giving
\[
-\tfrac\gamma2\,e^{i\bK\cdot\bx}\psi(\mathbf0,-\bx)
=-\tfrac\gamma2\,e^{i\bK\cdot\bx}\varphi(-\bx).
\]
\end{itemize}
Hence $h_1\restriction\Ran\Pi_1\cong h_{2+1B}(\bK;\gamma)
=2+2\gamma+\eps_D-\tfrac\gamma2S_\bK$, and using
$S_{\bK+\bpi}=-S_\bK$ one obtains
$h_{2+1B}(\bK;\gamma)=h_F(\bK+\bpi;\gamma)$. The counting criterion
(Theorem~\ref{thm:count}) with $C=-\tfrac\gamma2S_\bK$ gives
the zone indicator $\tfrac\gamma2 x_{\max}(\bK)>1$, hence
$\gamma_c^{2+1B}(\bK)=2/x_{\max}(\bK)=\gamma_c(\bK+\bpi)$.
\end{proof}

\begin{remark}[Numerical check]
Infinite-lattice evaluation (exact $M(\delta)$ with $b_0$ in elliptic
form, $A(\delta)$ by one-dimensional quadrature) confirms
$\delta_{2+1B}(\bK;\gamma)=\delta_F(\bK+\bpi;\gamma)$ at every grid
point: at~$\gamma=2$, $\delta_{2+1B}(\mathbf0;2)=\delta_F(\bpi;2)
=0.0354200$; at $\gamma=3$, $\delta_{2+1B}(\mathbf0;3)=\delta_F(\bpi;3)
=0.2528937$ (one~state), while
$\delta_{2+1B}(\bpi;3)=\delta_F(\mathbf0;3)=0.0519285$ is the doubly
degenerate $p$-wave pair. The~exact bosonic
$2+1$ zone coincides with the fermionic zone translated by $\bpi$.
\end{remark}

\end{document}